\documentclass[11pt]{article}
\usepackage{amssymb,amsmath,amsthm}
\usepackage{bbm}
\usepackage[noend]{algpseudocode}
\usepackage{graphicx}
\usepackage{verbatim}
\usepackage{url,xspace,hyperref}
\usepackage{caption}
\usepackage{subcaption}
\usepackage{multirow}
\usepackage{multicol}
\usepackage{latexsym}
\usepackage{amsmath,amssymb,enumerate}
\usepackage{xspace}
\usepackage[linesnumbered,ruled,vlined]{algorithm2e} 
\usepackage{float}
\usepackage{xcolor}
\usepackage{mathrsfs}
\usepackage{cleveref}
\usepackage{mathtools}
\usepackage{bm}
\usepackage{tikz}
\usepackage{caption}
\usepackage{multirow}
\usepackage[normalem]{ulem}
\usepackage{makecell}
\usepackage{enumitem}  
\usepackage[round]{natbib}
\usetikzlibrary{positioning,decorations.pathreplacing,quotes}
\usepackage{setspace}

\tikzset{tick/.style={draw, minimum width=0pt, minimum height=2pt, inner sep=0pt, label=below:$#1$},
    tick/.default={}}

\usepackage{fullpage}
\usepackage{geometry}
\usepackage{setspace}
\usepackage{changepage}
\usepackage{titlesec}

\hypersetup{
  colorlinks   = true, 
  urlcolor     = blue, 
  linkcolor    = black, 
  citecolor   = black 
}

\newtheorem{theorem}{Theorem}
\newtheorem*{theorem*}{Theorem}

\newtheorem{proposition}{Proposition}[section]
\newtheorem{claim}{Claim}
\newtheorem{lemma}[proposition]{Lemma}

\newtheorem{corollary}[proposition]{Corollary}
\newtheorem{remark}[proposition]{Remark}

\newcommand{\E}{\mathbb{E}} 
\newcommand{\OPT}{\mathtt{OPT}}
\newcommand{\ALG}{\mathtt{ALG}}

\title{\vspace{-0.5in} Online Contract Selection for Continual Coverage}
\author{}

 \author{
 Qinge Chi\thanks{Department of Computational Applied Mathematics and Operations Research, Rice University, USA.}\ 
 \thanks{Ken Kennedy Institute, Rice University, USA.} \and
 Sebastian Perez-Salazar\footnotemark[1]\ \footnotemark[2]
 }
 \date{}

\begin{document}

\maketitle

\begin{abstract}
Motivated by applications where a system must remain operational via continual procurement of contracts, we study two online contract selection problems under uncertain prices. At each time step, a price drawn from a known distribution is revealed online, and the decision-maker may initiate a contract of arbitrary duration, incurring a cost equal to the product of the price and the contract length; moreover, every time period must be covered by at least one active contract. We consider two models depending on how contracts cover time: a \emph{deferred model}, in which contracts are queued back-to-back, and a \emph{concurrent model}, in which contracts become active immediately and may overlap. In both settings, we seek online algorithms that minimize their competitive ratio, i.e., the ratio between the expected cost incurred by the online algorithm and the expected offline optimal cost when all prices are known in advance.

We first focus on the case where prices are independent and identically distributed (i.i.d.). For the deferred model, we characterize exactly the worst-case optimal competitive ratio, which is asymptotically $\zeta^* \approx 2.472$ as the time horizon grows. For the concurrent model, we prove a lower bound of $\zeta^*$ on the optimal competitive ratio and an asymptotic competitive ratio of at most $4.179$. These bounds improve upon the current lower bound of $2.148$ and upper bound of $6.052$ on the optimal competitive ratio. For both models, our algorithms are quantile-based that can be easily translated into practical threshold-based algorithms for any distribution. Our proofs follow from linear programs and duality arguments in quantile spaces. Lastly, we show that, in both models, no finite competitive ratio exists when the prices are still independent but not necessarily identically distributed, proving a striking division in the two price settings.
\end{abstract}



\section{Introduction}

Online selection has attracted increasing attention over the past decades, driven by a wide range of applications such as selling~\citep{correa2021posted,hajiaghayi2007automated}, hiring~\citep{epstein2024selection,perez2025robust}, and more general resource allocation problems~\citep{babaioff2007knapsack,kesselheim2014primal,karp1990optimal,feldman2009online}. Broadly, in online selection, a resource-constrained decision-maker observes online arriving values and must irrevocably decide whether to allocate resources to the currently observed value, trading off the value of the current arrival against the possibility of reserving resources for potentially better future ones.

A less explored setting arises in applications where systems must continuously procure costly contracts in order to remain operational, inducing a nontrivial online cost-minimization problem with a temporal covering constraint. 
For instance, a small business may rely on freelance services to maintain operations, and failing to maintain an active contract at any point in time may jeopardize the operations (see, e.g.,~\citep{Disser2019}). Another example is a person (or firm) who must maintain continuous insurance coverage (e.g., health, renter's, or equipment insurance), where a lapse in coverage may cause 
substantial financial losses in unforeseen events. 
These applications give rise to a new trade-off for the decision-maker: \emph{long contracts provide guaranteed coverage and more time to search for better future prices, but may lock in high prices for long periods of time, 
whereas shorter contracts reduce immediate costs but provide less time to explore future prices.} 

To study this new setting and to understand this new trade-off, we consider two online contract selection models in which a decision-maker observes online a sequence of independent and identically distributed (i.i.d.)\ nonnegative random values $X_1,\ldots,X_n$, drawn from a known distribution, where $X_i$ represents the cost per unit time of initiating a contract at time $i$. Upon observing a value, the decision-maker may initiate a contract of arbitrary duration at a cost equal to the product of the observed value and the contract duration. In an online manner, the decision-maker \emph{must ensure that every time period remains covered by at least one active contract}. Depending on how contracts are assigned to times, we consider the following models:

\medskip

In the \emph{Online Selection with Deferred Contracts} (OSDC) model, contracts are queued back-to-back. More precisely, upon observing value $X_i$ at time $i$, the decision-maker chooses an integer $T_i \geq 0$, covering times $\ell+1,\ldots,\ell+T_i$ at cost $T_i \cdot X_i$, where $\ell \geq i-1$ is the end time of the last active contract. If time $i$ is not yet covered by any contract, i.e., $\ell = i-1$, the decision-maker must choose $T_i \geq 1$ to ensure coverage of time $i$.

\medskip

In the \emph{Online Selection with Concurrent Contracts} (OSCC) model, introduced by~\cite{Disser2019}, newly selected contracts become active immediately and may overlap with previously active contracts. Upon observing value $X_i$, the decision-maker chooses an integer $T_i \geq 0$, which guarantees that times $i,i+1,\ldots,i+T_i-1$ are covered by at least one active contract, at cost $T_i \cdot X_i$. As before, if no contract is active at time $i$, then the decision-maker must choose $T_i \geq 1$.

\medskip



In both models, we seek online algorithms that minimize their corresponding overall expected cost while ensuring continuous contractual coverage at all times. We evaluate algorithms using the \emph{worst-case competitive ratio}, namely, the largest ratio possible among all possible distributions, between the expected cost incurred by the algorithm and the expected offline optimal cost. The offline solution knows all $X_1,\ldots,X_n$ upfront and contracts can be assigned optimally. The competitive ratio is always at least $1$ and quantifies the cost of lacking foreknowledge in the online setting.

Even though optimal algorithms for both models can be obtained through dynamic programming~(DP), natural formulations lead to optimal algorithms that depend on the current time, times covered by active contracts and the observed value. Unfortunately, in online selection problems, such algorithms have been proven notoriously hard to analyze within competitive analysis (see, e.g.,~\citep{Brustle2025splitting}). Instead, for OSCC,~\cite{Disser2019} give an algorithm based on the idea of selecting long contracts while simultaneously searching for better values using threshold strategies. By optimizing over contract lengths and the acceptance thresholds used during active contract periods, they obtain a competitive ratio of at most $6.052$. They also find a lower bound of $2.148$ on the competitive ratio. In contrast, for OSDC, no general guarantees are known.

\paragraph{Summary of our results.}
For OSDC, we fully characterize the optimal competitive ratio, showing that it is equal to a quantity that can be computed recursively through a finite system of equations. To obtain this result, we first provide a reduction to a sequence of cost-minimization single-selection problems, which also yields a characterization of the optimal algorithm. Then, we connect these problems via quantile algorithms and analyze them using linear programming~(LP) and duality. As a byproduct of our analysis, we also have a characterization of the worst-case instance. Asymptotically, we show that the competitive ratio approaches $\zeta^* \approx 2.472$, where $\zeta^*$ is the unique parameter that ensures a the existence of a solution to a nonlinear system (see~\eqref{eq:ode11}-\eqref{eq:ode22}).

For OSCC, we study a class of parameterized quantile-based algorithms. Using LP techniques, we derive a system that allows us to optimize over a small number of parameters and obtain an asymptotic competitive ratio of $4.179$, as well as improved guarantees for the special instance with uniformly distributed values. We also show that algorithms for OSCC can be translated to algorithms for OSDC by not worsening the cost incurred by the algorithm; thus, our tight result for OSDC allows us to conclude that the worst-case competitive ratio for OSCC is at least $\zeta^*$. Table~\ref{tab:table_1} provides a comparison with known results in OSCC.

\begin{table}[h!]\small
        \centering
        {\begin{tabular}{|c|c|c|c|c|}
        \hline
             Type of result       & \multicolumn{2}{c|} {Upper Bounds}  & Lower Bounds \\
            \hline
             &  Uniform distribution & Worst-case distribution & Worst-case distribution \\ \cline{2-4}
             Our results & $2.945$ (any $n$)  &   & $\zeta^* \approx 2.472$ (any $n$) \\
             & $2.908$ $(n\to \infty)$   & $4.179$ ($n\to \infty$)  & (competitive ratio for OSDC) \\
            \hline
            \cite{Disser2019}  & $2.965$ (any $n$) & $6.052$ (any $n$) & $2.148$ (numerical) \\
            \hline
        \end{tabular}}
        \caption{{Comparison with existing results on OSCC. For upper bounds, improvements are reflected by smaller values, vice versa for lower bounds.}}
        \label{tab:table_1}
    \end{table}

We extend the discussion to the case where values remain independent but not necessarily i.i.d. We prove that in both models is impossible to obtain a constant competitive ratio, showing a stark separation between the two value settings. 

\subsection{Problem Formulation} \label{subsec:intro1}

Let $\mathcal{F}$ be the family of CDFs with nonnegative support. The input to both OSDC and OSCC consists of a CDF~$F \in \mathcal{F}$ and an integer $n \geq 1$. Time is discrete and indexed by $i=1,\ldots,n$. Initially, at time $i=1$, no time is covered by a contract. In both problems, an online algorithm observes a sequence of i.i.d.\ random variables $X_1,\ldots,X_n \sim F$ one at a time, and outputs for each observation a contract length $T_i$ subject to the constraint that, \emph{upon observing a time, it must be covered by at least one contract,} i.e., $T_1 + T_2 + \cdots + T_i \geq i$ for all $i=1, \ldots, n$. 

For OSDC, at time $i$, given the observed $X_1,\ldots,X_i$ and the end time $\ell \geq i-1$ of the last active contract, an online algorithm $\text{ALG}^{\text{D}}$ outputs a contract length $T_i \geq 0$ on $X_i$, covering times $\ell+1,\ldots,\ell+T_i$, so the contracts admit a back-to-back structure. 

For OSCC, at time $i$, given the observed $X_1,\ldots,X_i$ and the set of active contracts, an online algorithm $\text{ALG}^{\text{C}}$ outputs a contract length $T_i \geq 0$ on $X_i$, which starts immediately at time $i$ and covers times $i,i+1,\ldots,i+t-1$. Note that multiple contracts can be concurrently active in OSCC. 

For both OSDC and OSCC, the expected cost of an online algorithm $\text{ALG}^{\mathcal{M}}$ is given by $$\ALG_n^{\mathcal{M}}(F) = \E \left[\sum_{i=1}^n T_i X_i\right],$$ for the model $\mathcal{M} \in \{C, D\}$, where $C$ represents the OSCC model and $D$ represents the OSDC model, and $\{T_i\}_{i=1}^n$ are contract lengths output by $\text{ALG}^{\mathcal{M}}$. 

In contrast, an offline algorithm on the same instance observes all realizations of $X_1,\ldots,X_n$ in advance. Note that both OSDC and OSCC models share the same offline optimal solution, which always selects the lowest among the arrived values, resulting in back-to-back contracts. The expected cost $\OPT_n(F)$ incurred by the offline optimal algorithm is given by 
\[
\OPT_n(F)= \E\!\left[ \sum_{i=1}^n \min\{X_1,\ldots,X_i\} \right].
\]

For an online algorithm $\text{ALG}^{\mathcal{M}}$ where $\mathcal{M} \in \{ C, D\}$, let $R_n^{\mathcal{M}}(\text{ALG}^{\mathcal{M}},F)={\ALG^{\mathcal{M}}_n(F)}/{\OPT_n(F)}$ be the competitive ratio of ALG$^\mathcal{M}$ in instances of size $n$ and CDF $F$. Define the worst-case competitive ratio in instances of size $n$ for $\text{ALG}^{\mathcal{M}}$ as 
\[
R_n^{\mathcal{M}}(\text{ALG}^{\mathcal{M}})=\sup_{F \in \mathcal{F}} R_n^{\mathcal{M}}(\text{ALG}^{\mathcal{M}},F), 
\]
and the optimal competitive ratio in instances of size $n$ as $\inf_{\text{ALG}^{{\mathcal{M}}}} R_n^{\mathcal{M}}(\text{ALG}^{\mathcal{M}})$. {To avoid excessive notation, we often omit the superscript in ALG and $\ALG_n$ when the model is clear from the context.}

\subsection{Overview of Our Contributions} \label{subsec:intro2}


{ In this work, we study the competitive ratios of OSDC and OSCC via quantile-based algorithms. Since the worst-case competitive ratio is defined as the supremum over all CDFs $F \in \mathcal{F}$, while the thresholds used by online algorithms naturally depend on the underlying distribution, we parameterize thresholds through their \emph{quantile} levels $q = F(\tau)$. This representation yields algorithm descriptions that are universal across all distributions and allows us to provide our competitive ratio analyses via LP techniques. We refer to the quantile levels used by the algorithms as \emph{quantile benchmarks}. Below, we summarize our results: a tight characterization of the competitive ratio for OSDC together with improved upper and lower bounds for OSCC.}

\paragraph{Tight competitive ratio for OSDC.} The main challenge in analyzing the optimal competitive ratio of contract selection problems is that an optimal algorithm's contract lengths depend on the current time, times covered by active contracts, and the observed value. For OSDC, we overcome this difficulty through the following reduction: the expected cost of the optimal algorithm for OSDC equals the sum of the expected costs of $n$ single-selection problems, where the $i$-th problem corresponds to the online single-selection problem of minimizing the expected value selected from $X_1,\ldots,X_i$, and the $i$-th single-selection problem is linked to the $(i-1)$-th problem via dynamic programming. This reduction yields a quantile-based characterization of the optimal algorithm, in which the quantile benchmark of contracting $X_i$ for different durations can be computed upfront using the DP for the single-selection problems (see Section~\ref{subsec:lb1} for details).

Furthermore, this sequence of single-selection problems allows us to characterize the optimal competitive ratio in instances of size $n$ as the optimal value of an infinite-dimensional LP (see~\ref{form:LPCP}). In this formulation, we have two types of variables: (i) variables capturing the expected value from each single-selection problem, and (ii) variables modeling the inverse cumulative distribution function $F^{-1}$. The objective is the sum of the expected costs of the single-selection problems, while the constraints encode (a) the dynamic programming recursion relating the $i$-th and $(i-1)$-th problems, and (b) the normalization that the expected offline optimal cost equals $1$. In the quantile space, constraints (a) and (b) are indeed linear. Using LP duality, we construct matching primal and dual solutions, which characterize the optimal quantiles used in each single-selection problem and the worst-case distribution. The optimal quantiles also serve as the qunatile benchmarks for the optimal algorithm in OSDC. 
Our matching primal-dual solutions also characterize the optimal competitive ratio in instances of length $n$ as a quantity $\zeta_n$ satisfying a finite system of equations.

To provide a value independent of $n$ in $\zeta_n$, we perform an asymptotic analysis of the system describing $\zeta_n$. As a result, we obtain the following system of ordinary differential equations, 
\begin{align}
&y' = \left( \frac{1}{\ln(y)} - \frac{2}{(\ln(y))^2} - \frac{2(1-y)}{y (\ln(y))^3}\right)^{-1} \cdot \left( - \frac{2(1-y)}{\ln(y)} - y - \frac{x}{\zeta^*}\right), \label{eq:ode11} \\
&y(0) = 0, \quad \lim_{x \uparrow 1} y(x) = 1, \label{eq:ode22}
\end{align}
where $y=y(x)$. We show that $\zeta_n\to \zeta^*$ where $\zeta^* \approx 2.472$. All details appear in Section~\ref{section:lowerbound}.

\paragraph{New competitive ratios for OSCC.} For OSCC, there is no known reduction like the one applied in OSDC. Thus, to establish our new bound on the competitive ratio for OSCC, we instead introduce a parameterized family of quantile-based algorithms for which we can derive explicit upper bounds on the expected cost. We leverage these bound using LP techniques to optimize over the parameters of the family and provide the asymptotic competitive ratio of $4.179$.

{Briefly, our family of quantile-based algorithms is parameterized by: (1) contract lengths, (2) corresponding quantile benchmarks, and (3) search durations. Each quantile benchmark determines the acceptance threshold for selecting a contract of a given length, while the associated search duration specifies how long the algorithm waits for a value below that benchmark. After selecting a contract, the algorithm searches for better future values using lower quantile benchmarks. If the corresponding search duration expires without selecting a new contract, the algorithm relaxes the quantile benchmark to a higher one and restarts the corresponding search duration. For every $n$ and CDF $F$, we derive explicit upper bounds on the expected cost incurred by the algorithm through a linear system of recursive inequalities. We refer for now to this meta-algorithm via ALG.}


For the analysis of the competitive ratio of ALG, we employ an LP in the quantile space. The construction is more involved than in the OSDC model, so we divide the analysis into two parts. We first consider the case in which the input distribution is uniform over an interval $[a,b]$, where $0\leq a<b$. Although simpler, this setting already illustrates how to construct dual solutions to the underlying LP, which we later extend to the general case. In the general case, we extend the LP approach to arbitrary distributions by adding a constraint that normalizes the expected optimal offline cost. However, this new constraint---which involves inverse of CDFs as variables---makes the resulting LP infinite-dimensional, introducing challenges that we address through approximations.

\paragraph{Ratios for uniform distributions.} Since the recursion that upper bounds the cost of ALG is a linear system, we can naturally formulate a maximization LP whose variables are the terms in the linear system. The resulting LP upper bounds the expected cost of ALG (see \ref{form:LPunif}). We apply a series of relaxations to this LP to reach a formulation whose dual, which is a minimization LP, has an explicit optimal solution characterized by the parameters defining ALG. By tuning these parameters, we show that $\lim_{n\to\infty} R_n^{\text{C}}\left(\text{ALG}, F_{U[a,b]} \right) \leq 2.908$, where $F_{U[a,b]}$ denotes the CDF of the uniform distribution over $[a,b]$. Furthermore, using the same parameters together with careful approximations for small values of $n$, we establish that $R_n^{\text{C}}\left(\text{ALG}, F_{U[a,b]} \right) \leq 2.945$ for every $n\geq 1$. The gap between the asymptotic bound and the bound for every $n$ arises from approximation terms that vanish as $n$ grows. This result improves upon the previously known upper bound of $2.965$ obtained by~\cite{Disser2019}. We provide the details in Section~\ref{section:unif}.

\paragraph{General competitive ratios.} A key property of the LP used in the uniform distribution case is that the terms involving the distribution, namely those depending on $F_{U[a,b]}$, are decoupled from the variables tracking the algorithm costs. Consequently, we can treat $F$ itself as a variable and optimize over all CDFs while imposing that the expected offline optimal value equals $1$. As in the OSDC model, this normalization is without loss of generality, since any instance can be rescaled to have expected offline optimal value of $1$ without affecting the competitive ratio.

After reformulating the problem in the quantile space to linearize the terms involving the inverse CDF, we obtain an infinite-dimensional LP that upper bounds $R_n^{\text{C}}(\mathrm{ALG})$ (see \ref{form:LPUB}). Applying relaxations similar to those used in the uniform case yields another infinite-dimensional LP whose dual solution can be partially matched with the solution derived for the uniform distribution case. The remaining unmatched term, stemming from the new primal constraint, reduces to a piecewise-convex function over $[0,1]$. This allows us to derive an upper bound on $R_n^{\text{C}}(\mathrm{ALG})$ given by the maximum value of this piecewise function, which depends only on the parameters defining ALG and $n$. { Unfortunately, analyzing this function for finite $n$ is difficult, as both the number of pieces and the structure of each piece depend on $n$ and the parameters of ALG.} This motivates an asymptotic analysis, in which we show that several terms in the function become negligible as $n$ grows. This yields a simpler optimization problem, from which we finally obtain $\limsup_{n\to\infty} R_n^{\mathrm{C}}(\mathrm{ALG}) \leq 4.172$ upon optimizing the parameters of ALG. We provide all the details in Section~\ref{section:general}. 

\paragraph{Improved lower bound for OSCC via OSDC.} { A transformation shows that an algorithm for OSCC can be adapted to an algorithm for OSDC such that the cost of the new algorithm in OSDC is no worse than the cost of the original algorithm in OSCC.} This implies that the competitive ratio for OSCC is at least the competitive ratio for OSDC. Our exact characterization for the competive ratio of OSDC then implies that the competitive ratio of OSCC is at least $\zeta^* \approx 2.472$. We provide the details in Section~\ref{section:noniid}.

\paragraph{Beyond i.i.d.\ values.} Finally, we extend our results to the setting where values are independent but drawn from possibly different distributions. In this setting, we show that for OSDC, no online algorithm achieves a bounded competitive ratio. This impossibility extends to OSCC, since the competitive ratio of OSDC remains a lower bound on that of OSCC in the non-identically distributed setting, as discussed earlier. This establishes a sharp separation between the i.i.d.\ and non-identical value models. 
We provide the construction in Section~\ref{section:noniid}. 
    
\subsection{Related Work}

\textbf{Prophet inequalities.} Prophet inequality problem, introduced by~\cite{Krengel1977} in the maximization single-selection setting, introduced the idea of comparing the expected value of an online decision-maker against a ``prophet'' that knows all information upfront. A tight competitive ratio of $1/2$ is possible with a simple single threshold algorithm~\citep{SamuelCahn1984}. This competitive ratio improves to $0.745$ under the i.i.d.\ observed values~\citep{HillKertz1982iid, Kertz1986iid, Abolhassani2017beatingiid, correa2021posted}. Prophet inequality problems have regained popularity due to their close connection to posted-price mechanism~\citep{hajiaghayi2007automated}, giving rise to a vast literature of variations such as multi-selection~\citep{Alaei2014}, matroid constraints~\citep{hajiaghayi2007automated, Chawla2010postedprice, Kleinberg2012matroid}, knapsack~\citep{Dutting2020stochastic, Jiang2025}, matching~\citep{Alaei2012onlineprophet}, and other combinatorial constraints~\citep{Rubinstein2017, Ehsani2024combinatorialauction}. Another line of interest is prophet inequalities with unknown distributions~\citep{Correa2019unknown, Rubinstein2019OptimalSP, Correa2024sample}. 

Closer to our work is the cost-minimization prophet inequality problem. In contrast to the maximization case, results on minimization are sparse. In the i.i.d.\ single-selection setting, ~\cite{Livanos2024minimization} characterized tight competitive ratio for a large family of distributions. However, in general, there is no finite competitive ratio~\citep{Esfandiari2017prophetsecretary, Livanos2024minimization}. Recently, \cite{Qin2024} studied a minimization problem in which a fraction of a divisible resource must be procured at random online prices.

\textbf{Problems over time.} Our problems also relate to online selection over time problems, a class that has emerged more recently. The maximization objective has been studied under both the adversarial setting~\citep{Fiat2015temp, Kesselheim2016improvedtemp} and the stochastic setting~\citep{Faw2022reusable, Feng2024reusable}, as well as in terms of utility functions~\citep{berzack2025dynamic} from a game theory perspective. The problem introduced by~\cite{abels2025prophet} can be viewed as the maximization counterpart of our problems. They showed that the optimal algorithm via DP admits a simple structure and presented a single threshold algorithm that is $0.396$-competitive. Later,~\cite{PS2024} presented an algorithm attaining the optimal competitive ratio of $0.618$. The minimization case has only been studied by~\cite{Disser2019}, our results for OSCC improve upon their current bounds on the competitive ratio. Their analysis relies on an Markov chain analysis, while we derive our results with an LP approach.  

\textbf{Mathematical optimization for online selection.} 
LP forms the backbone of our analysis. The idea of solving sequential decision problems using LP dates back to~\cite{Manne1960}. This has become a powerful tool in designing and analyzing online algorithms for problems such as online matching~\citep{Mehta2007onlinematching, Goyal2023matching}, online knapsack~\citep{babaioff2007knapsack, kesselheim2014primal}, stochastic probing~\citep{Gupta2013stochastic, epstein2024selection}, secretary problem~\citep{Buchbinder2014secretarylp, Chan2015secretaryLP, perez2025robust}, prophet inequalities~\citep{Alaei2012onlineprophet, Jiang2025}, and competition complexity~\citep{Brustle2024complexity}. Both~\cite{perez2025iid} and~\cite{Brustle2025splitting} presented frameworks of dual fitting for quantile-based algorithms of prophet inequalities (see also~\citep{Jiang2025}). In Section~\ref{section:lowerbound}, our dual solution follows a similar construction. However, in our case, we are to construct a matching primal solution to prove optimality.

\section{Preliminaries} \label{section:assumption}

{  We focus on quantile-based algorithms, where, given a threshold $\tau$ and a CDF $F$, the corresponding quantile benchmark $q$ is defined as $q = F(\tau)$. Our algorithms require that we can compute $\tau$ for each quantile benchmark $q$, while our analysis requires that $F^{-1}$ be strictly increasing and differentiable. In this section, we justify that we may assume without loss of generality that all CDFs considered in the remainder of the work satisfy these properties. } 

\begin{proposition} \label{prop:assumption}
Given a model, suppose there exists an algorithm $\text{ALG}$ such that it guarantees $\ALG_n(\hat{F}) \leq \beta \cdot \OPT_n(\hat{F})$ where $\beta \geq 1$, for any $\hat{F}$ strictly increasing and continuously differentiable. Then, the worst-case competitive ratio in instances of size $n$ holds $R_n(\text{ALG}) \leq \beta$.
\end{proposition}

The proof can be found in Appendix~\ref{app:assumption}. In short it follows by adding small noise to the observed values similar to~\cite{PS2024}. 
From now on, we thus assume all CDFs are strictly increasing and continuously differentiable.

We now present expressions for $\OPT_n(F)$ that will appear later in the work. The following lemma expresses $\OPT_n(F)$ in terms of the inverse $F^{-1}$ with a change of variable, which becomes handy when working with quantiles. The proof can also be found in Appendix~\ref{app:assumption}. 

\begin{lemma} \label{lm:opt}
For all positive integer $n$ and $F \in \mathcal{F}$, $\OPT_n(F) = \int_0^1 F^{-1}(u) \cdot \sum_{i=1}^n i (1-u)^{i-1} \, \mathrm{d}u$. With the change of variable $F^{-1}(u) = \int_0^u r(v) \, \mathrm{d}v$, $\OPT_n(F) = \int_0^1 r(v) \cdot \sum_{i=1}^n (1-v)^i \, \mathrm{d}v$. 
\end{lemma}
\begin{remark} \label{lm:uniformopt}
For a problem with $X_i \sim \text{Unif}[0, 1]$, $\OPT_n = \sum_{i=1}^n {1}/{(i+1)}$. Indeed, $F^{-1}(u) = u$ for $\text{Unif}[0, 1]$, thus $r(v) = 1$ and the result follows by integrating $\int_0^1 (1-v)^i \, \mathrm{d}v= {1}/{(i+1)}$ in Lemma~\ref{lm:opt}. 
\end{remark}

\section{Optimal Competitive Ratio for OSDC} \label{section:lowerbound}

In this section, we characterize the optimal algorithm $\text{ALG}^*$ for OSDC and its optimal competitive ratio on instances of size $n$. For simplicity, we denote ${R}_n^*=\inf_{\text{ALG}} R_n^D(\text{ALG})$.  

The first result of this section states that $R_n^*$ can be computed via a recursive system. For $n\geq 1$, let $P_n(t) = \sum_{i=1}^n t^i$ and $P_n'(t) = \sum_{i=1}^n i \cdot t^{i-1}$ its first derivative. Consider the following system, which is to hold for all $j = 1, \ldots, n-1$: 
\begin{align}
P_n'(1-\varepsilon_{j+1}) - P_n'(1-\varepsilon_{j}) = \varepsilon_{j+1}P_n'(1-\varepsilon_{j+1}) + P_n(1-\varepsilon_{j+1}) - \frac{j}{\zeta}. \label{eq:lbrecurrence}
\end{align}
\begin{theorem} \label{thm:lowerbound}
For any $n\geq 1$, $R_n^*=\zeta_n$, where $\zeta_n$ is the unique $\zeta$ such that the solution $\{\varepsilon_j\}_{j=1}^n$ to the recurrence relation \eqref{eq:lbrecurrence} satisfies $0 \leq \varepsilon_n < \varepsilon_{n-1} < \ldots < \varepsilon_1 = 1$. Furthermore, we have $\zeta_n = {n}/({\varepsilon_n P_n'(1-\varepsilon_n) + P_n(1-\varepsilon_n)})$. 
\end{theorem}
In addition, the following asymptotic result can be obtained from Theorem~\ref{thm:lowerbound}, in which the limit behavior of $\{\varepsilon_j\}_{j=1}^n$ as $n \rightarrow \infty$ is captured by a ordinary differential equation system. 
\begin{proposition} \label{prop:lowerbound}
We have $R_n^* \rightarrow \zeta^*$ as $n\to \infty$, where $\zeta^* \approx 2.472$ is the unique $\zeta$ such that the system \eqref{eq:ode11}-\eqref{eq:ode22} has a solution.  
\end{proposition}

The remainder of this section focuses on the proof of Theorem~\ref{thm:lowerbound}. We first describe the optimal algorithm $\text{ALG}^*$ for OSDC in Section~\ref{subsec:lb1}. Then, we proceed to formulate an LP for its competitive ratio and study a feasible solution of the dual LP which will yield the recurrence \eqref{eq:lbrecurrence} (Section~\ref{subsec:lb2}). In Section~\ref{subsec:lb3}, we provide a matching primal solution to the dual solution, which certifies the optimality $R_n^*=\zeta_n$. We defer the proof of Proposition~\ref{prop:lowerbound} to Appendix~\ref{app:asymptotic}.

\subsection{Optimal Algorithm for OSDC} \label{subsec:lb1}


We begin by describing the reduction that allows us to characterize $\text{ALG}^*$ via a sequence of single-selection problems. {The reduction idea is simple: for time $i$, it must be covered by one contract started at one of the times $1,\ldots, i$. Then, from the point of view of time $i$, we are solving a cost minimization single-selection problem in $X_1,\ldots,X_i$. Let $d_i$ be the optimal expected cost of a minimum cost single-selection problem over the sequence $X_1,\ldots,X_i$ where one value must be selected. Then, we have the following result:}
\begin{lemma} \label{lm:lbobj}
For any $n\geq 1$ and CDF $F$, the optimal algorithm $\text{ALG}^*$ for OSDC incurs an expected cost of $\ALG^*_n(F) = \sum_{i=1}^n d_i$. 
\end{lemma}
\begin{proof}
{  For $j=1,\ldots,n$, let $Y_j$ denote the value of the contract assigned by ALG$^*$ that covers time $j$. Then $\ALG_n^*(F)=\sum_{j=1}^n \mathbb{E}[Y_j]$. By the discussion above, $\mathbb{E}[Y_j]\geq d_j$, which implies that $\ALG_n^*(F)\geq \sum_{j=1}^n d_j$. For the other inequality, we will show below that the algorithms for the $n$ single-selection problems induce an algorithm ALG for OSDC such that $\ALG_n(F)=\sum_{j=1}^n d_j$. This completes the proof since $\ALG_n^*(F)\leq \ALG_n(F)$.}
\end{proof}
Each value $d_i$ can be computed via DP (see, e.g.,~\citep{ Jiang2025,perez2025iid} for the maximization case) via the following system:
\begin{align}
d_1 &= \int_0^1 F^{-1}(u) \, \mathrm{d}u, \label{eq:lbdp1} \\
d_i &= \min_{q \in [0, 1]} \left \{\int_{0}^{q} F^{-1}(u) \, \mathrm{d}u + (1 - q)d_{i-1} \right\}, \quad i = 2, \ldots, n. \label{eq:lbdp2}
\end{align}

{  A byproduct of Lemma~\ref{lm:lbobj} is a characterization of the optimal algorithm ALG$^*$ for OSDC {  by running the optimal algorithm for all $n$ single-selection problem simultaneously, which achieves $\ALG^*_n(F) = \sum_{i=1}^n d_j$}. Indeed, for each $i=2,\ldots,n$, let $q_i$ be the solution to the minimization problem in \eqref{eq:lbdp2}, and set $q_1=1$. Note that $0<q_n<\cdots<q_1=1$. Now, at time $i$, upon observing $X_i$, let $\ell\geq i-1$ denote the last time covered by a contract. If $\ell=n$, then, there is nothing to do and the algorithm terminates. Define $\ell+1 \leq j^* \leq n$ as the largest $\ell+1\leq j\leq n$ such that $F(X_i)\leq q_{j-i+1}$. If no such $j$ exists, then none of the corresponding single-selection problems for $j\geq \ell+1$ would accept $X_i$, and we set the contract length to $T_i=0$. Otherwise, we set $T_i=j^*-\ell$. Note that this ensure that we now have covered up to time $j^*$.}


Lemma~\ref{lm:lbobj} will allows us to formulate the competitive ratio of $\text{ALG}^*$ by linearizing~\eqref{eq:lbdp1}-\eqref{eq:lbdp2} and imposing one additional constraint.

\subsection{LP Formulation and Dual Solution} \label{subsec:lb2}
We reformulate $R_n^*$ as an LP using \eqref{eq:lbdp1}-\eqref{eq:lbdp2} by introducing variables $\hat{d}_i$ for $d_i$ and $h$ for $F^{-1}$, and imposing a normalizing constraint $\OPT_n(F) = 1$ with the expression in Lemma~\ref{lm:opt}. This yields the following linear program: 
\begin{alignat}{3}
& \sup_{\substack{\bm{d} \in \mathbb{R}_+^n \\ h: [0, 1] \rightarrow \mathbb{R}_+}} \quad && \sum_{i=1}^n \hat{d}_i \tag*{$(LP)_{\text{DC}}$}\label{form:LPCP} \\
& \text{s.t.} \quad && \hat{d}_1 \leq \int_0^1 h(u) \, \mathrm{d}u, \label{constraint:LBP1} \\ 
& && \hat{d}_i \leq \int_0^q h(u) \, \mathrm{d}u + (1-q)\cdot \hat{d}_{i-1}, \quad && i = 2, \ldots, n, \quad q \in [0, 1], \label{constraint:LBP2}\\
& && \int_0^1 h(u) \cdot P_n'(1-u) \, \mathrm{d}u = 1, \label{constraint:LBP3}\\
& && h(u) \leq h(u'), && \forall u \leq u', \quad u, u' \in [0, 1]. \label{constraint:LBP4} 
\end{alignat}
\begin{proposition} \label{prop:LBLP}
Let { $(\{\hat{d}_i^*\}_{i=1}^n,h)$} be the optimal solution to \ref{form:LPCP}. Then, $R_n^* = \sum_{i=1}^n \hat{d}_i^*$. 
\end{proposition}
\begin{proof}
The proof consists of the following two directions: \\
1. For any $F \in \mathcal{F}$ such that $\OPT_n(F) = 1$, given $\{d_i\}_{i=1}^n$ from \eqref{eq:lbdp1}-\eqref{eq:lbdp2}, we show $\sum_{i=1}^n \hat{d}_i^* \geq \sum_{i=1}^n d_i$; \\
2. For any feasible solution $(\{\hat{d}_i\}_{i=1}^n, h)$ to \ref{form:LPCP}, we show that $\{d_i\}_{i=1}^n$ computed using $F^{-1}(u)$ constructed from $h$ satisfies $\hat{d}_i \leq d_i$, which then implies $\sum_{i=1}^n \hat{d}_i^* \leq \sum_{i=1}^n d_i$.  \\
Combining the two gives $\sum_{i=1}^n \hat{d}_i^* = \sum_{i=1}^n d_i$, so by Lemma~\ref{lm:lbobj}, the optimal value to \ref{form:LPCP} gives the expected costs of the optimal algorithm for $F$ satisfying $\OPT_n(F) = 1$. Since scaling does not impact the ratio, we can normalize any $F$ to get $\OPT_n(F) = 1$. It follows that $R_n^* = \sum_{i=1}^n \hat{d}_i^*$. 

For the first direction, we first show that $\{d_i\}_{i=1}^n$ is a feasible solution to \ref{form:LPCP}. Take $h = F^{-1}$, then by Lemma \ref{lm:opt}, $h$ satisfies Constraint~\eqref{constraint:LBP3}. The feasibility of Constraint~\eqref{constraint:LBP1} follows immediately from \eqref{eq:lbdp1}. For $i \geq 2$, from \eqref{eq:lbdp2} we have 
\begin{align*}
d_i &= \min_{q \in [0, 1]} \left \{\int_{0}^{q} F^{-1}(u) \, \mathrm{d}u + (1 - q)d_{i-1} \right\} \leq \int_{0}^{q} h(u) \, \mathrm{d}u + (1 - q)d_{i-1}
\end{align*}
for all $q \in [0, 1]$, therefore, $d_i$ is feasible to \eqref{constraint:LBP2} for all $i = 2, \ldots, n$. The monotonicity and nonnegativity of $h$ follows from the definition of $F^{-1}$. Therefore, $\{d_i\}_{i=1}^n$ is feasible to \ref{form:LPCP}. The optimality of $\{\hat{d}_i^*\}_{i=1}^n$ implies $\sum_{i=1}^n \hat{d}_i^* \geq \sum_{i=1}^n d_i$. 

For the second direction, given any $(\{\hat{d}_i\}_{i=1}^n, h)$ feasible to \ref{form:LPCP}, with a standard argument we can assume $h$ is continuous and strictly increasing (see, e.g., \cite{perez2025iid}). As a result, $h$ defines an inverse CDF $F^{-1}$. Inductively, we show that $\hat{d}_i \leq d_i$ for $d_i$ computed using $F$ defined by $h$. Constraint~\eqref{constraint:LBP1} gives $\hat{d}_1 \leq \int_0^1 F^{-1}(u) \, \mathrm{d}u = d_1$ where the equality follows from \eqref{eq:lbdp1}. Now, suppose $\hat{d}_i \leq d_i$ for $i$ up to some $k < n$. For $i = k+1$, since Constraint~\eqref{constraint:LBP2} has to hold for all $q \in [0, 1]$, it holds for the $q$ that minimizes its RHS, therefore, 
\begin{align*}
\hat{d}_{k+1} \leq \min_{q \in [0, 1]} \left \{\int_{0}^{q} h(u) \, \mathrm{d}u + (1 - q)\hat{d}_{k} \right\} \leq \min_{q \in [0, 1]} \left \{\int_{0}^{q} F^{-1}(u) \, \mathrm{d}u + (1 - q)d_{k} \right\} = d_{k+1},
\end{align*}
where the second inequality replaced $h$ with $F^{-1}$ and used the induction hypothesis that $\hat{d}_k \leq d_k$. Therefore, we have $\hat{d}_i \leq d_i$ for all $i = 1, \ldots, n$. It follows that $\sum_{i=1}^n \hat{d}_i \leq \sum_{i=1}^n d_i$. 
\end{proof}

To solve~\ref{form:LPCP}, we use duality. Consider dual variables $\alpha_1$ for Constraint~\eqref{constraint:LBP1}, $\alpha_i(q)$ for Constraint~\eqref{constraint:LBP2} for $i = 2, \ldots, n$, $\zeta$ for Constraint~\eqref{constraint:LBP3}, and $\eta(u)$ for Constraint~\eqref{constraint:LBP4}, which results in the following LP:
\begin{alignat}{3}
& \inf_{\substack{\alpha_1 \in \mathbb{R}_+ \\ \alpha_i: [0, 1] \rightarrow \mathbb{R}_+ \\ \eta \in \mathcal{C}^1 [0, 1]}} \quad && \zeta \tag*{$(DLP)_{\text{DC}}$}\label{form:DLPCP} \\
& \text{s.t.} \quad && \alpha_1 \geq \int_0^1 (1-q) \alpha_{2}(q) \, \mathrm{d}q + 1, \label{constraint:LBD1} \\ 
& && \int_0^1 \alpha_{i}(q) \, \mathrm{d}q \geq \int_0^1 (1-q) \alpha_{i+1}(q) \, \mathrm{d}q + 1, \quad && i = 2, \ldots, n-1, \label{constraint:LBD2}\\
& && \int_0^1 \alpha_{n}(q) \, \mathrm{d}q \geq 1, \label{constraint:LBD3}\\
& && \zeta \cdot P_n'(1-u) + \frac{d \eta(u)}{du} \geq \alpha_1 + \sum_{i=2}^n \int_u^1 \alpha_{i}(q) \, \mathrm{d}q, \quad && u \in [0, 1], \label{constraint:LBD4} \\
& && \eta(1) = 0. \label{constraint:LBD5} 
\end{alignat}
The following result shows that the objective of \ref{form:DLPCP} upper bounds that of \ref{form:LPCP}, i.e., weak duality. The proof relies on manipulating the constraints, and we defer the details to Appendix~\ref{app:lowerbound}. 
\begin{lemma} \label{lm:lbweakdual}
\ref{form:DLPCP} and \ref{form:LPCP} admit weak duality, namely, for any feasible solutions $\zeta$ and $\{\hat{d}_i\}_{i=1}^n$ to \ref{form:DLPCP} and \ref{form:LPCP} respectively, $\zeta \geq \sum_{i=1}^n \hat{d}_i$. 
\end{lemma}
We now construct a feasible solution for \ref{form:DLPCP}. 
We focus on the feasible space with $\eta = 0$, and propose the following for $\alpha$: suppose there exists a sequence $0 \leq \varepsilon_n < \varepsilon_{n-1} < \ldots < \varepsilon_1 = 1$, let $\alpha_1 = \zeta$, and $\alpha_{i}(q) = \zeta \cdot P_n''(1-q) \cdot \mathbbm{1}_{[\varepsilon_i, \varepsilon_{i-1})} (q)$. Intuitively, $\varepsilon_i$ corresponds to the quantile benchmark for the $i$-th single-selection problem in the reduction of OSDC in Section~\ref{subsec:lb1}. It can be checked by direct substitution that the proposed solution satisfies Constraint~\eqref{constraint:LBD4} with $\frac{d \eta(u)}{du} = 0$. With this solution, we derive the recurrence relation \eqref{eq:lbrecurrence} that links consecutive $\varepsilon_j$ and $\zeta$ in the next result.
\begin{lemma} \label{lm:lbrecurrence}
Suppose $\alpha_1 = \zeta$, $\alpha_{i}(q) = \zeta \cdot P_n''(1-q) \cdot \mathbbm{1}_{[\varepsilon_i, \varepsilon_{i-1})} (q)$ for $i = 2, \ldots, n$ satisfy Constraint~\eqref{constraint:LBD1}-\eqref{constraint:LBD3} with equality. Then, given $\zeta$, $\{\varepsilon_j\}_{j=1}^n$ satisfies \eqref{eq:lbrecurrence}. 
\end{lemma}
\begin{proof}
Suppose the proposed solution tightens \eqref{constraint:LBD1}-\eqref{constraint:LBD3}. Summing up Constraint~\eqref{constraint:LBD1}-\eqref{constraint:LBD2} across $i = 1, \ldots, j$ with the proposed solution, we obtain
\begin{align*}
\text{LHS} = \alpha_1 + \sum_{i=2}^j \int_0^1 \alpha_i(q) \, \mathrm{d}q &= \zeta + \zeta \int_{\varepsilon_j}^{\varepsilon_1} P_n''(1-q) \, \mathrm{d}q = \zeta + \zeta (P_n'(1-\varepsilon_j) - 1) = \zeta \cdot P_n'(1-\varepsilon_j),
\end{align*}
where the second last equality used $\varepsilon_1 = 1$ and evaluated $P_n'(0) = 1$. For the RHS, we have 
\begin{align*}
\text{RHS} = \sum_{i=1}^j \left(\int_0^1 (1-q) \alpha_{i+1}(q) \, \mathrm{d}q + 1 \right) &= \zeta \int_{\varepsilon_{j+1}}^{\varepsilon_1} (1-q) P_n''(1-q) \, \mathrm{d}q + j \\*
&= \zeta(P_n'(1-\varepsilon_{j+1}) - \varepsilon_{j+1}P_n'(1-\varepsilon_{j+1}) - P_n(1-\varepsilon_{j+1})) + j, 
\end{align*}
where the last equality evaluated the integral, which involves performing integration by parts on $\int_{\varepsilon_{j+1}}^{\varepsilon_1} -q P_n''(1-q) \, \mathrm{d}q$, and simplified using $\varepsilon_1 = 1$ and $P_n(0) = 0$. Equating both sides and rearranging, we obtain the required recurrence relation. 
\end{proof}
Given $\zeta$, the sequence $\{\varepsilon_j\}_{j=1}^n$ can be iteratively solved using recurrence \eqref{eq:lbrecurrence}. However, we require $\{\varepsilon_j\}_{j=1}^n$ to define a valid partition on $[\varepsilon_n, 1] \subseteq [0, 1]$ to ensure the feasibility of $\{\alpha_i\}_{i=1}^n$ constructed earlier, therefore $\varepsilon_j$ must be strictly decreasing in $j$, with $\varepsilon_n \geq 0$ and $\varepsilon_1 = 1$. As such, we need the next result, which shows that $\zeta_n$ defined in Theorem~\ref{thm:lowerbound} is unique and allows $\{\varepsilon_j\}_{j=1}^n$ to have the desired properties. Later in Section~\ref{subsec:lb3}, we will see how this particular definition of $\zeta_n$ gives us the matching primal objective. Moreover, $\zeta_n > 1$, which shows that $\zeta_n$ rightfully defines the competitive ratio for a minimization problem. 
\begin{proposition} \label{prop:lbgammaunique}
For every positive integer $n > 2$, there exists a unique $\zeta_n > 1$ such that the system formed by \eqref{eq:lbrecurrence} with the initial condition $\varepsilon_1 = 1$ has a solution $\{\varepsilon_j\}_{j=1}^n$ that satisfies $0 \leq \varepsilon_n < \ldots < \varepsilon_2 < \varepsilon_1 = 1$ and $\zeta_n = {n}/{(\varepsilon_n P_n'(1-\varepsilon_n) + P_n(1-\varepsilon_n))}$. 
\end{proposition}
To prove Proposition \ref{prop:lbgammaunique}, we need the following intermediate result. 
\begin{lemma} \label{lm:lbgammaunique}
For every positive integer $n > 2$, there exists a unique $\zeta > 1$ such that the system formed by \eqref{eq:lbrecurrence} satisfies $0 = \varepsilon_n < \ldots < \varepsilon_2 < \varepsilon_1 = 1$. 
\end{lemma}
Lemma \ref{lm:lbgammaunique} guarantees the existence of a unique $\zeta$ that yields the desirable $\{\varepsilon_j\}_{j=1}^n$ with $\varepsilon_n = 0$. The proof relies on the monotonicity of each $\varepsilon_j$ in $\zeta$, which can be proved inductively. We defer the complete proof to Appendix~\ref{app:lowerbound}, and proceed directly to prove Proposition \ref{prop:lbgammaunique}. 
\begin{proof}[Proof of Proposition \ref{prop:lbgammaunique}]
We prove this by showing how to find $\zeta_n$ using the $\zeta$ in Lemma \ref{lm:lbgammaunique}. First observe that we can solve for $\varepsilon_j$ iteratively using \eqref{eq:lbrecurrence} starting from $\varepsilon_1 = 1$. It can be shown using induction that each $\varepsilon_j$ obtained from this procedure is differentiable and strictly decreasing in $\zeta$. The proof of Lemma \ref{lm:lbgammaunique} demonstrates the same argument, so we skip the details here for brevity. 

Now suppose we have solved for $0 = \varepsilon_n < \ldots < \varepsilon_2 < \varepsilon_1 = 1$ given the $\zeta$ in Lemma \ref{lm:lbgammaunique}. Note that ${n}/{(\varepsilon P_n'(1-\varepsilon) + P_n(1-\varepsilon))}$ is increasing in nonnegative $\varepsilon$ and evaluates to $1$ at $\varepsilon = 0$, so $\zeta_n > 1$ only if $\varepsilon_n > 0$. But we also have in Lemma \ref{lm:lbgammaunique} that $\zeta > 1$ is necessary for $\{\varepsilon_j\}_{j=1}^n$ to have the desired properties. Recall that $\varepsilon_n$ is decreasing in $\zeta$ when the end $\varepsilon_1 = 1$ is fixed. Therefore, to obtain $\zeta_n$, we can decrease $\zeta$ until its value coincides with ${n}/{(\varepsilon P_n'(1-\varepsilon) + P_n(1-\varepsilon))}$, which also results in $\varepsilon_n > 0$. Lemma \ref{lm:lbgammaunique} and continuity ensures that $\zeta_n$ constructed here preserves the properties of $\{\varepsilon_j\}_{j=1}^n$. The uniqueness of $\zeta_n$ follows from the strict monotonicity of $\varepsilon_n$ in $\zeta$. 
\end{proof}
This concludes the dual analysis part in the proof of Theorem~\ref{thm:lowerbound}.

\subsection{A Matching Primal Solution} \label{subsec:lb3}
So far, we have only established weak duality. To demonstrate that $\zeta_n$ indeed gives the optimal competitive ratio to OSDC, we show that for any $n$, we can construct a primal feasible solution using $\{\varepsilon_j\}_{j=1}^n$ in the dual solution proposed earlier, such that the following result holds:
\begin{proposition} \label{prop:lbstrongdual}
For every integer $n > 2$ and $\zeta_n$ given in Proposition \ref{prop:lbgammaunique}, there exist $\{d_i\}_{i=1}^n$ feasible to \ref{form:LPCP} such that $\sum_{i=1}^n d_i = \zeta_n$. 
\end{proposition} 
This shows that the two LPs have matching objective values with our proposed solutions. {  A bonus from this primal-dual fitting approach is the worst case distribution for OSDC from our primal solution construction, where we define $h(u) = \sum_{i=1}^{n-1} d_i \cdot \mathbbm{1}_{[\varepsilon_{i+1}, \varepsilon_{i})}(u) + (d_1 - d_2) \cdot \delta_{\{1\}}$ with $\delta_{\{1\}}$ being the Dirac delta function centered at $1$, and $\{d_i\}_{i=1}^n$ computed using the dual solution. This can be smoothed to recover $F^{-1}$ for $F \in \mathcal{F}$ by considering $F^{-1}(u) = h(u) + \xi_1 \cdot u$ on $[0, 1-\xi_2]$ and $F^{-1}(u) = {(d_1 - d_2)}/{\xi_2} + \xi_1 \cdot u$ on $(1-\xi_2, 1]$ where $\xi_1$ and $\xi_2$ can be made arbitrarily small.}

Before proceeding to the proof of Proposition \ref{prop:lbstrongdual}, we first use it to prove Theorem~\ref{thm:lowerbound}.
\begin{proof}[Proof of Theorem~\ref{thm:lowerbound}]
Given $\zeta_n$ as defined in Theorem~\ref{thm:lowerbound}, Proposition \ref{prop:lbgammaunique} ensures that the associated $\{\varepsilon_j\}_{j=1}^n$ defines a valid partition of $[\varepsilon_n, 1] \subseteq [0, 1]$. By Lemma \ref{lm:lbrecurrence}, with our construction of $\alpha_1$ and $\{\alpha_i\}_{i=2}^n$, $(\alpha_1, \{\alpha_i\}_{i=2}^n, \zeta_n)$ is a feasible solution to \ref{form:DLPCP}. On the other hand, Lemma \ref{lm:lbprimalsol} later shows that $\{\varepsilon_j\}_{j=1}^n$ obtained using $\zeta_n$ defines a corresponding solution $(\{d_i\}_{i=1}^n, h)$ that is feasible to \ref{form:LPCP}. Let $\{d_i^*\}_{i=1}^n$ denote the optimal solution to \ref{form:LPCP}. Since the two sets of solutions are dual and primal feasible respectively, we now have
\begin{align*}
\sum_{i=1}^n d_i^* \leq \zeta_n = \sum_{i=1}^n d_i \leq \sum_{i=1}^n d_i^*, 
\end{align*}
where the first inequality follows from weak duality (Lemma \ref{lm:lbweakdual}), the equality follows from Proposition \ref{prop:lbstrongdual}, and the last inequality is due to the optimality of $z^*$. This forces the equality $\sum_{i=1}^n d_i = z^*$ to hold, implying the optimality of $\sum_{i=1}^n d_i$. It follows from Proposition \ref{prop:LBLP} that $\zeta_n = \sum_{i=1}^n d_i$ is the competitive ratio of the optimal algorithm to OSDC, i.e., $R_n^* = \zeta_n$.  
\end{proof}

We now show how to construct the said primal solution, followed by showing its feasibility in Lemma \ref{lm:lbprimalsol}. With this solution, we prove Proposition \ref{prop:lbstrongdual} at the end of this section.   

For notational simplicity, in what follows, let $\Delta_i = d_{i-1} - d_i$. The procedures to construct a solution for \ref{form:LPCP} given $\{\varepsilon_j\}_{j=1}^n$ from the dual solution are as follows: 
\begin{enumerate}
    \item Compute $S_1 = \sum_{i=2}^{n-1} P_n(1-\varepsilon_i) \prod_{k=2}^{i-1} (1-\varepsilon_k)$ and $S_2 = \prod_{k=2}^{n-1} (1-\varepsilon_k)$ ;
    \item Solve the following simultaneous equations to obtain $\Delta_2$ and $d_{n-1}$;
    \begin{align}
    (1 + S_1)\Delta_2 + P_n(1-\varepsilon_n) \cdot d_{n-1} &= 1, \label{eq:lbsimuleq1} \\
    S_2 \Delta_2 - \varepsilon_n \cdot d_{n-1} &= 0. \label{eq:lbsimuleq2}
    \end{align}
    \item Compute $d_1 = d_{n-1} + \Delta_2 \sum_{i=1}^{n-2}\prod_{k=2}^i (1 - \varepsilon_k)$ and $d_2 = d_1 - \Delta_2$;
    \item Iteratively compute $d_{i+1} = d_i - \Delta_2 \prod_{k=2}^i (1 - \varepsilon_k)$ to construct a non-increasing sequence $\{d_i\}_{i=1}^{n}$, which satisfies ${\Delta_{i+1}}/{\Delta_i} = 1 - \varepsilon_i$ for $i = 2, \ldots, n-1$;
    \item Set $h(u) = \sum_{i=1}^{n-1} d_i \cdot \mathbbm{1}_{[\varepsilon_{i+1}, \varepsilon_{i})}(u) + \Delta_2 \cdot \delta_{\{1\}}$, where $\delta_{\{1\}}$ is the Dirac delta centered at $1$. 
\end{enumerate}
Observe that since $d_i$ is non-increasing in $i$, $h$ defined this way is non-decreasing. In addition, observe that the dual Constraint~\eqref{constraint:LBD4} corresponding to $h$ is not tight on $[0, \varepsilon_n)$ at the dual solution in Section~\ref{subsec:lb2}, while $h(u) = 0$ for $u \in [0, \varepsilon_n)$, so the proposed $h$ respects complementary slackness. The following result shows that our proposed solution $(\{d_i\}_{i=1}^n, h)$ tightens Constraint~\eqref{constraint:LBP1}-\eqref{constraint:LBP3}. The proof mainly consists of routine calculations, therefore, we defer it to Appendix~\ref{app:lowerbound}.
\begin{lemma} \label{lm:lbprimalsol}
Given $\{d_i\}_{i=1}^n$ as constructed above, the solution $h(u) = \sum_{i=1}^{n-1} d_i \cdot \mathbbm{1}_{[\varepsilon_{i+1}, \varepsilon_{i})}(u) + \Delta_2 \cdot \delta_{\{1\}}$ and $\hat{d}_i = d_i$ for $i = 1, \ldots, n$ is feasible to \ref{form:LPCP}, and satisfies \eqref{constraint:LBP1}-\eqref{constraint:LBP2} with equality.  
\end{lemma}
We note that the proof of Lemma \ref{lm:lbprimalsol} also reveals an important identity $d_n = (1-\varepsilon_n) d_{n-1}$. With this, we now prove Proposition \ref{prop:lbstrongdual}.
\begin{proof}[Proof of Proposition \ref{prop:lbstrongdual}]
From the proof of Lemma \ref{lm:lbprimalsol}, Constraint~\eqref{constraint:LBP3} becomes 
\begin{align}
\sum_{i=2}^{n-1} \Delta_i P_n(1-\varepsilon_{i}) + d_{n-1} \cdot P_n(1 - \varepsilon_n) + \Delta_2 = 1 \label{eq:lbnormal}
\end{align}
when evaluated using our proposed $h$. To link this with $\zeta$, we rearrange \eqref{eq:lbrecurrence} to obtain an expression for $\sum_{i=2}^n \Delta_i P_n(1-\varepsilon_{i})$ in terms of $\zeta$. Multiplying the recurrence by $\Delta_{j+1}$ and simplifying using $\Delta_{j+1} (1-\varepsilon_{j+1}) = \Delta_{j+2}$, we obtain 
\begin{align*}
\Delta_{j+2} P_n'(1-\varepsilon_{j+1}) - \Delta_{j+1} P_n'(1-\varepsilon_{j}) - \Delta_{j+1} P_n(1-\varepsilon_{j+1}) = -\Delta_{j+1} \cdot \frac{j}{\zeta},
\end{align*}
Summing the LHS from $j = 1$ to $n-2$, we obtain a telescopic sum that after re-indexing gives us:
\begin{align*}
\text{sum of LHS} &= \Delta_{n}P_n'(1-\varepsilon_{n-1}) - \Delta_2 - \sum_{j=2}^{n-1} \Delta_{j}P_n(1-\varepsilon_{j}).
\end{align*}
We also sum the RHS above and obtain 
\begin{align*}
\text{sum of RHS} &= -\frac{1}{\zeta} \sum_{j=1}^{n-2} j \cdot \Delta_{j+1} = -\frac{1}{\zeta} \left(\sum_{j=1}^{n} d_j - (n-1)d_{n-1} - d_n \right) = -\frac{1}{\zeta} \left(\sum_{j=1}^{n} d_j - (n - \varepsilon_n)d_{n-1} \right),
\end{align*}
where it is straight forward to verify the second equality by expanding the sum $\sum_{j=1}^{n-2} j \cdot \Delta_{j+1}$ then adding and subtracting $d_n$, and the last equality applied $d_n = (1-\varepsilon_n)d_{n-1}$. Since sum of LHS $=$ sum of RHS,
\begin{align*}
\sum_{j=2}^{n-1} \Delta_{j}P_n(1-\varepsilon_{j}) = \frac{1}{\zeta} \left(\sum_{j=1}^n d_j - (n - \varepsilon_n)d_{n-1} \right) + \Delta_{n}P_n'(1-\varepsilon_{n-1}) - \Delta_2.
\end{align*}
Renaming the index to $i$ in the above expression and substituting it into \eqref{eq:lbnormal}, 
\begin{align}
\frac{1}{\zeta} \cdot \sum_{j=1}^n d_j - \frac{n-\varepsilon_n}{\zeta} \cdot d_{n-1} + \Delta_{n}P_n'(1-\varepsilon_{n-1}) + d_{n-1} \cdot P_n(1 - \varepsilon_n) = 1. \label{eq:lbfinal}
\end{align}
The identity $d_n = (1-\varepsilon_n)d_{n-1}$ gives $\Delta_n = \varepsilon_n \cdot d_{n-1}$, so the last three terms in~\eqref{eq:lbfinal} becomes $ \left(-{(n-\varepsilon_n)}/{\zeta} + \varepsilon_n P_n'(1-\varepsilon_{n-1}) + P_n(1 - \varepsilon_n) \right) d_{n-1}$. Now using \eqref{eq:lbrecurrence}, we can rewrite the expression for $\zeta_n$ as $\zeta_n = {1}/{(P_n'(1-\varepsilon_n) - P_n'(1-\varepsilon_{n-1}))}$. Substituting in this expression for $\zeta_n$, 
\begin{align*}
&\quad -\frac{n-\varepsilon_n}{\zeta_n} + \varepsilon_n P_n'(1-\varepsilon_{n-1}) + P_n(1 - \varepsilon_n) \\
&= (\varepsilon_n - n)(P_n'(1-\varepsilon_n) - P_n'(1-\varepsilon_{n-1})) + \varepsilon_n P_n'(1-\varepsilon_{n-1}) + P_n(1 - \varepsilon_n) \\
&= \varepsilon_n P_n'(1-\varepsilon_n) + P_n(1 - \varepsilon_n) - \frac{n-1}{\zeta_n} - (P_n'(1-\varepsilon_n) - P_n'(1-\varepsilon_{n-1})) = 0,
\end{align*}
where the last line is obtained by adding and subtracting $P_n'(1-\varepsilon_n) - P_n'(1-\varepsilon_{n-1})$ and applying the definition of $\zeta_n$, and the equality to $0$ follows from \eqref{eq:lbrecurrence} at $j = n-1$. Therefore, with $\zeta = \zeta_n$, the last three terms in~\eqref{eq:lbfinal} evaluates to $0$, and we are left with ${1}/{\zeta_n} \cdot \sum_{j=1}^n d_j = 1$, thus~$\sum_{j=1}^n d_j = \zeta_n$. 
\end{proof}

\section{Algorithm for OSCC} \label{section:alg}

We now start the study of the OSCC model. For simplicity, we drop the superscript and denote an online algorithm on OSCC simply as $\text{ALG}$. Similar to OSDC, the main challenge in analyzing the competitive ratio of OSCC is that the optimal algorithm, in principle, outputs contract lengths  based on the distribution, the observed value, the current time, and times covered by active contracts. Potentially, a reduction similar to the single-selection problems for OSDC may exist for OSCC, which would greatly simplify the analysis; however, the existence of such a reduction remains open. Instead, in this section, we propose a family of algorithms and show how to upper bound its performance. We formally present the meta-algorithm in Algorithm~\ref{alg}. 

\begin{algorithm} [H]
\caption{Meta-Algorithm} \label{alg}
Input $F \in \mathcal{F}$, $1 \geq q_0 > \ldots > q_j \geq 0$, $0 \leq s_0 \leq \ldots \leq s_j$, $0 \leq t_0 \leq \ldots \leq t_j$ with $t_j \geq n$ and {$\sum_{l=0}^{k+1} s_l \leq t_k$ for $k=0,\ldots,{j-1}$} \;
$k = 0, c = s_0$ \;
\For{$i = 1, \ldots, n$}{
$c = c - 1$ \;
\uIf{$F(X_i) \leq q_k$}{
Find $l$ such that $q_{l+1} < F(X_i) \leq q_l$ \;
{Contract $X_i$} for $\min\{t_l, n-i+1\}$ time steps \;
$k = l+1$ \;
$c = s_k$ \;
}
\ElseIf{$c = 0$}{
\uIf{$k = 0$}{
{Contract $X_i$} for $1$ time step \;
$c = s_0$ \;
}
\Else{
$k = k - 1$ \;
$c = s_k$.
}
}
}
\end{algorithm}

{  The meta-algorithm is inspired by the structured exhibited by the optimal algorithm obtained from the dynamic program: At time $i$, if $\ell \geq i-1$ is the last time covered by a contract, then there exists thresholds $0=\tau_{i,\ell,n-i} \leq \tau_{i,\ell,n-i+1}\leq \cdots \leq \tau_{i,\ell,1}\leq +\infty$, with $\tau_{\tau_{i,\ell,1}}=+\infty$ if $\ell=i-1$. Then, upon observing $X_i$, the optimal algorithm outputs $T_i=j$ if $X_i\in ( \tau_{i,\ell,j-1},\tau_{i,\ell,j}]$. Our approach instead, relaxes the thresholds, the possible length of contracts and introduces search durations. Besides the CDF $F \in \mathcal{F}$, the algorithm takes in three sets of inputs, quantiles $\{q_0, \ldots, q_j\}$, search durations $\{s_0, \ldots, s_j\}$ and contract durations $\{t_0, \ldots, t_j\}$, where the quantiles partition $[0, 1]$ and durations are positive integers. Upon observing $X_i$ such that $F(X_i) \in \left(q_{l+1}, q_l\right]$, Algorithm~\ref{alg} contracts $X_i$ for $t_l$ times, and spends the next $s_{l+1}$ times searching for a value better than $q_{l+1}$. Note that the thresholds used by the algorithm correspond to $\tau_j=F^{-1}(q_j)$. Algorithm~\ref{alg} restricts the contract lengths to a set of prespecified durations $t_0 \leq \ldots \leq t_j$ and fixes the benchmark quantile for committing to each duration upfront instead of being fully adaptive at each time. Given past decisions, Algorithm~\ref{alg} searches for better values to secure longer contracts, thereby mimicking the structure of the optimal algorithm. The search durations determine how long the algorithm waits for a better value and compensate for the loss of adaptivity induced by the prespecified durations and quantiles, while maintaining feasibility. Specifically, the durations $s_k$ and $t_k$ are required to satisfy $\sum_{l=0}^{k+1}s_l\leq t_k$, so that after committing to a contract associated with $q_k$, the remaining covered periods exceed the total future search durations associated with $q_{k+1},\ldots,q_0$. This guarantees that every time period is covered by at least one active contract while running Algorithm~\ref{alg}.}

Algorithm~\ref{alg} terminates exactly when all times are covered, so the first obstacle to computing its expected cost is characterizing the exact durations of contracts for every $n$. In the remainder of this subsection, we show a workaround by presenting a system that computes an upper bound on the expected costs without knowing the exact contract durations. We say that the algorithm is \emph{in state $k$} when it is currently using the quantile $q_k$, and let state $j+1$ be the terminal state which the algorithm transits to after initiating a contract in state $j$. Denote Algorithm~\ref{alg} by $\text{ALG}$. We consider an algorithm $\text{ALG}'$ that has the same inputs and procedures as $\text{ALG}$, but terminates only when state $j+1$ is reached, and let $\ALG'(F)$ be the expected cost incurred by $\text{ALG}'$. Since for any $n$, the inputs satisfy $t_j \geq n$, $\text{ALG}'$ is guaranteed to cover all times when reaching state $j+1$. As such, $\ALG_n(F) \leq \ALG'(F)$. Let $d_k$ be the expected cost incurred by $\text{ALG}'$ in state $k$ for $k = 0, \ldots, j+1$. We will find an explicit that allows us to to compute $d_k$ and such that $\ALG_n(F)\leq \ALG'(F)=d_0$.

Let $C(k)$ be the expected cost of a contract initiated in state $k$ by $\text{ALG}'$. In what follows, we abuse notation and take $q_{j+1} = 0$. The following lemma provides a closed-form expression for $C(k)$.

\begin{lemma} \label{lm:hirecost}
Given inputs $(q_0, \ldots, q_j)$, $(t_0, \ldots, t_j)$ and CDF $F$, define the change of variable $F^{-1}(u) = \int_0^u r(v) \, \mathrm{d}v$. A contract initiated in state $k$ by $\text{ALG}'$ has an expected cost of $C(k) =  {1}/{q_k} \cdot \sum_{l=k}^j \int_0^{q_l} r(v) \cdot \min\{q_l - q_{l+1}, q_l - v\} \, \mathrm{d}v \cdot t_l$. 
\end{lemma}
\begin{proof}
Following the algorithm, the expected cost of a contract in state $k$ can be computed by 
\begin{align*}
C(k) &= \sum_{l=k}^{j} \mathbb{P}\left[ F(X_i) \in \left( q_{l+1}, q_l\right] \,\middle|\, F(X_i) \leq q_k\right] \cdot \mathbb{E}\left[ X_i \,\middle|\, F(X_i) \in \left( q_{l+1}, q_l\right]\right] \cdot t_l.
\end{align*}
The result follows from $\mathbb{P}\left[ F(X_i) \in \left( q_{l+1}, q_l\right] \,\middle|\, F(X_i) \leq q_k\right] = {(q_l - q_{l+1})}/{q_k}$ and the expression of the expectation given in Claim \ref{claim:kexpcost}, for which we defer the proof to Appendix~\ref{appsub:alg}.

\begin{claim} \label{claim:kexpcost}
For any $F \in \mathcal{F}$ and sequence $1 \geq q_0 > q_1 > \ldots > q_j \geq 0$, $\mathbb{E}\left[ X_i \,\middle|\, F(X_i) \in \left( q_{l+1}, q_l\right]\right] = {1}/{(q_l - q_{l+1})} \cdot \int_0^{q_l} r(v) \cdot \min\{q_l - q_{l+1}, q_l - v\} \, \mathrm{d}v$. 
\end{claim}
\end{proof} 

The expected costs $d_0,\ldots,d_k$ satisfy the following system, where $p(k)=1-(1-q_k)^{s_k}$.
\begin{align}
&d_0 = \frac{1-p(0)}{p(0)(1 - q_0)} \int_0^{1} r(v) \cdot \min\{1 - q_{0}, 1 - v\} \, \mathrm{d}v + C(0) + \frac{1}{q_0} \left(\sum_{l=1}^{j} (q_{l-1} - q_{l})d_{l} + q_j \cdot d_{j+1} \right). \label{eq:dp1} \\
&d_k = (1 - p(k)) \cdot d_{k-1} + p(k) \cdot \left(C(k) + \frac{1}{q_k} \left(\sum_{l=1}^{j-k} (q_{k+l-1} - q_{k+l})d_{k+l} + q_j \cdot d_{j+1} \right)\right), \nonumber \\
&\quad\qquad\qquad\qquad\qquad\qquad\qquad\qquad\qquad\qquad\qquad\qquad\qquad\qquad\quad k = 1, \ldots, j-1, \label{eq:dp2} \\
&d_j = (1 - p(j)) \cdot d_{j-1} + p(j) \cdot \left( C(j) + d_{j+1} \right), \label{eq:dp3} \\
&d_{j+1} = 0. \label{eq:dp4}
\end{align}
\begin{lemma} \label{lm:dp}
Given inputs to Algorithm~\ref{alg} and $C(k)$ as defined in Lemma \ref{lm:hirecost}, let $p(k) = 1 - \left( 1 - q_k \right)^{s_k}$. The expected cost of Algorithm~\ref{alg} satisfies ${\ALG}_n(F) \leq d_0$, where $\{d_k\}_{k=0}^{j+1}$ solve \eqref{eq:dp1}-\eqref{eq:dp4}. 
\end{lemma}
\begin{proof}
Since $\ALG_n(F) \leq \ALG'(F)$, it suffices to show \eqref{eq:dp1}-\eqref{eq:dp4} correctly computes the expected cost incurred in each state by $\text{ALG}'$. The overall expected cost $\ALG'(F)$ is then given by $d_0$, the cost in the initial state. For transition probabilities, observe that for $k = 1, \ldots, j$, $\text{ALG}'$ goes back to state $k-1$ from state $k$ if every costs $x$ observed within the $s_k$ expanse gives $F(x) > q_k$. As such, 
$$\mathbb{P}[\text{visiting state } k-1 \text{ from state } k] = \left( 1 - q_k \right)^{s_k} = 1 - p(k).$$ 
With the remaining probability of $p(k)$, $\text{ALG}'$ makes a contract of expected cost $C(k)$ and transits to a higher state $k+l$ when it observes $F(x) \in (q_{k+l}, q_{k+l-1}]$, which has the conditional probability
$$\mathbb{P}[\text{visiting state } k+l \text{ from state } k \mid \text{state } k] = \frac{q_{k+l-1} - q_{k+l}}{q_k}.$$
$\text{ALG}'$ visits the final state $j+1$ only after it initiates a contract for a value below $q_j$, which has a conditional probability of ${q_j}/{q_k}$ in state $k$. When $\text{ALG}'$ initiates a contract in state $j$, the only state it transits to is the terminal state $j+1$. Combining the transition probabilities, we can write down the expressions \eqref{eq:dp2}-\eqref{eq:dp4} for $d_1, \ldots, d_j$, with the terminal condition $d_{j+1} = 0$.

State $0$ requires a different approach since there is no lower states to transit to. With probability $1 - p(0)$, $\text{ALG}'$ stays in state $0$ and contracts a cost $x$ with $F(x) \in (q_0, 1]$ for one time. As such, 
\begin{align*}
d_0 &= (1-p(0))\left(\frac{1}{1 - q_0} \int_0^{1} r(v) \cdot \min\{1 - q_{0}, 1 - v\} \, \mathrm{d}v + d_0\right) \\
&\quad + p(0) \cdot \left(C(0) + \frac{1}{q_0} \left(\sum_{l=1}^{j} (q_{l-1} - q_{l})d_{l} + q_j \cdot d_{j+1} \right)\right),
\end{align*}
where ${1}/{(1 - q_0)} \cdot \int_0^{1} r(v) \cdot \min\{1 - q_{0}, 1 - v\} \, \mathrm{d}v$ corresponds to $\mathbb{E}\left[ X_i \,\middle|\, F(X_i) \in (q_0, 1]\right]$. Upon rearranging, this becomes \eqref{eq:dp1}. This completes the set up of the recursion.
\end{proof}

\begin{remark} \label{remark:alg}
For any positive integer $n$, for $k = 0, \ldots, j$, where $j = \left\lceil \log_2 \left( {n}/{4}\right)\right\rceil$, with $q_k = {1}/{2^k}$, $s_k = 2^k$, $t_k = 4 \cdot 2^k$, Algorithm~\ref{alg} recovers Algorithm 3 in \cite{Disser2019}, which guarantees a competitive ratio of $6.052$ for general distribution. 
\end{remark}

\section{Upper Bound for OSCC on Uniform Distribution} \label{section:unif}
In this section, as a warm up to apply our LP framework for bounding the competitive ratio of online algorithms, we study Algorithm~\ref{alg} with a specific family of inputs and values uniformly distributed over a closed interval $I$. Since scaling all values by the same factor does not change the competitive ratio, we may assume $I=[0,1]$, and denote its CDF by $F_{U[0,1]}$. In what follows, we use $\text{ALG}$ to denote Algorithm~\ref{alg} with the inputs specified in Section~\ref{subsec:unif1}; ${\ALG}_n\left(F_{U[0,1]} \right)$ to denote its expected costs on an instance of size $n$ with the uniform distribution; and $R_n^{\text{C}}\left(\text{ALG}, F_{U[0,1]} \right)$ to denote its competitive ratio.

We first present our main result of this section, which concludes that our algorithm has a constant competitive ratio that is lower than the existing $2.965$ bound by~\cite{Disser2019}. 

\begin{theorem} \label{thm:uniform}
There exists an algorithm $\text{ALG}$ with $R_n^{\text{C}}\left(\text{ALG}, F_{U[0,1]} \right) \leq 2.945$ for all $n\geq 1$ and $\lim_{n \to \infty} R_n^{\text{C}}\left(\text{ALG}, F_{U[0,1]} \right) = 2.908$. 
\end{theorem}

To prove Theorem~\ref{thm:uniform}, we start by setting up our choice of algorithm inputs and compute the associated costs specific to $F_{U[0,1]}$ in Section~\ref{subsec:unif1}. In Section~\ref{subsec:unif2}, we show that ${\ALG}_n\left(F_{U[0,1]} \right)$ can be upper bounded by an LP based on System~\eqref{eq:dp1}-\eqref{eq:dp4}. The main obstacle in analyzing this LP is the complexity of terms that represent contract costs and probabilities. By bounding these terms with simpler expressions, we obtain another LP that is an upper bound to the former and simple enough to characterize its optimal dual solution. At the end of this section, we present the proof of Theorem~\ref{thm:uniform}. Since the proof contains numerical results on the upper bound, we devote a separate subsection (Section~\ref{subsec:unif3}) to it. All proofs to the claims in this section are deferred to Appendix~\ref{appsub:unif}. 

\subsection{Algorithm and Expected Costs} \label{subsec:unif1}

Moving forward, we study $\text{ALG}$ with the following family of inputs: introduce parameters $q, a, \beta$, with $a > 0$, $q > 1$, ${q^2}/{(a(q-1))} \cdot \ln\left( {q}/{(2q-1)}\right) < \beta \leq 1$, and define $b = {a(q-1)}/{q^2}$. Set inputs
\begin{itemize}
    \item $q_k = {\beta}/{q^k}$ for all $k = 0, \ldots, j$; 
    \item $t_k = \left\lfloor a \cdot q^k + k + (2 - {a}/{q}) \right\rfloor$ for all $k = 0, \ldots, j$;
    \item and $s_0 = 1$, $s_k = \lceil b \cdot q^{k} \rceil$ for $k = 1, \ldots, j$. 
\end{itemize}
Note that when $k = \left\lceil \log_q ({n}/{a}) \right\rceil$, $\text{ALG}$ will contract a value below ${\beta}/{q^k}$ for $t_k \geq \lfloor a \cdot q^k \rfloor \geq n$ times since $a \cdot q^k \geq n$ and $n$ is an integer. Therefore, the largest possible value of $k$ is $j = \left\lceil \log_q ({n}/{a}) \right\rceil$. 

Before proceeding, we first show that our choice of $s_k$ and $t_k$ guarantees the coverage of all time steps. When the algorithm contracts with $k$ and enters state $k+1$, the worst case searching time is 
\begin{align*}
\sum_{l=0}^{k+1} s_k = \sum_{l=1}^{k+1} \lceil b \cdot q^{l} \rceil + 1 \leq b \sum_{l=1}^{k+1} q^l + (k+1) + 1 = \frac{bq^2}{q-1} \cdot q^k - \frac{bq}{q-1} + k + 2,
\end{align*}
where in the first inequality we applied $\lceil b \cdot q^{l} \rceil \leq b \cdot q^{l} + 1$. Then, by choosing $b = {a(q-1)}/{q^2}$, we can ensure $\sum_{l=1}^{k+1} \lceil b \cdot q^{l} \rceil + 1 \leq a \cdot q^k + k + (2 - {a}/{q})$. Since the $\sum_{l=1}^{k+1} \lceil b \cdot q^{l} \rceil + 1$ is an integer, taking the integer part of the RHS retains the inequality, which allows us to conclude $\sum_{l=0}^{k+1} s_l \leq t_k$. 

Now in System~\eqref{eq:dp1}-\eqref{eq:dp4}, the transition probabilities become $p(k) = 1 - \left( 1 - {\beta}/{q^k} \right)^{\lceil b \cdot q^{k} \rceil}$ for $k = 1, \ldots, j$ and $p(0) = \beta$; from state $k$, the conditional probability of advancing to a higher state $k+l$ is ${1}/{q_k} \cdot (q_{k+l-1} - q_{k+l}) = {(q-1)}/{q^{l}}$, and to state $j+1$ is ${1}/{q^{j-k}}$. For $d_0$, with $q_0 = \beta$, the first term in \eqref{eq:dp1} simplifies to ${(1-\beta^2)}/{(2\beta)}$. In the next result, we provide expressions of $C(k)$ under $F_{U[0,1]}$, obtained by taking $r(v) = 1$ in Lemma~\ref{lm:hirecost}. The details are deferred to Appendix~\ref{app:unif}. 
\begin{lemma} \label{lm:uniformkcost}
For uniformly distributed costs, a contract made by $\text{ALG}$ in state $k$ has an expected cost at most $C(k) = {\beta(q^2-1)}/{(2q^2)} \cdot \sum_{l=k}^j  {1}/{q^{2l-k}}\cdot \left\lfloor a \cdot q^l + l + \left(2 - {a}/{q}\right)\right\rfloor$ for $k = 0, \ldots, j-1$, and $C(j) = {\beta}/{(2q^{j})} \cdot \left\lfloor a \cdot q^j + j + \left(2 - {a}/{q}\right)\right\rfloor.$
\end{lemma}
We will use Lemma~\ref{lm:uniformkcost} for numerical guarantees when $n$ is small. For large $n$, using $p(k)$ and  $C(k)$ as they are becomes challenging to use for analysis. First, we simplify $C$ with an upper bound in the following lemma, and later we will show how to simplify $p(k)$ once we have introduced our LP. The proof of the next lemma is a simple calculation deferred to Appendix~\ref{app:unif}. 
\begin{lemma} \label{lm:costub}
Define $w(k) = {\beta}/{(2q^k)} \cdot \left( k + (2 - {a}/{q}) + {1}/{(q^2-1)}\right)$, $\tilde{C}(k) = {a\beta(q+1)}/{(2q)} + w(k)$, and $\tilde{C}(j) = {a\beta}/{2} + {\beta}/{(2q^j)} \cdot \left( j + \left(2 - {a}/{q} \right)\right)$. For $k = 0, \ldots, j-1$, $C(k) \leq \tilde{C}(k)$, and $C(j) \leq \tilde{C}(j)$. 
\end{lemma}
Similar to what we have done in Section~\ref{section:lowerbound}, we will turn System~\eqref{eq:dp1}-\eqref{eq:dp4} into an LP that allows us to derive a competitive ratio upper bound for $\text{ALG}$. 

\subsection{Upper Bound on Performance Guarantee via Linear Programming} \label{subsec:unif2}
In this subsection, we formulate an LP to upper bound ${\ALG}_n\left(F_{U[0,1]} \right)$. We obtain the constraints by substituting the conditional probabilities computed in Section~\ref{subsec:unif1} into System~\eqref{eq:dp1}-\eqref{eq:dp4}, treating $\{d_k\}_{k=0}^{j+1}$ as variables, and using placeholders $C_k$ and $p_k$ for costs and probabilities respectively. The placeholders make the LP versatile for studying the effect of changing contract costs and probabilities. We formally present this in the meta-LP \ref{form:LPunif}$\left(\bm{p}, \bm{C}\right)$, which takes a vector $\bm{p} = (p_1, \ldots, p_j) \in \mathbb{R}^j$ of probabilities and a vector $\bm{C} = (C_0, C_1, \ldots, C_j) \in \mathbb{R}^{j+1}$ of costs as inputs. 
\begin{alignat}{3}
& \max \quad && d_0 \tag*{$(LP)_{\text{Unif}}$}\label{form:LPunif} \\
& \text{s.t.} \quad && d_0 \leq \frac{1-\beta^2}{2\beta} + C_0 + \sum_{l=1}^j \frac{q-1}{q^l}d_l + \frac{1}{q^j} d_{j+1}, \label{constraint:alg1P1} \\ 
& && d_k \leq (1 - p_k) \cdot d_{k-1} + p_k \cdot \left( C_k + \sum_{l=1}^{j-k} \frac{q-1}{q^l}d_{k+l} + \frac{1}{q^{j-k}} d_{j+1} \right), \quad && k = 1, \ldots, j-1, \label{constraint:alg1P2} \\
& && d_j \leq (1 - p_j) \cdot d_{j-1} + p_j \cdot \left( C_j + d_{j+1} \right), \label{constraint:alg1P3}\\
& && d_{j+1} \leq 0, \label{constraint:alg1P4} \\
& && d_k \geq 0, && k = 0, \ldots, j+1. \nonumber
\end{alignat}
Now the solution to \eqref{eq:dp1}-\eqref{eq:dp4} is a feasible solution to \ref{form:LPunif} with appropriate inputs, so we have the following result stating that ${\ALG}_n\left(F_{U[0,1]} \right)$ can be upper-bounded by the objective of this LP. 
\begin{proposition} \label{prop:uniformlp}
Let $p(k) = 1 - \left( 1 - {\beta}/{q^k} \right)^{\lceil b \cdot q^{k} \rceil}$ and $\tilde{C}(k)$ in Lemma \ref{lm:costub}, let $\bm{p} = (p(1), \ldots, p(j))$, $\bm{\tilde{C}} = (\tilde{C}(0), \ldots, \tilde{C}(j))$, and $d_0^*$ be the optimal value of \ref{form:LPunif}$\left(\bm{p}, \bm{\tilde{C}} \right)$. Then,~${{\ALG}_n\left(F_{U[0,1]} \right) \leq d_0^*}$.
\end{proposition}
\begin{proof}
Replacing $C(k)$ in \eqref{eq:dp1}-\eqref{eq:dp4} with the respective upper bound $\tilde{C}(k)$ transforms \eqref{eq:dp1}-\eqref{eq:dp4} into the constraints of \ref{form:LPunif}$\left(\bm{p}, \bm{\tilde{C}} \right)$, so the solution $\{d_k\}_{k=0}^{j+1}$ to System~\eqref{eq:dp1}-\eqref{eq:dp4} with probabilities and costs defined in Section~\ref{subsec:unif1} is a feasible solution to \ref{form:LPunif}$\left(\bm{p}, \bm{\tilde{C}} \right)$. The conclusion follows from $d_0 \geq {\ALG}_n\left(F_{U[0,1]} \right)$ by Lemma~\ref{lm:dp} and the optimality of $d_0^*$. 
\end{proof}
For general $p(k)$, analyzing~\ref{form:LPunif} is non-trivial. However, our next result shows that replacing $p(k)$ by a uniform lower bound can only make the value of~\ref{form:LPunif} larger. We later leverage this to provide a closed-form dual optimal solution.
\begin{lemma} \label{lm:uniformrelax}
For $\tilde{p} \in [0, 1]$ such that $\tilde{p} \leq p(k)$ for all $k$. Let $\bm{\tilde{p}} = (\tilde{p}, \ldots, \tilde{p})\in \mathbb{R}^j$ and let $\Tilde{d}_0^*$ and $d_0^*$ be the optimal value of \ref{form:LPunif}$\left(\bm{\tilde{p}}, \bm{\tilde{C}} \right)$ and \ref{form:LPunif}$\left(\bm{p}, \bm{\tilde{C}} \right)$, respectively. Then, $d_0^* \leq \Tilde{d}_0^*$. 
\end{lemma}
\begin{proof}
Given the optimal solution $\{d_k^*\}_{k = 0}^{j+1}$ to \ref{form:LPunif}$\left(\bm{p}, \bm{\tilde{C}} \right)$, we show that it is a feasible to \ref{form:LPunif}$\left(\bm{\tilde{p}}, \bm{\tilde{C}} \right)$. The result then follows immediately from the maximality of $\Tilde{d}_0^*$. We will need the following claim, which follows by induction. 
\begin{claim} \label{claim:dmonotone}
Let $\{d_k^*\}_{k = 0}^{j+1}$ be the set of optimal solution to \ref{form:LPunif} with arbitrary inputs $p_k \in (0, 1]$, $C_k \geq 0$. If $C_k$ satisfies $q\cdot C_k \geq C_{k+1}$ for all $k$, then $\{d_k^*\}_{k = 0}^{j+1}$ is non-increasing in $k$, i.e., $d_{k-1}^* \geq d_k^*$. 
\end{claim}
Since $q\cdot w(k) \geq r(k+1)$, we have $q \cdot \tilde{C}(k) \geq \tilde{C}(k+1)$. Thus by Claim \ref{claim:dmonotone}, $d_{k-1}^* \geq d_k^*$ in \ref{form:LPunif}$\left(\bm{p}, \bm{\tilde{C}} \right)$. We now check that subtracting the RHS of Constraints~\eqref{constraint:alg1P1}-\eqref{constraint:alg1P3} evaluated at $\{d_k^*\}_{k=0}^{j+1}$ in \ref{form:LPunif}$\left(\bm{p}, \bm{\tilde{C}} \right)$ from that in \ref{form:LPunif}$\left(\bm{\tilde{p}}, \bm{\tilde{C}} \right)$ gives a nonnegative result. Let $$S_k = \tilde{C}(k) + \sum_{l=1}^{j-k} \frac{q-1}{q^l}d_{k+l}^* + \frac{1}{q^{j-k}} d_{j+1}^*$$ and rearrange Constraint~\eqref{constraint:alg1P2} to obtain $d_{k-1}^* - d_k^* = p(k) \cdot (d_{k-1}^* - S_k)$. From here, $d_{k-1}^* \geq d_k^*$ implies $d_{k-1}^* - S_k \geq 0$. Therefore, the difference in the RHS is $(p(k) - \tilde{p}) \cdot \left(d_{k-1}^* - S_k \right) \geq 0$ for Constraint~\eqref{constraint:alg1P2}. The same applies to Constraint~\eqref{constraint:alg1P3}. The RHS of Constraint~\eqref{constraint:alg1P1} is the same in both LP. Since the RHS of constraints in \ref{form:LPunif}$\left(\bm{\tilde{p}}, \bm{\tilde{C}} \right)$ at $\{d_k^*\}_{k=0}^{j+1}$ is at least as large as the one in \ref{form:LPunif}$\left(\bm{p}, \bm{\tilde{C}} \right)$, $\{d_k^*\}_{k=0}^{j+1}$ is feasible to \ref{form:LPunif}$\left(\bm{\tilde{p}}, \bm{\tilde{C}} \right)$. This concludes the proof. 
\end{proof}

Following Lemma \ref{lm:uniformrelax}, since $1-e^{-b\beta} \leq 1-\left( 1 - {\beta}/{q^k} \right)^{b \cdot q^k} \leq p(k)$ for all $k$, we can replace all $p(k)$ with $\tilde{p} = 1-e^{-b\beta} = 1 - e^{-{a\beta(q-1)}/{q^2}}$. We now solve \ref{form:LPunif}$\left(\bm{\tilde{p}}, \bm{\tilde{C}} \right)$ using duality. Introducing dual variable $\alpha_k$ for the constraint starting with $d_k$, \ref{form:LPunif}$\left(\bm{\tilde{p}}, \bm{\tilde{C}} \right)$ has the following dual:
\begin{alignat}{3}
& \min \quad && \left(\frac{1-\beta^2}{2\beta} + \tilde{C}(0) \right) \cdot \alpha_0 + \tilde{p} \cdot \sum_{k=1}^{j-1} \tilde{C}(k) \cdot \alpha_k + \tilde{C}(j) \cdot \alpha_j \tag*{$(DLP)_{\text{Unif}}$}\label{form:DLPunif} \\
& \text{s.t.} \quad && \alpha_0 \geq (1-\tilde{p}) \cdot \alpha_1 + 1, \label{constraint:alg1D1} \\
& && \alpha_k \geq (1-\tilde{p}) \cdot \alpha_{k+1} + \frac{q-1}{q^k} \alpha_0 + \tilde{p} \cdot \sum_{l=1}^{k-1} \frac{q-1}{q^{k-l}} \alpha_l, \quad && k = 1, \ldots, j-1, \label{constraint:alg1D2}\\
& && \alpha_j \geq \frac{q-1}{q^j} \alpha_0 + \tilde{p} \cdot \sum_{l=1}^{j-1} \frac{q-1}{q^{j-l}} \alpha_l, \label{constraint:alg1D3}\\
& && \alpha_{j+1} \geq \frac{1}{q^j} \alpha_0 + \tilde{p} \cdot \sum_{l=1}^{j} \frac{1}{q^{j-l}} \alpha_l, \label{constraint:alg1D4} \\
& && \alpha_k \geq 0, && k = 0, \ldots, j+1. \nonumber 
\end{alignat}

Intuitively, the optimal dual solution $\alpha_k^*$ can be interpreted as the expected number of visits to state $k$ by $\text{ALG}'$. The constraints in~\ref{form:DLPunif} capture the flow equations between states in $\text{ALG}'$.

We now characterize the optimal solution to \ref{form:DLPunif}. The main idea is that Constraint~\eqref{constraint:alg1D2}, when tightened, can be rearranged to give the following recurrence relation, 
\begin{align}
(1+(q-1)\tilde{p})\alpha_k = (1+q-\tilde{p})\alpha_{k+1} - q(1-\tilde{p})\alpha_{k+2}, \label{eq:dualrecursion}
\end{align}
which can then be solved explicitly. The solution tightens all constraints, hence is optimal to \ref{form:DLPunif}. This recurrence relation also yields a simpler feasible solution that becomes handy when upper bounding the competitive ratio for general distributions in Section~\ref{section:general}. Since the derivation is calculation-heavy, we present the solution in Lemma~\ref{lm:uniformdualsol} and defer the proof to Appendix~\ref{app:unif}. 

\begin{lemma} \label{lm:uniformdualsol}
Let $E = (2q-1)\tilde{p}-(q-1)$, $\theta = {(q-1)}/{E}$, $\lambda = {(1+(q-1)\tilde{p})}/{(q(1-\tilde{p}))}$. Then, $\alpha_0^* = \theta \left( {q\tilde{p}}/{(q-1)} - (1-\tilde{p})\lambda^{-j}\right)$, $\alpha_k^* = \theta (1 - \lambda^{k-(j+1)})$ for $k = 1, \ldots, j$ is optimal to \ref{form:DLPunif}.
\end{lemma}


\subsection{Proof of Theorem~\ref{thm:uniform}} \label{subsec:unif3}
We first present the bound for large $n$ using the objective of \ref{form:DLPunif}. Since both \ref{form:LPunif}$\left(\bm{\tilde{p}}, \bm{\tilde{C}} \right)$ and \ref{form:DLPunif} can be solved by tightening all constraints and solving the resulting system of equations, both systems have respective unique solutions, therefore the pair of LP admits strong duality, i.e., $\tilde{d}_0^* = \tilde{v}^*$, where $\tilde{v}^*$ denotes the optimal value of \ref{form:DLPunif}. We thus have $\tilde{v}^* = \tilde{d}_0^* \geq d_0^*$, where the inequality is justified by Lemma \ref{lm:uniformrelax}. Then by Proposition \ref{prop:uniformlp}, we can get an upper bound on ${\ALG}_n\left(F_{U[0,1]} \right)$ by evaluating $\tilde{v}^*$ using the solution in Lemma \ref{lm:uniformdualsol}:
\begin{align*}
\tilde{v}^* &= \theta\left(\frac{1-\beta^2}{2\beta} + \tilde{C}(0) \right) \left(\frac{q\tilde{p}}{q-1} - (1-\tilde{p})\lambda^{-j}\right) + \tilde{p} \theta \sum_{k=1}^{j-1} \tilde{C}(k) \cdot (1 - \lambda^{k-(j+1)}) + \theta \cdot \tilde{C}(j) \cdot (1 - \lambda^{-1}).
\end{align*}
We start the simplification of $\tilde{v}^*$ by noting that $\theta \cdot {q\tilde{p}}/{(q-1)} = \tilde{p} \theta + {\tilde{p}}/{E}$, and thus writing $\tilde{v}^*$ in the form $G_1 \cdot j + G_2 + G_3 - G_4$, which can then be expanded by substituting in the expressions for $\tilde{C}(k)$. For brevity, we skip the details and give the simplified expressions for $G_1, G_2, G_3, G_4$ below, 
\begin{align*}
G_1 &= \tilde{p}\theta \cdot \frac{a\beta(q+1)}{2q}, \\
G_2 &= \frac{\tilde{p}}{2\beta q} \cdot \frac{q^2(1 - \beta^2) + \theta \beta^2(q^2+1)}{E}, \\
G_3 &\leq \frac{\tilde{p}\beta}{2E}\left[(q+1)\left(2 - \frac{a}{q} + \frac{1}{q-1}\right) - \frac{1}{q^{j}}\left(j + 2 - \frac{a}{q} + \frac{q^{2} + 2q}{q^{2} - 1}\right)\right] \leq \frac{\tilde{p}\beta}{2E}(q+1)\left(2 - \frac{a}{q} + \frac{1}{q-1}\right), \\
G_4 &= \theta \left[ (1-\tilde{p}) \left( \frac{1-\beta^2}{2\beta} - \frac{a\beta(q+1)}{2q} \cdot \theta(1-\tilde{p}) \right) \lambda^{-j} + \frac{a\beta \tilde{p}}{2\lambda} \cdot \frac{2(1-\tilde{p})+ \tilde{p}q}{E}\right]. 
\end{align*}
Here $G_3$ contains terms involving $\lambda$ which have negligible contributions to $\tilde{v}^*$ {when $n$ is large}. Before we provide a bound on $\tilde{v}^*$, we need the following claim, where recall $\OPT_n(F_{U[0,1]})=\sum_{i=1}^n 1/(i+1)$:
\begin{claim} \label{claim:jOPT}
Let $\kappa = {(1 - \gamma - \ln a)}/{(\ln q)} + 1$, where $\gamma$ is the Euler–Mascheroni constant. For any positive integer $n$, $j \leq {\OPT_n}/{(\ln q)} + \kappa - \text{error}_n/{(\ln q)}$, with $\text{error}_n \leq 1/(2n) \rightarrow 0$ as $n$ grows. 
\end{claim}
Let $M = G_2 + G_3 - G_4$. Using Claim \ref{claim:jOPT} to express $j$ in terms of $\OPT_n$, we have ${\tilde{v}^*} \leq G_1 \cdot {\OPT_n}/{\ln q} + (G_1 \cdot \kappa + M)$. {  For fixed $a,\beta,q$, when $n$ grows, we have $\lim_{n\to \infty } R_n^C (\text{ALG},F_{U[0,1]})\leq G_1/\ln q$. Note that $G_1/\ln q$ only depends on $a,\beta$ and $q$. By solving numerically, $\min_{a, \beta, q} {G_1}/{(\ln q)}$ is achieved at $a = 4, \beta = 0.89, q = 2.27$, we obtain $\lim_{n\to \infty} R_n^C(\text{ALG},F_{U[0,1]})\leq 2.908$.}

{  We now focus on $R_n^{\text{C}}\left(\text{ALG}, F_{U[0,1]} \right)$ for all $n\geq 1$. We use the same parameters $a,\beta$ and $q$ obtained above.} Note that the first term in $G_4$ is non-positive with $\lambda^{-j}$ decreasing in $j$, thus $M$ is decreasing in $j$, and we can upper bound $G_1 \cdot \kappa + M$ with $G_4$ evaluated at the smallest $j$. At the same time, $\OPT_n$ can also be lower bounded at the smallest $j$, and recall that $j = \left\lceil \log_q ({n}/{a}) \right\rceil$, for instances with $n \in [\lfloor q^{j-1} \cdot a + 1\rfloor, \lfloor q^{j} \cdot a \rfloor]$, $j$ stays the same, so the smallest $\OPT_n$ for each $j$ happens at $n = \lfloor q^{j-1} \cdot a + 1\rfloor$. Evaluating the above at $j = 66$, we get ${(G_1 \cdot \kappa + M)}/{\OPT_n} \leq 0.037$. Therefore, for $j \geq 66$, $R_n^{\text{C}}\left(\text{ALG}, F_{U[0,1]} \right) \leq 2.908+0.037 = 2.945$. 

For $j \leq 65$ we do the following: For $j \leq 40$, we can compute directly the exact solution to \eqref{eq:dp1}-\eqref{eq:dp4}. For $j \in [41, 65]$, for numerical stability we avoid computing $p(k)$ and solve \ref{form:LPunif}$\left(\bm{\tilde{p}}, \bm{C} \right)$ instead, where $\bm{C} = (C(0), \ldots, C(k))$. This is justified because the optimal solution to this LP is also an upper bound on ${\ALG}_n\left(F_{U[0,1]} \right)$, with the same proof as Lemma \ref{lm:uniformrelax}. We use Claim \ref{claim:jOPT} to approximate $\OPT_n$ for relatively large $n$ for computational efficiency. The largest competitive ratio upper bound for $j \leq 40$ from the exact solution to System~\eqref{eq:dp1}-\eqref{eq:dp4} is $2.902$ at $j = 40$, while the largest upper bound via \ref{form:LPunif}$\left(\bm{\tilde{p}}, \bm{C} \right)$ is $2.938$ at $j = 41$. 

Figure \ref{fig:unifdp} plots the worst case numerical upper bound (red solid line) computed using above-mentioned methods and the asymptotic ratio, with the existing bound $2.965$ for comparison. In the figure, the competitive ratio computed using the exact recursion is increasing in the size of the instance and approaching the asymptotic ratio of $2.908$, while the switch to \ref{form:LPunif}$\left(\bm{\tilde{p}}, \bm{C} \right)$ at $j=41$ is reflected by the jump in the plot. We note that compared to the exact recursion, \ref{form:LPunif}$\left(\bm{\tilde{p}}, \bm{C} \right)$ (orange dotted line) overestimates the competitive ratio for small instances, but this gaps disappears as $n$ grows. Table \ref{table:unifdp} reports the results for $j = 55, \ldots, 65$. Overall, we obtain $R_n^{\text{C}}\left(\text{ALG}, F_{U[0,1]} \right) \leq 2.945$ for all $n\geq 1$.
\begin{figure}[H]
\centering
\begin{minipage}{0.5\textwidth}
    \centering
    \includegraphics[width=\linewidth]{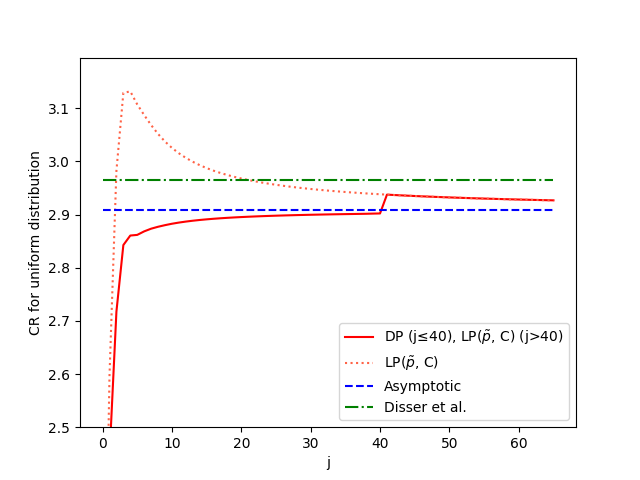}
    \caption{Numerical $R_n^{\text{C}}\left(\text{ALG}, F_{U[0,1]} \right)$ for $j \leq 65$}
    \label{fig:unifdp}
\end{minipage}
\begin{minipage}[t]{0.49\textwidth}\small
    \centering
\begin{tabular}{|c|c|}
\hline
$j$ & \ref{form:LPunif}$\left(\bm{\tilde{p}}, \bm{C} \right)$ objective \\
\hline
55 & 2.930066 \\
56 & 2.929678 \\
57 & 2.929305 \\
58 & 2.928944 \\
59 & 2.928595 \\
60 & 2.928258 \\
61 & 2.927932 \\
62 & 2.927616 \\
63 & 2.927310 \\
64 & 2.927014 \\
65 & 2.926727 \\
\hline
\end{tabular}
\caption{LP objective for $j = 55, \ldots, 65$}
\label{table:unifdp}
\end{minipage}
\end{figure}

\section{Competitive Ratio Upper Bound for OSCC} \label{section:general}
We now extend the analysis to any CDF $F$. For this, we consider $\text{ALG}$ with the same family of inputs described at the beginning of Section~\ref{subsec:unif1}. We require that $a$ and $q$ to satisfy ${a}/{q} \geq 2$ so that $t_0 \leq \left\lfloor a \cdot q^0 \right\rfloor$, which simplifies the expression of $t_0$ and thus part of the analysis (see details in Lemma~\ref{lm:betastar}). We then assume $n \rightarrow \infty$ and perform asymptotic analysis.

\begin{theorem} \label{thm:general}
There exists an algorithm $\text{ALG}$ which achieves $\lim_{n \rightarrow \infty} R_n^{\text{C}}(\text{ALG}) \leq 4.179$. 
\end{theorem}

The remainder of this section is dedicated to the proof of Theorem~\ref{thm:general}, which demonstrates that the LP framework built in previous sections can be naturally extended to general distributions. In what follows, for any CDF $F \in \mathcal{F}$, we assume the change of variable $F^{-1}(u) = \int_0^u r(v) \, \mathrm{d}v$. In addition, we would occasionally abuse notation and treat $\beta / q^{j+1}$ as $0$ when necessary.

The proof consists of three parts. In Section~\ref{subsec:general1}, we first formulate an infinite-dimensional LP based on System~\eqref{eq:dp1}-\eqref{eq:dp4}, which encodes the CDF $F$ as a variable so that we get the worst case over all $F \in \mathcal{F}$, and show that the LP gives an upper bound on the competitive ratio on instances of size $n$. We then present the corresponding dual, prove weak duality of the primal-dual pair which ensures the right upper bounding direction. In Section~\ref{subsec:general2}, we present a dual feasible solution and use it to establish that the dual objective is equivalent to the supremum of a piecewise function. Section~\ref{subsec:general3} focuses on studying this function and upper bounding its supremum. At the end of this section, we prove Theorem~\ref{thm:general}. Proofs of claims in this section are deferred to Appendix~\ref{app:general}. 

\subsection{Upper Bound on Competitive Ratio via LP} \label{subsec:general1}
Similar as in Section~\ref{section:unif}, we formulate an LP for which the solution to System \eqref{eq:dp1}-\eqref{eq:dp4} for $\text{ALG}$ is a feasible solution. We begin by upper bounding the cost $C(k)$ for $\text{ALG}$. For notational simplicity, we first define a set of functions that will appear in the bound for $C(k)$ throughout this section. For $l = 0, \ldots, j$, let $\rho_l = l + 2 - {a}/{q}$, and define the function $\psi_l: [0, 1] \rightarrow \mathbb{R}$ and $\psi_{\text{fail}}(v)$ as follow
\begin{align*}
\psi_l(v) &= 
\begin{cases}
\min \left\{\frac{\beta (q-1)}{q^{l+1}}, \frac{\beta}{q^l} - v \right\}, & l = 0, \ldots, j-1, \\
\frac{\beta}{q^j} - v, & l = j.
\end{cases}\qquad \text{and}\qquad  \psi_{\text{fail}}(v) = \min\{1 - \beta, 1 - v\}.
\end{align*}
\begin{lemma} \label{lm:generalkcost}
For any CDF $F \in \mathcal{F}$, there exists $r(v) \geq 0$  for $v \in [0, 1]$ such that a contract made by $\text{ALG}$ in state $k$ has an expected cost at most $\tilde{C}(k)$, where for $k = 0, \ldots, j$, $\tilde{C}(k) = {q^k}/{\beta} \cdot \left( \sum_{l=k}^{j} \left(aq^l + \rho_l \right) \int_0^{\beta / q^l} r(v) \cdot \psi_l(v) \, \mathrm{d}v \right)$.
\end{lemma}
\begin{proof}
Substituting the inputs of $\text{ALG}$ into Lemma \ref{lm:hirecost} and removing the floor gives the result. Note that for $l = j$, $ \min \left\{{\beta}/{q^j} - 0, {\beta}/{q^j} - v \right\} = {\beta}/{q^j} - v$, which corresponds to $\psi_j(v)$. 
\end{proof}
Since we use the same inputs as in Section~\ref{subsec:unif1}, the probabilities $p(k)$ remain the same. Note that the first term in \eqref{eq:dp1}, which corresponds to the expected cost incurred when failing to contract and remaining in state $0$, simplifies to ${1}/{\beta} \cdot \int_0^{1} r(v) \cdot \psi_{\text{fail}}(v) \, \mathrm{d}v.$ We now transform \eqref{eq:dp1}-\eqref{eq:dp4} into LP constraints by treating $\{d_k\}_{k=0}^{j+1}$ as variables and replacing $C(k)$ with the upper bound $\tilde{C}(k)$ in Lemma \ref{lm:generalkcost}. Similar to Section~\ref{subsec:lb2}, we impose the constraint $\OPT_n(F) = 1$ by using Lemma~\ref{lm:opt} and treat $r$ as a variable, leading to the infinite-dimensional LP \ref{form:LPUB}$\left(\bm{p}\right)$:
\begin{alignat}{2}
& \sup_{\substack{\bm{d} \in \mathbb{R}_+^{j+2} \\ r: [0, 1] \rightarrow \mathbb{R}_+}} \quad d_0 \tag*{$(LP)_{\text{UB}}$}\label{form:LPUB} \\
\text{s.t.} \quad & d_0 \leq \frac{1}{\beta}\int_0^{1} r(v) \cdot \psi_{\text{fail}}(v) \, \mathrm{d}v + \tilde{C}(0) + \sum_{l=1}^j \frac{q-1}{q^l}d_l + \frac{1}{q^j} d_{j+1}, \label{constraint:alg2P1} \\ 
& d_k \leq (1 - p_k) \cdot d_{k-1} + p_k \cdot \left( \tilde{C}(k) + \sum_{l=1}^{j-k} \frac{q-1}{q^l}d_{k+l} + \frac{1}{q^{j-k}} d_{j+1} \right), k = 1, \ldots, j-1, \label{constraint:alg2P2}\\
& d_j \leq (1 - p_j) \cdot d_{j-1} + p_j \cdot \left( \tilde{C}(j) + d_{j+1} \right), \label{constraint:alg2P3}\\
& d_{j+1} \leq 0, \label{constraint:alg2P4}\\
& \int_0^1 r(v) \cdot \sum_{i=1}^n (1-v)^i \, \mathrm{d}v = 1. \label{constraint:alg2P5}
\end{alignat}
By construction, the solution to \eqref{eq:dp1}-\eqref{eq:dp4} with $F$ such that $\OPT_n(F) =\int_0^1 r(v) \cdot \sum_{i=1}^n (1-v)^i \, \mathrm{d}v = 1$ (see Lemma~\ref{lm:opt}) will be a feasible solution to \ref{form:LPUB} with the appropriate input probabilities. This leads to the following bound on the worst-case competitive ratio $R_n^{\text{C}}(\text{ALG}) = \sup_{F \in \mathcal{F}} R_n^{\text{C}}(\text{ALG},F)$. 
\begin{proposition} \label{prop:generallp}
Let $\tilde{p} = 1 - e^{-{a\beta(q-1)}/{q^2}}$ and $\bm{\tilde{p}} = (\tilde{p}, \ldots, \tilde{p}) \in \mathbb{R}^j$. Let $\Tilde{d}_0^*$ be the optimal value of \ref{form:LPUB}$\left(\bm{\tilde{p}}\right)$. Then for $n\geq 1$, the competitive ratio of $\text{ALG}$ satisfies ${R_n^{\text{C}}(\text{ALG}) \leq \Tilde{d}_0^*}$.
\end{proposition}
\begin{proof}
For any $F \in \mathcal{F}$, since $F^{-1}$ is nonnegative and strictly increasing, $r(v)$ in $F^{-1}(u) = \int_0^u r(v) \, \mathrm{d}v$ is positive for all $v \in [0, 1]$. Now suppose $F$ satisfies $\OPT_n(F) = 1$, by Lemma \ref{lm:opt}, its corresponding $r$ satisfies Constraint~\eqref{constraint:alg2P5}. Replacing all $C(k)$ in \eqref{eq:dp1}-\eqref{eq:dp4} with the upper bound $\tilde{C}(k)$ turns \eqref{eq:dp1}-\eqref{eq:dp4} into Constraint~\eqref{constraint:alg2P1}-\eqref{constraint:alg2P4}. Therefore, given $F$ and its corresponding $r$, the solution $\{d_k\}_{k=0}^{j+1}$ to System~\eqref{eq:dp1}-\eqref{eq:dp4} is feasible to \ref{form:LPUB}$\left(\bm{p}\right)$, with $\bm{p} = (p(1), \ldots, p(j))$ where $p(k) = 1 - \left( 1 - {\beta}/{q^k} \right)^{\lceil b \cdot q^{k} \rceil}$ and $p(0) = \beta$ same as in Section~\ref{subsec:unif1}. By Lemma \ref{lm:dp} and the optimality, the optimal solution to \ref{form:LPUB}$\left(\bm{p}\right)$ is an upper bound on $\ALG_n(F)$, for any $F$ satisfying $\OPT_n(F) = 1$. 

Now, note that $\tilde{p} \leq 1-\left( 1 - {\beta}/{q^k} \right)^{b \cdot q^k} \leq p(k)$, and clearly $q \cdot \tilde{C}(k) \geq \tilde{C}(k+1)$, therefore the proofs of Claim \ref{claim:dmonotone} and Lemma \ref{lm:uniformrelax} apply here, which justifies that the optimal solution to \ref{form:LPUB}$\left(\bm{\tilde{p}}\right)$ upper bounds the optimal solution to \ref{form:LPUB}$\left(\bm{p}\right)$. Since scaling does not impact the ratio, we can normalize any $F \in \mathcal{F}$ such that $\OPT_n(F) = 1$. Together with the objective being the supremum over all distributions, it follows that $R_n^{\text{C}}(\text{ALG}) = \sup_F {\ALG_n(F)}/{\OPT_n(F)} \leq \Tilde{d}_0^*$. 
\end{proof}
Following Proposition \ref{prop:generallp}, we will rely on weak duality to upper bound $R_n^{\text{C}}(\text{ALG},F)$. We formally present the dual of \ref{form:LPUB}$\left(\bm{\tilde{p}}\right)$ in \ref{form:DLPUB}, in which the variable $\alpha_k$ corresponds to Constraint~\eqref{constraint:alg2P1}-\eqref{constraint:alg2P4}, and $\zeta$ corresponds to Constraint~\eqref{constraint:alg2P5}.
\begin{alignat}{3}
& \inf_{\bm{\alpha} \in \mathbb{R}_+^{j+2}} \quad && \zeta \tag*{$(DLP)_{\text{UB}}$}\label{form:DLPUB} \\
& \text{s.t.} \quad && \alpha_0 \geq (1-\tilde{p}) \cdot \alpha_1 + 1, \label{constraint:alg2D1} \\ 
& && \alpha_k \geq (1-\tilde{p}) \cdot \alpha_{k+1} + \frac{q-1}{q^k} \alpha_0 + \tilde{p} \cdot \sum_{l=1}^{k-1} \frac{q-1}{q^{k-l}} \alpha_l, \quad && k = 1, \ldots, j-1, \label{constraint:alg2D2}\\
& && \alpha_j \geq \frac{q-1}{q^j} \alpha_0 + \tilde{p} \cdot \sum_{l=1}^{j-1} \frac{q-1}{q^{j-l}} \alpha_l, \label{constraint:alg2D3}\\
& && \alpha_{j+1} \geq \frac{1}{q^j} \alpha_0 + \tilde{p} \cdot \sum_{l=1}^{j} \frac{1}{q^{j-l}} \alpha_l, \label{constraint:alg2D4} \\
& && \zeta \sum_{i=1}^n (1-v)^i \geq N(\bm{\alpha},v), && v \in [0, 1]. \label{constraint:alg2D5} 
\end{alignat}

where $N(\bm{\alpha},v): \mathbb{R}_+^{j+1}\times [0, 1] \rightarrow \mathbb{R}$ equals
\begin{align*} 
\frac{\alpha_0}{\beta} \left( \psi_{\text{fail}}(v) + \sum_{l=0}^j (aq^l + \rho_l)\psi_{l}(v) \cdot \mathbbm{1}_{[0, \beta / q^l]}(v)\right) + \frac{\tilde{p}}{\beta} \sum_{k=1}^{j} \alpha_k q^k \left( \sum_{l=k}^j (aq^l + \rho_l)\psi_{l}(v) \cdot \mathbbm{1}_{[0, \beta / q^l]}(v)\right).
\end{align*}
{  The following result shows that \ref{form:DLPUB} is a weak dual to~\ref{form:LPUB}$\left(\bm{p}\right)$. The proof is in Appendix~\ref{app:general}.}
\begin{lemma} \label{lm:weakdual}
Formulation \ref{form:DLPUB} is a weak dual to \ref{form:LPUB}$\left(\bm{\tilde{p}}\right)$. Namely, for any feasible solution $\zeta$ and $d_0$ to \ref{form:DLPUB} and \ref{form:LPUB}$\left(\bm{\tilde{p}}\right)$ respectively, $\zeta \geq d_0$. 
\end{lemma}

In the next subsection, we present a dual feasible solution, and show that the proposed solution leads to a convex problem that upper bounds $R_n^{\text{C}}(ALG)$. We later optimize this convex problem.

\subsection{A Dual Feasible Solution} \label{subsec:general2}

Note that Constraint~\eqref{constraint:alg2D1}-\eqref{constraint:alg2D4} admit the same structure as in \ref{form:DLPunif}, so the optimal solution in Lemma \ref{lm:uniformdualsol} applies to \ref{form:DLPUB}. To keep subsequent calculations cleaner, we build on Lemma \ref{lm:uniformdualsol} to give a simpler feasible solution in the next result. We point out that the solution proposed in this section is a relaxation of the optimal solution in Lemma \ref{lm:uniformdualsol} by omitting the correction term $\lambda$, which later simplifies $N(\bm{\alpha}, v)$ to a cleaner expression. 
\begin{lemma} \label{lm:generaldualsol}
Let $E = (2q-1)\tilde{p}-(q-1)$ and $\theta = {(q-1)}/{E}$. Then $\alpha_0 = \theta \cdot {q\tilde{p}}/{(q-1)}$, $\alpha_k = \theta$ for $k = 1, \ldots, j$ satisfy Constraint~\eqref{constraint:alg2D1}-\eqref{constraint:alg2D4}. 
\end{lemma}
\begin{proof}
Clearly, $\alpha_1 = \ldots = \alpha_j$ satisfies Equation~\eqref{eq:dualrecursion} for $k = 1, \ldots, j-2$, so $\{\alpha_k\}_{k=0}^{j}$ is feasible to Constraint~\eqref{constraint:alg2D2}. It can be checked by routine calculation that $\alpha_0 = \theta \cdot {q\tilde{p}}/{(q-1)}$ and $\alpha_1 = \theta$ satisfies Constraint~\eqref{constraint:alg2D1}. We now check the feasibility for Constraint~\eqref{constraint:alg2D3}. Since $\alpha_j = \theta$ and $\tilde{p} < 1$, 
\begin{align*}
\text{RHS of \eqref{constraint:alg2D3}} &= \frac{q-1}{q^j} \cdot \theta \cdot \frac{q\tilde{p}}{q-1} + \tilde{p} \cdot \sum_{l=1}^{j-1} \frac{q-1}{q^{j-l}} \theta = \tilde{p}\cdot \frac{1}{q^{j-1}} \cdot \theta + \tilde{p} \left( 1 - \frac{1}{q^{j-1}}\right) \theta = \tilde{p}\theta < \text{LHS of \eqref{constraint:alg2D3}},
\end{align*}
Therefore, $\{\alpha_k\}_{k=0}^{j}$ is feasible to Constraint~\eqref{constraint:alg2D1}-\eqref{constraint:alg2D3}. $\alpha_{j+1}$ can be computed from Constraint~\eqref{constraint:alg2D4} and is not involved in future analysis, therefore we omit its derivation. This concludes the proof. 
\end{proof}
{ Moving forward, we refer to $N(\bm{\alpha}, v)$ simply as $N(v)$ when using $\bm{\alpha} = \{\alpha_k\}_{k=0}^j$ in Lemma~\ref{lm:generaldualsol}. Note that $N(v)$ depends on $a, \beta, q$. Let $W_n(a, \beta, q) = \sup_{v \in [0, 1]} {N(v)}/{\sum_{i=1}^n (1-v)^i}$. Then, a quick verification shows that $(\bm{\alpha}, \zeta=W_n(a,\beta,q))$ is feasible to~\ref{form:DLPUB}. Using the weak duality (Lemma~\ref{lm:weakdual}), we can conclude that $R_n^C(\text{ALG})\leq W_n(a,\beta,q)$. Our task then becomes finding parameters $a,b,q$ such that the RHS of this inequality is as small as possible. The main challenge is that $W_n$ is a complex expression that we need to bound carefully. We will perform an asymptotic analyses with fixed parameters $a,b,q$ and then optimize the resulting expression.}


We start by evaluating $N(v)$ explicitly. To avoid lengthy expressions, we break the hiring duration into the ``main factor" $aq^l$ and the ``adjusting factor" $\rho_l$, and evaluate terms in $N(v)$ involving these two terms separately. The following result gives closed form expressions for $N(v)$ at these two factors respectively, both being a piecewise function over $[0, 1]$, with each piece defined over the interval $\left({\beta}/{q^{l+1}}, {\beta}/{q^l}\right]$. For notational simplicity, we first define some shorthands of recurring expressions. For $l = 0, \ldots, j$, define the following:
\begin{align*}
A_l &= \beta q^l(1+q) - \left( \frac{q}{a}+1\right)\beta + \frac{q}{a}, \quad B_l = q^{2l+1}, \\
\tilde{A}_l &= \beta(q-1)\left( \left( 2-\frac{a}{q} \right)l + \frac{(l-1)l}{2}\right) + \beta q \left(l + 2-\frac{a}{q} \right), \quad \tilde{B}_l = q^{l+1} \left(l + 2-\frac{a}{q} \right). 
\end{align*}
\begin{lemma} \label{lm:Nclosedform}
Define $M_{\text{fail}}: (\beta, 1] \rightarrow \mathbb{R}$ as $M_{\text{fail}}(v) = {q}/{a} \cdot (1-v)$. For $l = 0, \ldots, j$, define $M_{l}: \left( {\beta}/{q^{l+1}}, {\beta}/{q^l}\right] \rightarrow \mathbb{R}$ as $M_l(v) = A_l - B_l \cdot v$, and $R_{l}: \left( {\beta}/{q^{l+1}}, {\beta}/{q^l}\right] \rightarrow \mathbb{R}$ as $R_l(v) = \tilde{A}_l - \tilde{B}_l \cdot v$. Then, $N(v) = {(a\tilde{p})}/{(\beta E)} \cdot \left( M_{\text{fail}}(v) + \sum_{l=0}^j \left( M_l(v) + {1}/{a} \cdot R_l(v)\right)\right)$. 
\end{lemma}
The proof follows by simple calculations; we defer it to Appendix~\ref{app:general}. Lemma \ref{lm:Nclosedform} reveals that $N(v)$ takes on a piecewise structure, therefore to characterize $W_n(a, \beta, q)$, it suffices to study the maximum on each segment of the piecewise function $M_{\text{fail}}(v) + \sum_{l=0}^j \left( M_l(v) + {1}/{a} \cdot R_l(v)\right)$ divided by $\sum_{i=1}^n (1-v)^i$. Note that $\sum_{i=1}^n (1-v)^i={(1-v)(1-(1-v)^n)}/{v}$. Now ${M_{\text{fail}}(v)}/{\sum_{i=1}^n (1-v)^i}$ simplifies to ${qv}/{(a(1-(1-v)^n))}$, which is increasing in $v \in (\beta, 1]$, therefore, 
\begin{align*}
\sup_{v \in (\beta, 1]} \frac{M_{\text{fail}}(v)}{\sum_{i=1}^n (1-v)^i} &= \lim_{v \rightarrow 1} \frac{qv}{a(1-(1-v)^n)} = \frac{q}{a},
\end{align*}
which is at most ${1}/{2}$ by our restriction on the choice of $a$ and $q$. It remains to study the maximum of ${(M_l(v) + {1}/{a} \cdot R_l(v))}/{\sum_{i=1}^n (1-v)^i}$ on $v \in \left( {\beta}/{q^{l+1}}, {\beta}/{q^l}\right]$ for each $l$. As a result, define $\Gamma_l: \left( {\beta}/{q^{l+1}}, {\beta}/{q^l}\right] \rightarrow [0, \infty)$ as $\Gamma_l(v) = {\left( M_l(v) + {1}/{a} \cdot R_l(v) \right)v}/{((1-v)(1-(1-v)^n))}$, $W_n(a, \beta, q)$ is now equivalent to the following: 
\begin{align} \label{eq:gammal}
\sup_{l \in \{0, \ldots,j\}} \sup_{v \in (\beta/q^{l+1}, \beta/q^l]} \Gamma_l(v).
\end{align}
Even though the inner supremum of problem \eqref{eq:gammal} is a convex problem, it is not practical to solve via first order conditions because the exact expression of $\Gamma_l(v)$ cannot be simplified further. Another difficulty that precludes simple numerical analyses is that the number of optimization problems in~\eqref{eq:gammal} depends on $n$. Therefore, we perform further approximation and simplifications as $n \rightarrow \infty$.

\subsection{Upper Bounding the Dual Objective} \label{subsec:general3}
The main focus of this section is to obtain an upper bound on \eqref{eq:gammal}. We start by showing that we can impose $\Gamma_0(v) \leq 1$ with a careful choice of the parameter $\beta$. This is key in obtaining Lemma \ref{lm:maxbound} later, which saves us from the trouble of searching for the maximum over all $l$.  
\begin{lemma} \label{lm:betastar}
Let $\beta^* = {a}/{(q(a-1))} \cdot \left(2\sqrt{q} - 1 - {q}/{a}\right)$, then $\beta^*$ satisfies $0 < \beta^* < 1$. For $n \rightarrow \infty$, $\sup_{v \in \left( \beta / q, \beta\right]} \Gamma_0(v) \leq 1$ for all $0 < \beta \leq \beta^*$.
\end{lemma}
The proof of the lemma follows by simple verification deferred to Appendix~\ref{app:general}. We now simplify $\Gamma_l(v)$ for $l \geq 1$. A key observation is that as $l$ gets large, the term ${1}/{a} \cdot R_l(v)$ becomes negligible compared to $M_l(v)$. We formalize this in the next result, which upper bounds the ratio between ${1}/{a} \cdot R_l(v)$ and $M_l(v)$ with a sequence indexed in $l$ that converges to $0$ as $l \rightarrow \infty$. As a result, we can omit the more cumbersome term $R_l(v)$ in the limit. 
\begin{lemma} \label{lm:rhovsm}
Let $\varepsilon_l(v) = {1}/{a} \cdot R_l(v) / M_l(v)$. For $l = 1, \ldots, j$ and any positive integer $j$, $\varepsilon_l(v) \leq \tilde{\varepsilon}_l$, where $\tilde{\varepsilon}_l = \max \left\{{(l-1)l}/{(2aq^{l-1})}, {l(l+1)}/{(2aq^{l})}\right\}$. 
\end{lemma}
A derivative check shows that $\varepsilon_{l}$ is monotone in $v$, so the proof of Lemma~\ref{lm:rhovsm} follows by evaluating and bounding $\varepsilon_{l}(v)$ at the endpoints. We defer the calculation details to Appendix~\ref{app:general}. Expressing $\Gamma_l(v)$ as ${M_l(v)\cdot(1+\varepsilon_l(v)) v}/{((1-v)(1-(1-v)^n))}$, Lemma \ref{lm:rhovsm} justifies that it suffices to focus on ${M_l(v) \cdot v}/{((1-v)(1-(1-v)^n))}$ for large $l$, since $\varepsilon_l(v)$ gets negligibly small as $l$ increases. 

We now have the ingredients for the next result, which reduces the search space over $l$ in the optimization of \eqref{eq:gammal} to only two points, $l = j$ and $l$ where $\varepsilon_l(v)$ achieves its maximum. Although the latter can be challenging to characterize, it can be solved numerically: the upper bound $\tilde{\varepsilon}_l$ in Lemma~\ref{lm:rhovsm} suggests that $\varepsilon_l(v)$ has only one maximum, which typically happens at a small $l$, so an enumeration of $\varepsilon_l(v)$ at end points of $\left( {\beta}/{q^{l+1}}, {\beta}/{q^l}\right]$ for small $l$ would suffice to find this maximum. For the next result, we let $\varepsilon_l = \sup_{v \in (\beta / {q^{l+1}}, \beta / q^l]} \varepsilon_l(v)$, and define
\begin{align*}
\varepsilon^* = \max_{l \in \{0, \ldots, j/2-1\}} \varepsilon_l, \qquad  \varphi_0^* = \max_{x \in [0, \beta a / q]} \frac{1}{1-e^{-x}}\left(\frac{\beta q(1+q)}{a} - \frac{q^3}{a^2} \cdot x\right) x.
\end{align*}
\begin{lemma} \label{lm:maxbound}
For any $0 < \beta \leq \beta^*$, $\lim_{n \rightarrow \infty} \sup_{l \in \{0, \ldots,j\}} \sup_{v \in (\beta/q^{l+1}, \beta/q^l]} \Gamma_l(v) \leq \max \{1 + \varepsilon^*, \varphi_0^*\}$. 
\end{lemma}
The proof of this result is technical and we provide it in detail in Appendix~\ref{app:general}; however, the idea is simple. We divide $l \in \{0, \ldots,j\}$ into two regimes, $l \leq {j}/{2} - 1$ and $l \geq {j}/{2}$: 
\begin{align*}
\sup_{l \in \{0, \ldots,j\}} \sup_{v \in (\beta/q^{l+1}, \beta/q^l]} \Gamma_l(v) = \max\left\{\sup_{l \in \{0, \ldots, j/2 - 1\}} \sup_{v \in (\beta / {q^{l+1}}, \beta / q^l]} \Gamma_l(v),  \sup_{l \in \{j/2, \ldots,j\}} \sup_{v \in (\beta/q^{l+1}, \beta/q^l]} \Gamma_l(v) \right\}. 
\end{align*}
Then, to prove the lemma is enough to show that, when $n \rightarrow \infty$, the maximum in the RHS above is upper bounded by $\max \{1 + \varepsilon^*, \varphi_0^*\}$. We accomplish this by upper bounding the supremums in each of the two regimes.

Now $W_n(a, \beta, q)$ is simplified to the best of our ability, and we are ready to present the proof of Theorem~\ref{thm:general}, in which we upper bound $\varphi_0^*$ with an expression in terms of $a, q$ and $\beta$ so that we can optimize to determine the numerical values of the parameters. Then, since both $\varphi_0^*$ and $\varepsilon^*$ do not have closed form solutions, we obtain these numerically after determining the parameters. 
\begin{proof}[Proof of Theorem~\ref{thm:general}]
Recall that we have $R_n^{\text{C}}(\text{ALG}) \leq W_n(a, \beta, q)$. We have seen in Section~\ref{subsec:general2} that $W_n(a, \beta, q)=({(a\tilde{p})}/{(\beta E)})\cdot \sup_{l \in \{0, \ldots,j\}} \sup_{v \in (\beta/q^{l+1}, \beta/q^l]} \Gamma_l(v)$. Then, using Lemma~\ref{lm:maxbound},
\begin{align*}
\lim_{n \rightarrow \infty} R_n^{\text{C}}(\text{ALG}) &\leq \frac{a\tilde{p}}{\beta E} \lim_{n \rightarrow \infty} \sup_{l \in \{0, \ldots,j\}} \sup_{v \in (\beta/q^{l+1}, \beta/q^l]} \Gamma_l(v) \leq  \frac{a\tilde{p}}{\beta E} \cdot \max \{\varphi_0^*, 1 + \varepsilon^*\}.
\end{align*}
{ We will use the RHS of this inequality to provide our guarantees. However, optimizing this RHS itself is hard, even numerically. Instead, we provide a relaxation to $\varphi_0^*$ that we can optimize and obtain parameters $a,\beta$ and $q$. Then, we use these parameters in the RHS expression above to conclude the proof of Theorem~\ref{thm:general}.}

We find an upper bound to $\varphi_0^*$ using the following claim.
\begin{claim} \label{claim:largejub}
Define $g_k(x) = (x+1)\left({\beta q(1+q)}/{(aq^k)} - {q^3}/{(a^2 q^{2k})} \cdot x\right)$. Then for all $x \geq 0$, $g_k(x) \geq {1}/{(1-e^{-x})} \cdot \left({\beta q(1+q)}/{(aq^k)} - {q^3}/{(a^2 q^{2k})} \cdot x\right) x$. Moreover, $g_k(x)$ has the unique maximum value $\varphi_k = {1}/{(4q)} \cdot \left( \beta(1+q) + {q^{2-k}}/{a}\right)^2$ for all integer $k \geq 0$.
\end{claim}
{  The claim implies that $\varphi_0^* \leq \max_{x\in [0,\beta a/q]} g_0(x) \leq \varphi_0$. Now, we solve $\inf_{a, \beta, q} {a\tilde{p}}/{(\beta E)} \cdot \varphi_0$ with the restriction that $\beta \leq \beta^*$ that is justified by Lemma \ref{lm:betastar}. Numerically, the minimum is achieved at $a = 3.6$, $q = 1.49$, and $\beta = 0.954$, which gives ${a\tilde{p}}/{(\beta E)} = 3.568$. }


With the above choice of parameters, we numerically obtain $\varphi_0^* = 1.165$. For $\varepsilon^*$, we compute $\varepsilon_l(v)$ at the end points ${\beta}/{q^l}$, and by enumerating the results for small $l$, we obtain that the maximum of $\varepsilon_l(v)$ happens at $l = 6$, with $\varepsilon^* = 0.171$. Substituting the numerical results and the choice of parameters back $\varphi_0^*$ and $\varepsilon^*$, we have ${a\tilde{p}}/{(\beta E)} \cdot \max \{\varphi_0^*, 1 + \varepsilon^*\} \leq 3.568 \cdot (1+0.171)=4.179$, which concludes the proof of Theorem~\ref{thm:general}. 
\end{proof}

\section{Lower Bound for OSCC and Beyond i.i.d. Values} \label{section:noniid}
In this section, we close our discussion on online contract selection with two results: a lower bound for the competitive ratio of OSCC via OSDC, and an impossibility result for non-i.i.d.\ models.  

{  We first extend the models to the setting where values are independent but drawn from possibly different distributions, i.e., the general (independent) setting. The input to the general OSDC and OSCC consists of a sequence of $n$ CDFs $F_1, \ldots, F_n$. For a model $\mathcal{M} \in \{C, D\}$, at time $i$, an online algorithm $\text{ALG}^{\mathcal{M}}$ observes the random variable $X_i$ following $F_i$, and proceeds as in the i.i.d.\ setting. The expected offline optimal cost is given by $\OPT_n(\{F_i\}_{i=1}^n)= \E_{X_1 \sim F_1, \ldots, X_n \sim F_n}\!\left[ \sum_{i=1}^n \min\{X_1,\ldots,X_i\} \right].$ The competitive ratio of $\text{ALG}^{\mathcal{M}}$ in instances of size $n$ and CDFs $F_1, \ldots, F_n \in \mathcal{F}$ is given by $
R_n^{\mathcal{M}}(\text{ALG}^{\mathcal{M}},\{F_i\}_{i=1}^n) ={\ALG^{\mathcal{M}}_n(\{F_i\}_{i=1}^n)}/{\OPT_n(\{F_i\}_{i=1}^n)}$.

We remark that the reduction in Section~\ref{subsec:lb1} from OSDC to $n$ single-selection problems holds in the general setting, and the optimal online expected cost can be computed by a sequence of $n$ DPs. Although the $i$-th problem is now a cost minimization single-selection problem over the sequence $X_1,\ldots,X_i$. Let $D_{i, j}$ be optimal expected cost when $j$ observations remain possible. Then, $D_{i, 1} = \E_{X_i \sim F_i}[X_i]$ and $D_{i, j} = \E_{X_{i-j+1} \sim F_{i-j+1}}[\min\{X_{i-j+1}, D_{i, j-1}\}]$ (the minimum of the current observed value and the expected value at the next step) for all $1 < j \leq i$. As in Lemma~\ref{lm:lbobj}, since $D_{i, i}$ gives the expected optimal cost of the $i$-th problem, the expected cost incurred by the optimal algorithm $\text{ALG}^*$ for the general OSDC is given by $\ALG_n^*(\{F_i\}_{i=1}^n) = \sum_{i=1}^n D_{i, i}$.

Our first result shows that the optimal algorithm for OSDC always achieves a lower expected cost than any online algorithms for OSCC on the same instance $F_1,\ldots,F_n$.

\begin{lemma} \label{lm:ALGLB}
For any CDFs $F_1, \ldots, F_n $, any algorithm $\text{ALG}^{\text{C}}$ for OSCC, and the optimal algorithm $\text{ALG}^*$ for OSDC, it holds that $\ALG_n(\{F_i\}_{i=1}^n) \geq \ALG^*_n(\{F_i\}_{i=1}^n)$.
\end{lemma}
\begin{proof}
Observe that any algorithm $\text{ALG}^{\text{C}}$ on OSCC can be transformed to a feasible algorithm $\text{ALG}^{\text{D}}$ on OSDC by ``queuing'' the concurrent contracts such that they are back-to-back instead of overlapping. The result then follows since any feasible algorithm for OSDC incurs at least the expected cost of the optimal $\text{ALG}^*$.
\end{proof}
Since OSDC and OSCC share the same offline optimal algorithm, Lemma~\ref{lm:ALGLB} immediately implies a lower bound on the competitive ratio of the i.i.d.\ OSCC via $R_n^* = \zeta_n$ given in Theorem~\ref{thm:lowerbound}. 
\begin{corollary} \label{cr:lbrelax}
For any positive integer $n$, under the i.i.d.\ setting, no online algorithm $\text{ALG}^{\text{C}}$ on OSCC achieves $R_n^{\text{C}}(\text{ALG}^{\text{C}}) < \zeta_n$, where $\zeta_n$ is defined in Theorem~\ref{thm:lowerbound}. 
\end{corollary}
Our second result extends the discussion beyond the i.i.d.\ setting, showing that we lose the constant competitive ratio for OSDC if the identically distributed assumption is removed. Lemma~\ref{lm:ALGLB} implies that in non-i.i.d.\ OSCC, no online algorithms achieve a constant competitive ratio as well. 
\begin{proposition} \label{prop:noniid}
When $X_i$ are independent but non-identically distributed, i.e., $X_i \sim F_i$ where $F_i \in \mathcal{F}$ for all $i$, no algorithm for OSDC achieves a constant competitive ratio. 
\end{proposition}
\begin{proof}
Consider the following instance: $X_1=1$ and for $i=2,\ldots,n$ let $X_i = 2^{i}\cdot Z_i$, where $Z_i$ is a $\{0,1\}-$Bernoulli random variable with parameter $1/2$ (i.e., $\mathbb{P}[Z_i=1]=1/2$). 
For $\text{ALG}^*$, $D_{1, 1} = \E[X_1] = 1$ and $D_{i, 1} = \E[X_i] = 2^{i-1}$ for all $i = 2, \ldots, n$. Then since the maximum realization of $X_{i-1}$ is $2^{i-1} = D_{i, 1}$, we have $D_{i, 2} = \E[\min\{X_{i-1}, D_{i, 1}\}] = \E[X_{i-1}] = 2^{i-2}$. By repeating this calculation iteratively, we have $D_{i, i-1} = \E[X_2] = 2 > 1 = X_1$. Therefore, $D_{i, i} = \E[\min\{X_{1}, D_{i, i-1}\}] = \E[X_{1}] = 1$ for all $i = 1, \ldots, n$. It follows that $\ALG_n^* = \sum_{i=1}^n D_{i, i} = n$. 

On the other hand, at each time $i$, we have $\mathbb{E}[\min\{X_1,\ldots,X_i\}] = 1\cdot 2^{i-1}$. Thus, 
the offline optimal solution incurs a cost of $1$ if and only if no cost has realized $0$ up to time $i$, otherwise it incurs $0$ cost. Summing up the expected cost at time $i$, we get 
$\OPT_n = 1 + 1/2 \cdot 1 + \left( 1/2 \right)^2 \cdot 1 + \ldots + \left(1/2\right)^{n-1} \cdot 1 = 2\left( 1 - \left(1/2\right)^n\right)$. Then, the ratio ${\ALG_n^*}/{\OPT_n} = {n}/{(2(1-(1/2)^n))} = \Theta(n)$ for this instance implies that no algorithm for OSDC achieves a constant competitive ratio. 
\end{proof}
} 

\section{Final Remarks}

In this work, we study OSDC and OSCC in the class of online contract selection problems with continuous contractual coverages. We develop LP machinery to analyze the competitive ratios of online algorithms for these two problems, solving the competitive ratio of OSDC exactly and improving upon existing competitive ratio bounds for OSCC by applying the LP framework on a general class of algorithms. Our tight result for OSDC in comparison with the upper bound analysis for OSCC highlights the gap between settings with and without control over the contract starting times. We note that we have to let $n \rightarrow \infty$ to simplify many involved expressions in Section~\ref{section:general}, which has posed the main obstacle in obtaining a competitive ratio upper bound for every integer $n$. As a result, we leave characterizing a competitive ratio upper bound for finite $n$ open.  


{ Proposition~\ref{prop:noniid} rules out constant competitive ratios for both contract selection problems under non-identically distributed random values. An interesting open question is whether an intermediate model where values are drawn independently from different distributions but arrive in random order admits a finite competitive ratio. Our techniques rely heavily on the i.i.d.\ assumption, so analyzing this setting will likely require fundamentally new ideas. Another interesting direction is to study free-order models, in which the online algorithm can adaptively choose the order in which values are observed, and determine whether such models admit constant competitive ratios under an appropriate notion of offline optimal cost.}



\small
\bibliographystyle{plainnat}
\bibliography{ref}

@ARTICLE{Disser2019,
  author={Yann Disser and John Fearnley and Martin Gairing and Oliver Göbel and Max Klimm and Daniel Schmand and Alexander Skopalik and Andreas Tönnis},
  journal={Mathematics of Operations Research}, 
  title={Hiring Secretaries over Time: The Benefit of
Concurrent Employment}, 
  year={2019},
  volume={45},
  number={1},
  pages={323-352}}

@misc{PS2024,
      title={Optimal Guarantees for Online Selection Over Time}, 
      author={Sebastian Perez-Salazar and Victor Verdugo},
      year={2024},
      eprint={2408.11224},
      archivePrefix={arXiv},
      primaryClass={cs.GT},
      url={https://arxiv.org/abs/2408.11224}, 
}

@ARTICLE{Manne1960,
  author={Alan S. Manne},
  journal={Management Science}, 
  title={Linear Programming and Sequential Decisions}, 
  year={1960},
  volume={6},
  number={3},
  pages={259-267}}

@inproceedings{Liu2021,
  title        = {Variable Decomposition for Prophet Inequalities and Optimal Ordering},
  author       = {Allen Liu and Renato Paes Leme and Martin Pál and Jon Schneider and Balasubramanian Sivan},
  year         = 2021,
  month        = {July},
  booktitle    = {EC'21: Proceedings of the 22nd ACM Conference on Economics and Computation},
  publisher    = {Association for Computing Machinery},
  pages        = {692},}

@article{abels2025prophet,
  title={Prophet inequalities over time},
  author={Abels, Andreas and Pitschmann, Elias and Schmand, Daniel},
  journal={ACM Transactions on Economics and Computation},
  volume={13},
  number={4},
  pages={1--36},
  year={2025},
  publisher={ACM New York, NY}
}

@article{correa2021posted,
  title={Posted price mechanisms and optimal threshold strategies for random arrivals},
  author={Correa, Jos{\'e} and Foncea, Patricio and Hoeksma, Ruben and Oosterwijk, Tim and Vredeveld, Tjark},
  journal={Mathematics of operations research},
  volume={46},
  number={4},
  pages={1452--1478},
  year={2021},
  publisher={INFORMS}
}

@inproceedings{hajiaghayi2007automated,
  title={Automated online mechanism design and prophet inequalities},
  author={Hajiaghayi, Mohammad Taghi and Kleinberg, Robert and Sandholm, Tuomas},
  booktitle={AAAI},
  volume={7},
  pages={58--65},
  year={2007}
}

@article{perez2025robust,
  title={Robust online selection with uncertain offer acceptance},
  author={Perez-Salazar, Sebastian and Singh, Mohit and Toriello, Alejandro},
  journal={Mathematics of Operations Research},
  volume={50},
  number={3},
  pages={2226--2260},
  year={2025},
  publisher={INFORMS}
}

@article{epstein2024selection,
  title={Selection and ordering policies for hiring pipelines via linear programming},
  author={Epstein, Boris and Ma, Will},
  journal={Operations Research},
  volume={72},
  number={5},
  pages={2000--2013},
  year={2024},
  publisher={INFORMS}
}

@inproceedings{babaioff2007knapsack,
  title={A knapsack secretary problem with applications},
  author={Babaioff, Moshe and Immorlica, Nicole and Kempe, David and Kleinberg, Robert},
  booktitle={International Workshop on Approximation Algorithms for Combinatorial Optimization},
  pages={16--28},
  year={2007},
  organization={Springer}
}

@inproceedings{kesselheim2014primal,
  title={Primal beats dual on online packing LPs in the random-order model},
  author={Kesselheim, Thomas and T{\"o}nnis, Andreas and Radke, Klaus and V{\"o}cking, Berthold},
  booktitle={Proceedings of the forty-sixth annual ACM symposium on Theory of computing},
  pages={303--312},
  year={2014}
}

@inproceedings{karp1990optimal,
  title={An optimal algorithm for on-line bipartite matching},
  author={Karp, Richard M. and Vazirani, Umesh V. and Vazirani, Vijay V.},
  booktitle={Proceedings of the twenty-second annual ACM symposium on Theory of computing},
  pages={352--358},
  year={1990}
}

@inproceedings{feldman2009online,
  title={Online stochastic matching: Beating 1-1/e},
  author={Feldman, Jon and Mehta, Aranyak and Mirrokni, Vahab and Muthukrishnan, Shan},
  booktitle={2009 50th Annual IEEE Symposium on Foundations of Computer Science},
  pages={117--126},
  year={2009},
  organization={IEEE}
}

@inproceedings{berzack2025dynamic,
  title={Dynamic Rental Games with Stagewise Individual Rationality},
  author={Berzack, Batya and Oshman, Rotem and Talgam-Cohen, Inbal},
  booktitle={Proceedings of the 26th ACM Conference on Economics and Computation},
  pages={785--785},
  year={2025}
}

@article{perez2025iid,
author = {Perez-Salazar, Sebastian and Singh, Mohit and Toriello, Alejandro},
title = {The I.I.D. Prophet Inequality with Limited Flexibility},
journal = {Mathematics of Operations Research},
volume = {51},
number = {1},
pages = {218-254},
year = {2026},
  publisher={INFORMS}
}

@article{Brustle2025splitting,
  title={Splitting Guarantees for Prophet Inequalities via Nonlinear Systems},
  author={Johannes Brustle and Sebastian Perez-Salazar and Victor Verdugo},
  journal={Mathematics of Operations Research},
  year={2025},
  publisher={INFORMS}
}

@BOOK{Kolmogorov1975analysis,
  TITLE = {Introductory Real Analysis},
  AUTHOR = {Andrey Nikolaevich Kolmogorov and Sergei Vasilyevich Fomin},
  YEAR = {1975}, 
  PUBLISHER = {Dover Publications},
}

@article{Krengel1977,
  title={Semiamarts and finite values},
  author={Ulrich Krengel and Louis Sucheston},
  journal={Bulletin of the American Mathematical Society},
  year={1977},
  volume={83},
  number={4},
  pages={745--747}
}

@article{SamuelCahn1984,
  title={Comparison of Threshold Stop Rules and Maximum for Independent Nonnegative Random Variables},
  author={Ester Samuel-Cahn},
  journal={The Annals of Probability},
  year={1984},
  volume={12},
  number={4},
  pages={1213--1216}
}

@inproceedings{Alaei2012onlineprophet,
  title={Online prophet-inequality matching with applications to ad allocation},
  author={Saeed Alaei and MohammadTaghi Hajiaghayi and Vahid Liaghat},
  booktitle={Proceedings of the 13th ACM Conference on Electronic Commerce},
  pages={18--35},
  year={2012},
  publisher = {Association for Computing Machinery}
}

@article{Esfandiari2017prophetsecretary,
author = {Esfandiari, Hossein and Hajiaghayi, MohammadTaghi and Liaghat, Vahid and Monemizadeh, Morteza},
title = {Prophet Secretary},
journal = {SIAM Journal on Discrete Mathematics},
volume = {31},
number = {3},
pages = {1685--1701},
year = {2017}
}

@inproceedings{Fiat2015temp,
  title={The Temp Secretary Problem},
  author={Amos Fiat and Ilia Gorelik and Haim Kaplan and Slava Novgorodov },
  booktitle={Bansal, N., Finocchi, I. (eds) Algorithms - ESA 2015.Lecture Notes in Computer Science},
  volume = {9294},
  pages={631--642},
  year={2015},
  organization = {Springer}
}

@inproceedings{Livanos2024minimization,
  title={Minimization is Harder in the Prophet World},
  author={Vasilis Livanos and Ruta Mehta},
  booktitle={Proceedings of the 2024 Annual ACM-SIAM Symposium on Discrete Algorithms},
  pages={424--461},
  year={2024}
}

@inproceedings{Chawla2010postedprice,
author = {Chawla, Shuchi and Hartline, Jason D. and Malec, David L. and Sivan, Balasubramanian},
title = {Multi-parameter mechanism design and sequential posted pricing},
year = {2010},
publisher = {Association for Computing Machinery},
booktitle = {Proceedings of the Forty-Second ACM Symposium on Theory of Computing},
pages = {311–320},
}

@article{Ehsani2024combinatorialauction,
author = {Soheil Ehsani and MohammadTaghi Hajiaghayi and Thomas Kesselheim and Sahil Singla},
title = {Prophet Secretary for Combinatorial Auctions and Matroids},
journal = {SIAM Journal on Computing},
volume = {53},
number = {6},
pages = {1641--1662},
year = {2024}
}

@inproceedings{Kleinberg2012matroid,
author = {Kleinberg, Robert and Weinberg, Seth Matthew},
title = {Matroid prophet inequalities},
year = {2012},
publisher = {Association for Computing Machinery},
booktitle = {Proceedings of the Forty-Fourth Annual ACM Symposium on Theory of Computing},
pages = {123–136},
}

@article{Qin2024,
author = {Qin, Junjie and Vardi, Shai and Wierman, Adam},
title = {Minimization Fractional Prophet Inequalities for Sequential Procurement},
journal = {Mathematics of Operations Research},
volume = {49},
number = {2},
pages = {928-947},
year = {2024},
publisher={INFORMS}
}

@article{HillKertz1982iid,
 author = {T. P. Hill and Robert P. Kertz},
 journal = {The Annals of Probability},
 number = {2},
 pages = {336--345},
 publisher = {Institute of Mathematical Statistics},
 title = {Comparisons of Stop Rule and Supremum Expectations of I.I.D. Random Variables},
 volume = {10},
 year = {1982}
}

@article{Kertz1986iid,
author = {Kertz, Robert P},
title = {Stop rule and supremum expectations of i.i.d. random variables: a complete comparison by conjugate duality},
year = {1986},
volume = {19},
number = {1},
journal = {Journal of Multivariate Analysis},
pages = {88–112},
}

@article{Jiang2025,
author = {Jiang, Jiashuo and Ma, Will and Zhang, Jiawei},
title = {Tight Guarantees for Multiunit Prophet Inequalities and Online Stochastic Knapsack},
journal = {Operations Research},
volume = {73},
number = {3},
pages = {1703-1721},
year = {2025},
publisher={INFORMS}
}

@article{Alaei2014,
author = {Alaei, Saeed},
title = {Bayesian Combinatorial Auctions: Expanding Single Buyer Mechanisms to Many Buyers},
journal = {SIAM Journal on Computing},
volume = {43},
number = {2},
pages = {930-972},
year = {2014},
}

@article{Buchbinder2014secretarylp,
author = {Buchbinder, Niv and Jain, Kamal and Singh, Mohit},
title = {Secretary Problems via Linear Programming},
journal = {Mathematics of Operations Research},
volume = {39},
number = {1},
pages = {190-206},
year = {2014},
publisher = {INFORMS}
}

@inproceedings{Rubinstein2019OptimalSP,
  title={Optimal Single-Choice Prophet Inequalities from Samples},
  author={Aviad Rubinstein and Jack Z. Wang and S. Matthew Weinberg},
  booktitle={Information Technology Convergence and Services},
  year={2019},
}

@inproceedings{Correa2019unknown,
author = {Correa, Jos\'{e} and D\"{u}tting, Paul and Fischer, Felix and Schewior, Kevin},
title = {Prophet Inequalities for I.I.D. Random Variables from an Unknown Distribution},
year = {2019},
publisher = {Association for Computing Machinery},
booktitle = {Proceedings of the 2019 ACM Conference on Economics and Computation},
pages = {3–17},
}

@inproceedings{Abolhassani2017beatingiid,
author = {Abolhassani, Melika and Ehsani, Soheil and Esfandiari, Hossein and HajiAghayi, MohammadTaghi and Kleinberg, Robert and Lucier, Brendan},
title = {Beating 1-1/e for ordered prophets},
year = {2017},
publisher = {Association for Computing Machinery},
booktitle = {Proceedings of the 49th Annual ACM SIGACT Symposium on Theory of Computing},
pages = {61–71},
}

@article{Dutting2020stochastic,
author = {D\"{u}tting, Paul and Feldman, Michal and Kesselheim, Thomas and Lucier, Brendan},
title = {Prophet Inequalities Made Easy: Stochastic Optimization by Pricing Nonstochastic Inputs},
journal = {SIAM Journal on Computing},
volume = {49},
number = {3},
pages = {540-582},
year = {2020},
}

@inproceedings{Rubinstein2017,
author = {Rubinstein, Aviad and Singla, Sahil},
title = {Combinatorial prophet inequalities},
year = {2017},
publisher = {Society for Industrial and Applied Mathematics},
booktitle = {Proceedings of the Twenty-Eighth Annual ACM-SIAM Symposium on Discrete Algorithms},
pages = {1671–1687},
}

@InProceedings{Kesselheim2016improvedtemp,
  author =	{Kesselheim, Thomas and T\"{o}nnis, Andreas},
  title =	{Think Eternally: Improved Algorithms for the Temp Secretary Problem and Extensions},
  booktitle =	{24th Annual European Symposium on Algorithms (ESA 2016)},
  pages =	{54:1--54:17},
  year =	{2016},
  volume = {57},
  publisher =	{Schloss Dagstuhl -- Leibniz-Zentrum f{\"u}r Informatik},
}

@article{Feng2024reusable,
author = {Feng, Yiding and Niazadeh, Rad and Saberi, Amin},
title = {Technical Note—Near-Optimal Bayesian Online Assortment of Reusable Resources},
journal = {Operations Research},
volume = {72},
number = {5},
pages = {1861-1873},
year = {2024},
publisher = {INFORMS}
}

@article{Faw2022reusable,
author = {Faw, Matthew and Papadigenopoulos, Orestis and Caramanis, Constantine and Shakkottai, Sanjay},
title = {Learning To Maximize Welfare with a Reusable Resource},
year = {2022},
publisher = {Association for Computing Machinery},
volume = {6},
number = {2},
journal = {Proceedings of the ACM on Measurement and Analysis of Computing Systems},
pages = {1--30}
}

@article{Goyal2023matching,
author = {Goyal, Vineet and Udwani, Rajan},
title = {Online Matching with Stochastic Rewards: Optimal Competitive Ratio via Path-Based Formulation},
journal = {Operations Research},
volume = {71},
number = {2},
pages = {563-580},
year = {2023},
publisher = {INFORMS}
}

@article{Mehta2007onlinematching,
author = {Mehta, Aranyak and Saberi, Amin and Vazirani, Umesh and Vazirani, Vijay},
title = {AdWords and generalized online matching},
year = {2007},
publisher = {Association for Computing Machinery},
volume = {54},
number = {5},
journal = {Journal of the ACM},
pages = {22–es},
}

@inproceedings{Chan2015secretaryLP,
author = {Chan, T-H. Hubert and Chen, Fei and Jiang, Shaofeng H.-C.},
title = {Revealing optimal thresholds for generalized secretary problem via continuous LP: impacts on online K-item auction and bipartite K-matching with random arrival order},
year = {2015},
publisher = {Society for Industrial and Applied Mathematics},
booktitle = {Proceedings of the Twenty-Sixth Annual ACM-SIAM Symposium on Discrete Algorithms},
pages = {1169–1188},
}

@article{Correa2024sample,
author = {Correa, Jos\'{e} and Cristi, Andr\'{e}s and Epstein, Boris and Soto, Jos\'{e} A.},
title = {Sample-Driven Optimal Stopping: From the Secretary Problem to the i.i.d. Prophet Inequality},
journal = {Mathematics of Operations Research},
volume = {49},
number = {1},
pages = {441-475},
year = {2024},
publisher = {INFORMS}
}

@InProceedings{Gupta2013stochastic,
author="Gupta, Anupam and Nagarajan, Viswanath",
editor="Goemans, Michel and Correa, Jos{\'e}",
title="A Stochastic Probing Problem with Applications",
booktitle="Integer Programming and Combinatorial Optimization",
year="2013",
publisher="Springer Berlin Heidelberg",
pages="205--216",
}

@article{Brustle2024complexity,
author = {Brustle, Johannes and Correa, Jos\'{e} and Duetting, Paul and Verdugo, Victor},
title = {The Competition Complexity of Dynamic Pricing},
journal = {Mathematics of Operations Research},
volume = {49},
number = {3},
pages = {1986-2008},
year = {2024},
publisher = {INFORMS}
}

@book{Havil2003,
 author = {Julian Havil},
 publisher = {Princeton University Press},
 title = {Gamma: Exploring Euler's Constant},
 year = {2003}
}

\small
\appendix

\section{Missing Proofs from Section~\ref{section:assumption}} \label{app:assumption}

\begin{proof}[Proof of Proposition \ref{prop:assumption}]
We show that given an algorithm $\text{ALG}$ such that $\ALG_n(\hat{F}) \leq \beta \cdot \OPT_n(\hat{F})$ for any $\hat{F}$ strictly increasing and continuously differentiable, there exists an algorithm $\text{ALG}'$ such that $\ALG'_n(F) \leq \beta(1 + \frac{2\varepsilon}{n}) \cdot \OPT_n(F)$ for any $F \in \mathcal{F}$ and $\varepsilon > 0$. The statement then follows by taking $\varepsilon$ arbitrarily small.  

For any $F \in \mathcal{F}$, with the standard smoothing argument which results in a loss that can be made arbitrarily small (see \cite{Liu2021}), we assume $F$ is continuously differentiable. Since rescaling does not impact ratios, we also assume that $\OPT_n(F) = 1$. Suppose $\{X_i\}_{i=1}^n$ are i.i.d.\ random variables following $F$. For any $\varepsilon > 0$, consider random variables $\{\hat{X}_i\}_{i=1}^n$ that with probability $\frac{1}{1 + \varepsilon / n^3}$ follow $F$ and with probability $\frac{\varepsilon / n^3}{1 + \varepsilon / n^3}$ follow $\text{Exp}(1)$. The distribution $\hat{F}$ of $\{\hat{X}_i\}_{i=1}^n$ is given by $\hat{F}(x) = \frac{1}{1 + {\varepsilon}/{n^3}} F(x) + \frac{{\varepsilon}/{n^3}}{1 + {\varepsilon}/{n^3}}(1 - e^{-x})$. Then since $\hat{F}'(x) = \frac{1}{1 + \varepsilon / n^3} F'(x) + \frac{\varepsilon / n^3}{1 + \varepsilon / n^3} \cdot e^{-x} > 0$, $\hat{F}$ is strictly increasing, which allows us to apply $\text{ALG}$ and get $\ALG_n(\hat{F}) \leq \beta \cdot \OPT_n(\hat{F})$. Note that $1 - \hat{F}(x) = \frac{1}{1 + \varepsilon / n^3} (1 - F(x)) + \frac{\varepsilon / n^3}{1 + \varepsilon / n^3} \cdot e^{-x}$. From here, we have the following upper bound on $\OPT_n(\hat{F}) = \sum_{i=1}^n \E \left[\min_{l \in \{1, \ldots, i\}} \{\hat{X}_l\} \right] = \sum_{i=1}^n \int_0^{\infty} (1 - \hat{F}(x))^i \, \mathrm{d}x$, 
\begin{align*}
\OPT_n(\hat{F}) &= \sum_{i=1}^n \int_0^{\infty} \left(\frac{1}{1 + \frac{\varepsilon}{n^3}} (1 - F(x)) + \frac{\frac{\varepsilon}{n^3}}{1 + \frac{\varepsilon}{n^3}} \cdot e^{-x} \right)^i \, \mathrm{d}x \\
&\leq \sum_{i=1}^n \int_0^{\infty} \frac{1}{\left( 1 + \frac{\varepsilon}{n^3}\right)^i} \left( (1-F(x))^i + i\left( 1 + \frac{\varepsilon}{n^3}\right)^{i-1} \cdot \frac{\varepsilon}{n^3} \cdot e^{-x} \right)\, \mathrm{d}x \\
&\leq \sum_{i=1}^n \frac{1}{\left( 1 + \frac{\varepsilon}{n^3}\right)^i} \int_0^{\infty} (1-F(x))^i \, \mathrm{d}x + \sum_{i=1}^n \int_0^{\infty} \frac{i \cdot \varepsilon \cdot e^{-x}}{n^3} \, \mathrm{d}x \\
&\leq  \frac{1}{1 + \frac{\varepsilon}{n^3}} \OPT_n(F) + \frac{\varepsilon}{n^3} \int_0^{\infty} e^{-x} \, \mathrm{d}x \cdot \sum_{i=1}^n i = \frac{1}{1 + \frac{\varepsilon}{n^3}} \OPT_n(F) + \frac{\varepsilon}{n^3} \cdot \frac{n(n+1)}{2}, 
\end{align*} 
where the first inequality comes from Claim \ref{claim:inequality}, the second inequality is due to $\frac{i \cdot \varepsilon \cdot e^{-x}}{n^3(1 + \varepsilon / n^3)} \leq \frac{i \cdot \varepsilon \cdot e^{-x}}{n^3}$, and the last inequality is due to $\frac{1}{(1 + \varepsilon / n^3)^i} \leq \frac{1}{1 + \varepsilon / n^3}$ for all $i \geq 1$. 

\begin{claim} \label{claim:inequality}
$\left( 1 - F(x) + \frac{\varepsilon \cdot e^{-x}}{n^3}\right)^i \leq (1 - F(x))^i + i\left( 1 + \frac{\varepsilon}{n^3}\right)^{i-1} \cdot \frac{\varepsilon \cdot e^{-x}}{n^3}$ for $i \in \{1, \ldots, n\}$, $\varepsilon > 0$, and $x \in [0, \infty)$. 
\end{claim}

Now we construct a new algorithm $\text{ALG}'$ from $\text{ALG}$ that runs on an instance for OSDC or OSCC with $F$ and $n$ time steps as follows: at time $i$, with probability $\frac{1}{1 + \varepsilon / n^3}$ , observe the realization of $X_i$ following $F$; otherwise, draw a sample from $\text{Exp}(1)$. As a result, the ``observed value" now becomes $\{\hat{X}\}_{i=1}^n$. Run $\text{ALG}$ on $\{\hat{X}\}_{i=1}^n$. If the realization of $\hat{X}_i$ comes from $F$, follow the decision of $\text{ALG}$; if the realization comes from $\text{Exp}(1)$, accept $\hat{X}_i$ for all remaining time steps and pay the corresponding cost. In the worst case, this contract would last for $n$ time steps, therefore, $\ALG'_n(F)$ have the following upper bound, 
\begin{align*}
\ALG'_n(F) &\leq \ALG_n(\hat{F}) +  \mathbb{P}[\exists \hat{X}_i \sim \text{Exp}(1)] \cdot 1 \cdot n \\
& \leq \beta \cdot \OPT_n(\hat{F}) + \left(1 - \left(1 - \frac{\frac{\varepsilon}{n^3}}{1 + \frac{\varepsilon}{n^3}} \right)^n\right) \cdot n \\
& \leq \beta \cdot \left( \frac{1}{1 + \frac{\varepsilon}{n^3}} \OPT_n(F) + \frac{\varepsilon}{n^3} \cdot \frac{n(n+1)}{2}\right) + \frac{n \cdot \frac{\varepsilon}{n^3}}{1 + \frac{\varepsilon}{n^3}} \cdot n \\
&\leq \beta \cdot \frac{1}{1 + \frac{\varepsilon}{n^3}} \OPT_n(F) + \beta \cdot \frac{2\varepsilon}{n} \leq \beta\left( 1 + \frac{2\varepsilon}{n}\right) \cdot \OPT_n(F),
\end{align*}
where in the third inequality, we applied the binomial expansion and got $\left(1 - \frac{\varepsilon / n^3}{1 + \varepsilon / n^3} \right)^n \geq 1 - \frac{n \cdot \varepsilon / n^3}{1 + \varepsilon / n^3}$, in the second last inequality, we used $\beta \geq 1$, $n+1 \leq 2n$ for $n \geq 1$ and $\frac{1}{1 + \varepsilon / n^3} \leq 1$, and in the last inequality, we used $\frac{1}{\left(1 + \varepsilon / n^3\right)^2} \leq 1$ and $\OPT_n(F) = 1$. This concludes the proof. 
\end{proof}

\begin{proof}[Proof of Claim \ref{claim:inequality}]
Define $\phi(t) = (1 - F(x) + t)^i$. Since $\phi$ is convex, we have $\phi\left( \frac{\varepsilon}{n^3} \cdot e^{-x}\right) \leq \phi(0) + \phi'\left( \frac{\varepsilon}{n^3} \cdot e^{-x}\right) \cdot \frac{\varepsilon}{n^3} \cdot e^{-x}$. Applying $1 - F(x) \leq 1$ and $e^{-x} \leq 1$ for nonnegative $x$ gives the desired result. 
\end{proof}

\begin{proof}[Proof of Lemma \ref{lm:opt}]
Since $\OPT_n(F) = \sum_{i=1}^n \mathbb{E} \left[\min_{l = 1, \ldots, i} \{X_l\} \right]$, we have 
\begin{align*}
\OPT_n(F) = \sum_{i=1}^n \int_0^{\infty} x \cdot i(1 - F(x))^{i-1} \cdot f(x) \, \mathrm{d}x = \int_0^1 F^{-1}(u) \cdot \sum_{i=1}^n i(1-u)^{i-1} \, \mathrm{d}u,
\end{align*}
where in the first equality we computed the distribution of the first order statistic, and in the second equality we introduced the change of variable $u = F(x)$ and brought the sum inside the integral. Since $F^{-1}(u)$ is non-decreasing and differentiable in $u$, it can be written in terms of the integral $\int_0^u r(v) \, \mathrm{d}v$, for some nonnegative function $r$. By exchanging the order of integrals and performing the integration gives us the second characterization of $\OPT_n$.
\end{proof} 

\section{Missing Proofs from Section~\ref{section:lowerbound}} \label{app:lowerbound}

\begin{proof}[Proof of Lemma \ref{lm:lbweakdual}]
Multiplying Constraint~\eqref{constraint:LBD4} with $h(u)$ and integrating, we have 
\begin{align*}
\zeta \int_0^1 h(u) \cdot P_n'(1-u) \, \mathrm{d}u + \int_0^1 h(u) \cdot \frac{d \eta(u)}{du} \, \mathrm{d}u \geq \alpha_1 \int_0^1 h(u) \, \mathrm{d}u + \int_0^1 h(u) \cdot \left( \sum_{i=2}^n \int_u^1 \alpha_{i}(q) \, \mathrm{d}q \right) \, \mathrm{d}u.
\end{align*}
Applying Constraint~\eqref{constraint:LBP3} to the first term on the LHS and exchanging the order of integration in the RHS,  
\begin{align*}
\zeta + \int_0^1 h(u) \cdot \frac{d \eta(u)}{du} \, \mathrm{d}u &\geq \alpha_1 \int_0^1 h(u) \, \mathrm{d}u + \int_0^1 \sum_{i=2}^n \alpha_{i}(q) \int_0^q h(u) \, \mathrm{d}u \, \mathrm{d}q \\
&\geq \alpha_1 \cdot \hat{d}_1 + \int_0^1 \sum_{i=2}^n \alpha_{i}(q) \left(\hat{d}_i - (1-q)\hat{d}_{i-1} \right) \, \mathrm{d}q \geq \sum_{i=1}^n \hat{d}_i,
\end{align*}
where the second inequality applied Constraint~\eqref{constraint:LBP1}-\eqref{constraint:LBP2}, and the last inequality is obtained by regrouping the terms into $\left(\alpha_1 - \int_0^1 (1-q)\alpha_2(q) \, \mathrm{d}q \right) \hat{d}_1 + \int_0^1 \alpha_n(q) \, \mathrm{d}q \cdot \hat{d}_n + \sum_{i=2}^{n-1} \left(\int_0^1 \alpha_i(q) \, \mathrm{d}q - \int_0^1 (1-q)\alpha_{i+1}(q) \, \mathrm{d}q \right) \hat{d}_i$ and applying Constraint~\eqref{constraint:LBD1}-\eqref{constraint:LBD3}. Weak duality follows once we show $\int_0^1 h(u) \cdot \frac{d \eta(u)}{du} \, \mathrm{d}u \leq 0$. Indeed, 
\begin{align*}
\int_0^1 h(u) \cdot \frac{d \eta(u)}{du} \, \mathrm{d}u &= (h(1) \cdot \eta(1) - h(0) \cdot \eta(0)) - \int_0^1 \eta(u) \cdot \frac{d h(u)}{du} \, \mathrm{d}u \leq - \int_0^1 \eta(u) \cdot \frac{d h(u)}{du} \, \mathrm{d}u \leq 0,
\end{align*}
where the first inequality applied Constraint~\eqref{constraint:LBD5} and the nonnegativity of $h$, and the last inequality applied the nonnegativity of both $\eta$ and $\frac{d h(u)}{du}$ since $h$ is nondecreasing. This concludes the proof. 
\end{proof}

\begin{proof}[Proof of Lemma \ref{lm:lbgammaunique}]
Note that to have $\varepsilon_n = 0$, we can also solve for $\varepsilon_j$ backwards using \eqref{eq:lbrecurrence} starting from $\varepsilon_n = 0$, and seek $\zeta$ that eventually leads to $\varepsilon_1 = 1$. The proof consists of three parts. We first prove the differentiability and monotonicity of each $\varepsilon_j$ for $j = n-1, \ldots, 1$ with respect to $\zeta$. We then show that $\zeta > 1$ is a necessary condition to ensure the desired properties in $\{\varepsilon_j\}_{j=1}^n$. Finally, we show that $\{\varepsilon_j\}_{j=1}^n$ is strictly decreasing when $\zeta > 1$, and monotonicity in $\zeta$ allows us to conclude the uniqueness of $\zeta$ that achieves $\varepsilon_1 = 1$.
\paragraph{Differentiability and Monotonicity in $\zeta$:} \eqref{eq:lbrecurrence} can be rearranged into 
\begin{align}
P_n'(1-\varepsilon_{j}) - \frac{j}{\zeta} = (1-\varepsilon_{j+1}) P_n'(1-\varepsilon_{j+1}) - P_n(1-\varepsilon_{j+1}). \label{eq:lbrecurrence2}
\end{align}
First, note that the LHS of \eqref{eq:lbrecurrence2} is continuous and strictly decreasing in $\varepsilon_j$, therefore it is invertible with a differentiable inverse, thus all $\frac{\partial \varepsilon_j}{\partial \zeta}$ are well-defined for $\varepsilon_j \in [0, 1]$. By induction, we show that $\frac{\partial \varepsilon_j}{\partial \zeta} > 0$ for all $j = n-1, \ldots, 1$. At $j = n-1$, with $\varepsilon_n = 0$, the RHS of \eqref{eq:lbrecurrence2} evaluates to $P_n'(1) - P_n(1) = \frac{n(n+1)}{2} - n = \frac{n(n-1)}{2}$. Taking derivative of \eqref{eq:lbrecurrence2} with respect to $\zeta$, we get
\begin{align*}
-P_n''(1-\varepsilon_{n-1}) \cdot \frac{\partial \varepsilon_{n-1}}{\partial \zeta} + \frac{n-1}{\zeta^2} = 0,
\end{align*}
which gives $\frac{\partial \varepsilon_{n-1}}{\partial \zeta} = \frac{n-1}{\zeta^2 P_n''(1-\varepsilon_{n-1})} > 0$ because $P_n''(t) = \sum_{i=1}^n i(i-1) t^{i-2} > 0$ for all $t \in [0, 1]$. Now suppose $\frac{\partial \varepsilon_j}{\partial \zeta} > 0$ for all $j = n-1, \ldots, k+1$ for some $k \geq 1$. At $j = k$, taking derivative of \eqref{eq:lbrecurrence2} with respect to $\zeta$,
\begin{align*}
&\qquad -P_n''(1-\varepsilon_{k}) \cdot \frac{\partial \varepsilon_{k}}{\partial \zeta} + \frac{k}{\zeta^2} = -(1-\varepsilon_{k+1}) P_n''(1-\varepsilon_{k+1}) \cdot \frac{\partial \varepsilon_{k+1}}{\partial \zeta}, 
\end{align*}
so it follows from the induction hypothesis $\frac{\partial \varepsilon_{k+1}}{\partial \zeta} > 0$ that $\frac{\partial \varepsilon_{k}}{\partial \zeta} > 0$. Therefore, $\varepsilon_j$ for all $j = n-1, \ldots, 1$ are strictly increasing in $\zeta$.
\paragraph{$\zeta > 1$ is necessary:} We now show that to guarantee $\{\varepsilon_j\}_{j=1}^n$ strictly decreasing in $j$ with $\varepsilon_1 = 1$, a necessary condition is $\zeta > 1$. First note that in the RHS of \eqref{eq:lbrecurrence2}, $(1-\varepsilon) P_n'(1-\varepsilon) - P_n(1-\varepsilon) = \sum_{i=1}^n i(1-\varepsilon)^i$ is decreasing in $\varepsilon$. Therefore, for $\varepsilon_{j-1} > \varepsilon_j$, we need the RHS at $j-1$ to be smaller than the RHS at $j$: 
\begin{align*}
(1-\varepsilon_{j}) P_n'(1-\varepsilon_{j}) - P_n(1-\varepsilon_{j}) < (1-\varepsilon_{j+1}) P_n'(1-\varepsilon_{j+1}) - P_n(1-\varepsilon_{j+1}) = P_n'(1-\varepsilon_{j}) - \frac{j}{\zeta}, 
\end{align*}
where the equality is simply \eqref{eq:lbrecurrence2} at $j$. At $j = 1$, this becomes
\begin{align*}
(1-\varepsilon_{1}) P_n'(1-\varepsilon_{1}) - P_n(1-\varepsilon_{1}) < P_n'(1-\varepsilon_{1}) - \frac{1}{\zeta}.
\end{align*}
Substituting $\varepsilon_1 = 1$ into the above gives $\zeta > 1$. Therefore, to have both $\varepsilon_{j}$ decreasing in $j$ and $\varepsilon_1 = 1$, $\zeta$ has to be greater than $1$. 
\paragraph{Monotonicity of $\varepsilon_j$:} We now show by induction that with $\zeta > 1$, the resulting $\{\varepsilon_j\}_{j=1}^n$ is strictly decreasing in $j$. Recall that given $\varepsilon_{j+1}$, $\varepsilon_j$ can be solved using \eqref{eq:lbrecurrence2}, so at $j = n-1$, we solve the equation 
\begin{align*}
P_n'(1-\varepsilon_{n-1}) - \frac{n-1}{\zeta} = \frac{n(n-1)}{2}.
\end{align*}
To enforce $\varepsilon_{n-1} > \varepsilon_{n}$, we require $P_n'(1-\varepsilon_{n-1}) < P_n'(1-\varepsilon_{n}) = \frac{n(n+1)}{2}$. Plugging this into the above equation, we have $\frac{n(n+1)}{2} - \frac{n-1}{\zeta} > \frac{n(n-1)}{2}$, which implies having $\zeta > \frac{n-1}{n}$ ensures $\varepsilon_{n-1} > \varepsilon_{n}$. This is directly satisfied by $\zeta > 1$. Now suppose we already have $\varepsilon_{k+1} > \varepsilon_{k+2} > \ldots > \varepsilon_n$ for some $k \geq 1$. Then, since both $\varepsilon_{k}$ and $\varepsilon_{k+1}$ has to satisfy equation \eqref{eq:lbrecurrence2}, we have 
\begin{align*}
P_n'(1-\varepsilon_{k}) &= (1-\varepsilon_{k+1}) P_n'(1-\varepsilon_{k+1}) - P_n(1-\varepsilon_{k+1}) + \frac{k}{\zeta} \\
&< (1-\varepsilon_{k+2}) P_n'(1-\varepsilon_{k+2}) - P_n(1-\varepsilon_{k+2}) + \frac{k+1}{\zeta} = P_n'(1-\varepsilon_{k+1}),
\end{align*}
where the inequality used the induction hypothesis $\varepsilon_{k+1} > \varepsilon_{k+2}$ and the RHS of \eqref{eq:lbrecurrence2} is decreasing in $\varepsilon$. Since $P_n'(1-\varepsilon)$ is decreasing in $\varepsilon$, the above implies $\varepsilon_k > \varepsilon_{k+1}$. Following the induction, we obtain $\varepsilon_{1} > \varepsilon_{2} > \ldots > \varepsilon_n = 0$ for $\zeta > 1$. 

Given the above properties, nonnegativity of $\{\varepsilon_j\}_{j=1}^n$ follows immediately from the monotonicity and $\varepsilon_n = 0$. Since for $j < n$, all $\varepsilon_j$ are strictly increasing in $\zeta$, we can tune $\zeta$ by increasing its value until the resulting $\{\varepsilon_j\}_{j=1}^n$ has exactly $\varepsilon_1 = 1$. Uniqueness of this value for $\zeta$ is ensured by the monotonicity of $\varepsilon_j(\zeta)$. This concludes the proof. 
\end{proof}

\begin{proof}[Proof of Lemma \ref{lm:lbprimalsol}]
Recall $h(u) = \sum_{i=1}^{n-1} d_i \cdot \mathbbm{1}_{[\varepsilon_{i+1}, \varepsilon_{i})}(u) + \Delta_2 \cdot \delta_{\{1\}}$. We first verify the nonnegativity of the solution. The System~\eqref{eq:lbsimuleq1}-\eqref{eq:lbsimuleq2} has the unique solution $d_{n-1} = \left(\varepsilon_n \cdot \frac{1+S_1}{S_2} + P_n(1-\varepsilon_n) \right)^{-1}$ and $\Delta_2 = \frac{\varepsilon_n}{S_2} \cdot d_{n-1}$. Recall that both $S_1 = \sum_{i=2}^{n-1} P_n(1-\varepsilon_i) \prod_{k=2}^{i-1} (1-\varepsilon_k)$ and $S_2 = \prod_{k=2}^{n-1} (1-\varepsilon_k)$ are positive and $\varepsilon_n \in [0, 1)$, so both $\Delta_2$ and $d_{n-1}$ in the solution are positive. Then since $d_1 = d_{n-1} + \Delta_2 \sum_{i=1}^{n-2}\prod_{k=2}^i (1 - \varepsilon_k) > d_{n-1} + \Delta_2$, we have both $d_1 > 0$ and $d_2 = d_1 - \Delta_2 > 0$. Recall from the construction of $\{d_i\}_{i=1}^n$ that $d_{i+1} = d_i - \Delta_2 \prod_{k=2}^i (1 - \varepsilon_k)$. With this at $i = n-1$, \eqref{eq:lbsimuleq2} becomes $d_{n-1} - d_n - \varepsilon_n \cdot d_{n-1} = 0$, which gives the identity $d_n = (1-\varepsilon_n) d_{n-1} > 0$. This construction also ensures that $d_i$ is non-increasing in $i$, so $d_i \geq d_n > 0$ for all $i = 3, \ldots, n-2$. Therefore, both $\{d_i\}_{i=1}^n$ and $h$ are nonnegative. Moreover, it follows from the construction that $h$ is non-decreasing. 

Now inductively, we show that for each $i$, the solution $\hat{d}_i = d_i$ tightens Constraint~\eqref{constraint:LBP1}-\eqref{constraint:LBP2}. For $i = 1$,  
\begin{align*}
\hat{d}_1 = \int_0^1 h(u) \, \mathrm{d}u = \sum_{i=1}^{n-1} d_i(\varepsilon_i - \varepsilon_{i+1}) + \Delta_2 = d_1 - \sum_{i=2}^{n-1} \Delta_i \cdot \varepsilon_i - \varepsilon_n \cdot d_{n-1} + \Delta_2 = d_1,
\end{align*}
where the second last equality regrouped the terms and used $\varepsilon_1 = 1$, and the last equality evaluated the sum: from $\frac{\Delta_{i+1}}{\Delta_i} = 1 - \varepsilon_i$ we get $\varepsilon_i = \frac{\Delta_i - \Delta_{i+1}}{\Delta_i}$, so $\sum_{i=2}^{n-1} \Delta_i\varepsilon_i = \sum_{i=2}^{n-1} (\Delta_i - \Delta_{i+1}) = \Delta_2 - \Delta_n$, while $\Delta_n = \varepsilon_n \cdot d_{n-1}$ follows from $d_n = (1-\varepsilon_n) d_{n-1}$. Now suppose $\hat{d}_i = d_i$ for $i = 1, \ldots, k$ for some $k < n$. For $i = k+1$, first note that \eqref{constraint:LBP2} holds for all $q \in [0, 1]$, therefore the constraint is tightened at $q$ that minimizes its RHS. Solving the first order condition $h(q) - d_k = 0$, we obtain $h(q) = d_k$, which implies that $q \in [\varepsilon_{k+1}, \varepsilon_k)$ minimizes the RHS, thus tightens the constraint. Therefore, $\hat{d}_{k+1} = \int_0^q h(u) \, \mathrm{d}u + (1-q)\cdot \hat{d}_{k}$ for $q \in [\varepsilon_{k+1}, \varepsilon_k)$, so
\begin{align*}
\hat{d}_{k+1} &= \int_0^{\varepsilon_{k+1}} h(u) \, \mathrm{d}u + \int_{\varepsilon_{k+1}}^q h(u) \, \mathrm{d}u + (1-q)d_{k} = \sum_{i=k+1}^{n-1} d_i(\varepsilon_i - \varepsilon_{i+1}) + (1 - \varepsilon_{k+1})d_k \\
&= d_k - \sum_{i=k+1}^{n-1} \Delta_i \cdot \varepsilon_i - \varepsilon_n \cdot d_{n-1} = d_k - (\Delta_{k+1} - \Delta_n) - \Delta_n = d_{k+1},
\end{align*}
where the first equality splits the integral into the interval $[0, \varepsilon_{k+1})$ and $[\varepsilon_{k+1}, q]$ and applied the induction hypothesis $\hat{d}_k = d_k$, the third equality regrouped terms, and the second last equality evaluated the sum using the same calculation as in the base case $i=1$. Finally for $i = n$, 
\begin{align*}
\hat{d}_n = \int_{0}^{\varepsilon_n} h(u) \, \mathrm{d}u + (1-\varepsilon_n)\hat{d}_{n-1} = (1-\varepsilon_n)\hat{d}_{n-1} = (1-\varepsilon_n)d_{n-1} = d_n 
\end{align*}
since $h(u) = 0$ on $[0, \varepsilon_n)$. It follows that $\hat{d}_i = d_i$ for $i = 1, \ldots, n$, which proves the feasibility and tightness of $\{d_i\}_{i=1}^{n-1}$ to Constraint~\eqref{constraint:LBP1}-\eqref{constraint:LBP2}.

For the feasibility of Constraint~\eqref{constraint:LBP3}, evaluating the LHS of the Constraint~\eqref{constraint:LBP3} with the proposed $h$, 
\begin{align*}
\int_0^1 h(u) \cdot P_n'(1-u) \, \mathrm{d}u &= \sum_{i=1}^{n-1} d_i(P_n(1-\varepsilon_{i+1}) - P_n(1-\varepsilon_{i})) + \Delta_2 \\
&= \sum_{i=2}^{n-1} \Delta_2 \prod_{k=2}^{i-1}(1-\varepsilon_k) \cdot P_n(1-\varepsilon_{i}) + d_{n-1} \cdot P_n(1 - \varepsilon_n) + \Delta_2,
\end{align*}
where the last equality regrouped terms and applied $P_n(1 - \varepsilon_1) = 0$ and $\Delta_{i+1} = \Delta_2 \prod_{k=2}^i (1 - \varepsilon_k)$ from the construction. To satisfy Constraint~\eqref{constraint:LBP3}, we equate the above to $1$. But this is exactly equation \eqref{eq:lbsimuleq1} after rearranging. Therefore, $(\{d_i\}_{i=1}^n, h)$ satisfies Constraint~\eqref{constraint:LBP3}. 
\end{proof}

\section{Asymptotic Analysis of Theorem~\ref{thm:lowerbound} (Proof of Proposition~\ref{prop:lowerbound})} \label{app:asymptotic}

We now present the derivation of the differential equation system that yields the asymptotic competitive ratio. To study what happens as $n \rightarrow \infty$, we let $\varepsilon_j = \frac{\lambda_{j}}{n}$ and multiply \eqref{eq:lbrecurrence} by $\frac{1}{n}$: 
\begin{align*}
&\underbrace{\frac{1}{n}\left[P_n'\left( 1 - \frac{\lambda_{j+1}}{n}\right) - P_n'\left( 1 - \frac{\lambda_{j}}{n}\right) \right]}_{A} = \underbrace{\frac{1}{n} \cdot \frac{\lambda_{j+1}}{n}P_n'\left( 1 - \frac{\lambda_{j+1}}{n}\right)}_B + \underbrace{\frac{1}{n} \cdot P_n\left( 1 - \frac{\lambda_{j+1}}{n}\right)}_C - \frac{j/n}{\zeta}.
\end{align*} 
Now, suppose $\lim_{n \rightarrow \infty} \frac{j}{n} = x \in (0, 1)$. Also suppose that the piecewise function obtained by joining all $\lambda_j$ converges: for each $n$, define the piecewise linear function $\lambda_n: [0, 1] \rightarrow \mathbb{R}$ to be linear between $\lambda_n \left( \frac{j-1}{n}\right)$ and $\lambda_n \left( \frac{j}{n}\right)$ with $\lambda_n \left( \frac{j}{n}\right) = \lambda_j$, and let $\lambda(x) = \lim_{n \rightarrow \infty} \lambda_n \left( \frac{j}{n}\right)$. Now we analyze the behavior for each of $A$, $B$, and $C$ as $n \rightarrow \infty$. In what follows, note that $P_n'(t)$ has the closed form $\frac{1 - t^n (1 + n(1-t))}{(1-t)^2}$, and $P_n(t)$ has the closed form $\frac{t(1-t^n)}{1-t}$. With $n \rightarrow \infty$, for $A$ we have 
\begin{align*}
\frac{1}{n}\left[P_n'\left( 1 - \frac{\lambda_{j+1}}{n}\right) - P_n'\left( 1 - \frac{\lambda_{j}}{n}\right) \right] &= \frac{1}{1/n} \cdot \left( \frac{1 - \left( 1 - \frac{\lambda_{j+1}}{n}\right)^n (1 + \lambda_{j+1})}{(\lambda_{j+1})^2} - \frac{1 - \left( 1 - \frac{\lambda_{j}}{n}\right)^n (1 + \lambda_{j})}{(\lambda_{j})^2}\right) \\
&\rightarrow \left( \frac{1 - e^{-\lambda(x)}(1 + \lambda(x))}{(\lambda(x))^2}\right)' \\*
&= - \lambda'(x) \left(-\frac{e^{-\lambda(x)}}{\lambda(x)} -\frac{2e^{-\lambda(x)}}{(\lambda(x))^2} + \frac{2}{(\lambda(x))^3}\cdot \left(1 - e^{\lambda(x)}\right) \right),
\end{align*}
For $B$ we have 
\begin{align*}
\frac{1}{n} \cdot \frac{\lambda_{j+1}}{n}P_n'\left( 1 - \frac{\lambda_{j+1}}{n}\right) &= \frac{1 - \left( 1 - \frac{\lambda_{j+1}}{n}\right)^n (1 + \lambda_{j+1})}{\lambda_{j+1}} \rightarrow \frac{1 - e^{-\lambda(x)}(1 + \lambda(x))}{\lambda(x)} = \frac{1 - e^{-\lambda(x)}}{\lambda(x)} - e^{-\lambda(x)}.
\end{align*}
For $C$ we have  
\begin{align*}
\frac{1}{n} \cdot P_n\left( 1 - \frac{\lambda_{j+1}}{n}\right) &= \left( 1 - \frac{\lambda_{j+1}}{n}\right) \cdot \frac{1 - \left( 1 - \frac{\lambda_{j+1}}{n}\right)^n}{\lambda_{j+1}} \rightarrow \frac{1 - e^{-\lambda(x)}}{\lambda(x)}.
\end{align*}
Lastly, in the limit, $\frac{j/n}{\zeta}$ converges to $\frac{x}{\zeta}$. Combining everything, 
\begin{align*}
\lambda'(x) \left(\frac{e^{-\lambda(x)}}{\lambda(x)} + \frac{2e^{-\lambda(x)}}{(\lambda(x))^2} - \frac{2}{(\lambda(x))^3}\cdot \left(1 - e^{\lambda(x)}\right) \right) = 2 \cdot \frac{1 - e^{-\lambda(x)}}{\lambda(x)} - e^{-\lambda(x)} - \frac{x}{\zeta}.
\end{align*}
Assuming $\varepsilon_n$ approaches $0$ as $n \rightarrow \infty$, the above equation has boundary conditions $\lambda(0) = +\infty$, $\lim_{x \uparrow 1}\lambda(x) = 0$. Applying the change of variable $y = e^{-\lambda(x)}$ gives the System~\eqref{eq:ode11}-\eqref{eq:ode22}. The following result shows that there exists a unique $\zeta$ such that this system has a solution. 
\begin{lemma} \label{lm:lbode}
There exists a unique $\zeta^* > 1$ such that the System~\eqref{eq:ode11}-\eqref{eq:ode22} has a unique solution.  
\end{lemma}
\begin{proof}
We first show the existence of a unique solution to the initial value problem \eqref{eq:ode11} with $y(0) = 0$ on $y < 1$. Let $F(y, x, \zeta) = \left( \frac{1}{\ln(y)} - \frac{2}{(\ln(y))^2} - \frac{2(1-y)}{y (\ln(y))^3}\right)^{-1} \cdot \left( - \frac{2(1-y)}{\ln(y)} - y - \frac{x}{\zeta} \right)$. Define $D = \{(x, y): 0 < x < 1, 0 < y < 1\}$. For any fixed $y$ and $\zeta$, $F$ is affine in $x$, thus $F$ is continuous in $x$ on $D$. Take any $y_0, y_1, x_0, x_1$ between $0$ and $1$, and define the rectangle $R = [x_0, x_1] \times [y_0, y_1]$, so that $R \subseteq D$. Consider $(x, y) \in R$ with $x$ fixed. Then note that $\ln(y)$ is continuously differentiable on $[y_0, y_1]$, and $F(y)$ is a rational combination of continuously differentiable functions, therefore $F$ is continuously differentiable in $y$. Since $R$ is a compact set, $\frac{\partial F}{\partial y}$ is bounded on $R$. Let $M = \sup_{(x, y) \in R} \left| \frac{\partial F}{\partial y} \right|$, then $M < \infty$, and it follows from the Mean Value Theorem that for any $y, \hat{y} \in [y_0, y_1]$, $|F(y) - F(\hat{y})| \leq M |y - \hat{y}|$. Therefore, $F$ is locally Lipschitz continuous in $y$ on $R$. By Picard–Lindelöf theorem (see~\cite{Kolmogorov1975analysis}, Section 8.2), the differential equation \eqref{eq:ode11} with the initial condition $y(\tilde{x}) = \tilde{y}$ where $(\tilde{x}, \tilde{y}) \in R$ has a unique solution on $[\tilde{x} - \epsilon, \tilde{x} + \epsilon]$ for some $\varepsilon > 0$. Now for $\delta > 0$ close to $0$, let $x_{\delta} > 0$ be some $x$ that is very close to $0$. Since $R$ can be taken as any closed subset of $D$, it follows from continuity that the differential equation \eqref{eq:ode11} with the initial condition $y(x_{\delta}) = \delta$ has a unique solution on $(0, 1)$. Taking both $\delta \downarrow 0$ and $x_{\delta} \rightarrow 0^+$, by continuity, \eqref{eq:ode11} with the initial condition $y(0) = 0$ has a unique solution on $[0, 1)$. 

Now recall that $y = e^{-\lambda(x)}$, where $\lambda(x)$ is the limit of $\lambda_j$, which is defined via $\varepsilon_j$, and in Proposition \ref{prop:lbgammaunique} we have shown that $\varepsilon_j$ is strictly monotone in $\zeta$. Then since $y$ is a composition of two monotonic functions, $y$ is monotone in $\zeta$. The monotonicity implies the uniqueness of $\zeta$ such that the resulting unique solution $y$ to \eqref{eq:ode11} with the initial value condition $y(0) = 0$ satisfies $\lim_{x \uparrow 1} y(x) = 1$. 
\end{proof}
We now verify that at $\zeta = \zeta_n$, $\varepsilon_n$ indeed approaches $0$ in the limit. 
\begin{lemma} \label{lm:lbasymp}
At $\zeta = \zeta_n$, $\varepsilon_n$ in the solution to \eqref{eq:lbrecurrence} approaches $0$ as $n \rightarrow \infty$.
\end{lemma}
\begin{proof}
By contradiction, suppose $\varepsilon_n > \delta > 0$ for all $n$. Then for $\zeta_n$ defined in Theorem~\ref{thm:lowerbound}, 
\begin{align*}
\zeta_n = \frac{n}{\varepsilon_n P_n'(1-\varepsilon_n) + P_n(1-\varepsilon_n)} > \frac{n}{\delta P_n'(1-\delta) + P_n(1-\delta)}
\end{align*}
since $\varepsilon P_n'(1-\varepsilon) + P_n(1-\varepsilon)$ is decreasing in $\varepsilon \geq 0$. Taking $n \rightarrow \infty$, using the closed form of the sums in $P_n(t)$ and $P_n'(t)$, the denominator becomes $\frac{\delta}{\delta^2} + \frac{1-\delta}{\delta} = \frac{2}{\delta}-1 < \infty$, so $\zeta_n$ approaches $\infty$. However, recall from the proof of Proposition \ref{prop:lbgammaunique} that $\varepsilon_n$ is decreasing in $\zeta$, so we should have $\zeta_n < \zeta^*$ for $\zeta^*$ stated in Lemma \ref{lm:lbode} which gives the solution with $\varepsilon_n$ approaching $0$ in the limit. This contradicts with the result that $\zeta^* < \infty$. Therefore, we must have $\varepsilon_n$ approaching $0$ in the limit with $\zeta_n$. 
\end{proof}
This implies that as $n \rightarrow \infty$, $\zeta_n$ coincides with the unique $\zeta^*$ in Lemma \ref{lm:lbode}, so Proposition \ref{prop:lowerbound} follows.

\section{Missing proofs from Section~\ref{section:alg} and Section~\ref{section:unif}} \label{app:unif}
\subsection{Missing proofs from Section~\ref{section:alg}} \label{appsub:alg}
\begin{proof}[Proof of Claim \ref{claim:kexpcost}]
Using the same change of variable $F^{-1}(u) = \int_0^u r(v) \, \mathrm{d}v$ as in Lemma \ref{lm:opt}: 
\begin{align*}
\mathbb{E}\left[ X_i \,\middle|\, F(X_i) \in \left( q_{l+1}, q_l\right]\right] = \frac{1}{q_l - q_{l+1}} \int_{q_{l+1}}^{q_l} F^{-1}(u) \, \mathrm{d}u &= \frac{1}{q_l - q_{l+1}} \int_{q_{l+1}}^{q_l} \int_0^u r(v) \, \mathrm{d}v \, \mathrm{d}u \\
&= \frac{1}{q_l - q_{l+1}} \int_0^{q_l} r(v) \cdot \int_{\max\{q_{l+1}, v\}}^{q_l} 1 \, \mathrm{d}u \, \mathrm{d}v\\
&= \frac{1}{q_l - q_{l+1}} \int_0^{q_l} r(v) \cdot \min\{q_l - q_{l+1}, q_l - v\} \, \mathrm{d}v,
\end{align*}
where in the second line we exchanged the order of integrals. 
\end{proof}

\subsection{Missing proofs from Section~\ref{section:unif}} \label{appsub:unif}

\begin{proof}[Proof of Lemma~\ref{lm:uniformkcost}]
Taking $r(v) = 1$ and substituting the inputs of $\text{ALG}$ into Lemma \ref{lm:hirecost}, $C(k)$ becomes
\begin{align*}
C(k) &= \frac{q^k}{\beta} \sum_{l=k}^j \int_0^{\beta / q^l} \min\left\{\frac{\beta(q-1)}{q^{l+1}}, \frac{\beta}{q^l} - v \right\} \, \mathrm{d}v \cdot \left\lfloor a \cdot q^l + l + \left(2 - \frac{a}{q}\right)\right\rfloor \\
&= \frac{\beta(q^2-1)}{2q^2} \sum_{l=k}^j  \frac{1}{q^{2l-k}}\cdot \left\lfloor a \cdot q^l + l + \left(2 - \frac{a}{q}\right)\right\rfloor,
\end{align*}
Similarly, the cost in state $j$ becomes 
\begin{align*}
C(j) &= \frac{q^j}{\beta} \cdot \int_0^{\beta / q^j} \left( \frac{\beta}{q^l} - v \right) \, \mathrm{d}v \cdot \left\lfloor a \cdot q^j + j + \left(2 - \frac{a}{q}\right)\right\rfloor = \frac{\beta}{2q^{j}} \cdot \left\lfloor a \cdot q^j + j + \left(2 - \frac{a}{q}\right)\right\rfloor. 
\end{align*}
\end{proof}
\begin{proof}[Proof of Lemma \ref{lm:costub}]
We upper bound $C(k)$ by removing the floor and extend the sum to infinity. For $k = 0, \ldots, j-1$, we have
\begin{align*}
C(k) &\leq \frac{\beta(q^2-1)}{2q^2} \sum_{l=k}^{\infty} \frac{1}{q^{2l-k}}\cdot \left( a \cdot q^l + l + \left(2 - \frac{a}{q}\right)\right) = \frac{a\beta(q^2-1)}{2q^2}\sum_{l=k}^{\infty} \frac{1}{q^{l-k}} + \frac{\beta(q^2-1)}{2q^2} \sum_{l=k}^{\infty} \frac{l + (2 - \frac{a}{q})}{q^{2l-k}}.
\end{align*}
Now replace the infinite sums with their respective closed form to obtain
\begin{align*}
C(k)&\leq \frac{a\beta(q+1)}{2q} + \frac{\beta(q^2-1)}{2q^2} \left( \left(2-\frac{a}{q} \right) \cdot \frac{q^2}{q^k(q^2-1)} + \frac{1}{q^k} \cdot \frac{q^2 \cdot k(q^2 - 1) + q^2}{(q^2 - 1)^2}\right), 
\end{align*}
which simplifies to $\tilde{C}(k)$. The closed form of the sum $\sum_{l=k}^{\infty} \frac{l}{q^{2l-k}} = q^k \sum_{l=k}^{\infty} \frac{l}{q^{2l}}$ is derived by taking $x = \frac{1}{q^2}$ in the identity $\sum_{l=k}^{\infty} l \cdot x^l = \frac{x^k(k(1-x) + x)}{(1-x)^2}$ for $0 < x < 1$, which is obtained by taking the derivative of $\sum_{l=k}^{\infty} x^l = \frac{x^k}{1-x}$ and multiply by $x$. Similarly, we obtain $C(j) \leq \frac{a\beta}{2} + \frac{\beta}{2q^j} \left( j + \left(2 - \frac{a}{q} \right)\right) = \tilde{C}(j)$. 
\end{proof} 

\begin{proof}[Proof of Lemma~\ref{lm:uniformdualsol}]
First note that an optimal solution to \ref{form:DLPunif} tightens all constraints. Suppose $\{\alpha_k\}_{k=0}^{j}$ tightens Constraint~\eqref{constraint:alg1D2}. Then for $k = 1, \ldots, j-2$, $\alpha_k$ satisfies the recurrence relation~\eqref{eq:dualrecursion}, which is simply a rearrangement of the following:
\begin{align*}
\alpha_k &= (1-\tilde{p}) \cdot \alpha_{k+1} + \frac{q-1}{q^k} \alpha_0 + \tilde{p} \cdot \sum_{l=1}^{k-1} \frac{q-1}{q^{k-l}} \alpha_l \\
&= (1-\tilde{p}) \cdot \alpha_{k+1} + q \left[\frac{q-1}{q^{k+1}} \alpha_0 + \tilde{p} \cdot \sum_{l=1}^{k} \frac{q-1}{q^{k+1-l}} \alpha_l + (1-\tilde{p})\alpha_{k+2} - (1-\tilde{p})\alpha_{k+2} - \tilde{p} \cdot \frac{q-1}{q}\alpha_{k} \right] \\
&= (1-\tilde{p}) \cdot \alpha_{k+1} + q[\alpha_{k+1} - (1-\tilde{p})\alpha_{k+2}] - \tilde{p} (q-1)\cdot \alpha_{k},
\end{align*}
where in the third equality, we identify that $\frac{q-1}{q^{k+1}} \alpha_0 + \tilde{p} \cdot \sum_{l=1}^{k} \frac{q-1}{q^{k+1-l}} \alpha_l + (1-\tilde{p})\alpha_{k+2} = \alpha_{k+1}$ from Constraint~\eqref{constraint:alg1D2}. Equation~\eqref{eq:dualrecursion} is a second order homogeneous recurrence relation, which has characteristic equation $q(1-\tilde{p}) \tilde{\lambda}^2 - (1 + q - \tilde{p})\tilde{\lambda} + (1 + (q-1)\tilde{p}) = 0$ with distinct roots $\tilde{\lambda}_1 = 0$ and $\tilde{\lambda}_2 = \frac{1 + (q-1)\tilde{p}}{q(1-\tilde{p})} = \lambda$. Therefore, the optimal solution $\alpha_k^*$ has the solution form $A + B \cdot \lambda^k$. We now derive the boundary conditions for obtaining $A$ and $B$. By tightening Constraint~\eqref{constraint:alg1D1} and \eqref{constraint:alg1D2}, $\alpha_1^*$ and $\alpha_2^*$ satisfy
\begin{align*}
&\alpha_1^* = (1-\tilde{p}) \alpha_2^* + \frac{q-1}{q} \cdot \alpha_0^* = (1-\tilde{p}) \alpha_2^* + \frac{q-1}{q}((1-\tilde{p})\alpha_1^*+1),
\end{align*}
which results in 
\begin{align*}
(1+(q-1)\tilde{p}) \alpha_1^* = q(1-\tilde{p})\alpha_2^* + q-1.
\end{align*}
Similarly, with the tightened Constraint~\eqref{constraint:alg1D2}, we have the following relation connecting $\alpha_{j-1}^*$ and $\alpha_j^*$,
\begin{align*}
\alpha_{j-1}^* = (1-\tilde{p}) \alpha_{j}^* + \frac{q-1}{q^{j-1}} \alpha_{0}^* + \tilde{p} \cdot \sum_{l=1}^{j-2} \frac{q-1}{q^{j-1-l}} \alpha_{l}^* &= (1-\tilde{p}) \alpha_{j}^* + q \left( \frac{q-1}{q^{j}} \alpha_{0}^* + \tilde{p} \cdot \sum_{l=1}^{j-1} \frac{q-1}{q^{j-l}} \alpha_{l}^* - \tilde{p} \cdot \frac{q-1}{q} \cdot \alpha_{j-1}^* \right) \\
&= (1-\tilde{p}) \alpha_{j}^* + q \cdot \alpha_{j}^*  - \tilde{p}(q-1) \alpha_{j-1}^*, 
\end{align*}
which results in 
\begin{align*}
(1+(q-1)\tilde{p})\alpha_{j-1}^* = (1+q-\tilde{p}) \alpha_j^*.
\end{align*}
Plugging in $\alpha_k^* = A + B \cdot \lambda^k$, we obtain the following system,
\begin{align*}
(1+(q-1)\tilde{p}) (A + B \cdot \lambda) - q(1-\tilde{p})(A + B \cdot \lambda^2) &= q-1, \\
(1+(q-1)\tilde{p})(A + B \cdot \lambda^{j-1}) - (1+q-\tilde{p})(A + B \cdot \lambda^j) &= 0.
\end{align*}
Solving this system yields the solution $A = \theta$, $B = -\theta \cdot \lambda^{-(j+1)}$. Therefore, $\alpha_k^* = \theta \cdot (1 - \lambda^{k-(j+1)})$. $\alpha_0^*$ can be solved by tightening Constraint~\eqref{constraint:alg1D1}, which after direct calculation and algebraic rearrangement yields $\alpha_0^* = \theta \cdot \left( \frac{q\tilde{p}}{q-1} - (1-\tilde{p})\lambda^{-j}\right)$. Note that this set of solutions are nonnegative: recall that $\beta > -\frac{q^2}{a(q-1)} \cdot \ln\left( \frac{q}{2q-1}\right)$, a check with direct substitution shows that this ensures $E > 0$ and $\lambda > 1$, together with $q > 1$, this guarantees nonnegativity of the solution. Finally, we verify that the proposed solution tightens Constraint~\eqref{constraint:alg1D3}. With the given $\alpha_0^*, \ldots, \alpha_{j-1}^*$, we have 
\begin{align*}
\text{RHS of \eqref{constraint:alg1D3}} &= \frac{q-1}{q^j} \theta \cdot \left( \frac{q\tilde{p}}{q-1} - (1-\tilde{p})\lambda^{-j}\right) + \tilde{p} \cdot \sum_{l=1}^{j-1} \frac{q-1}{q^{j-l}} \theta \cdot (1 - \lambda^{l-(j+1)}) \\
&= \frac{\theta \tilde{p}}{q^{j-1}} - \frac{\theta(q-1)(1-\tilde{p})}{(\lambda q)^j} + \frac{\theta \tilde{p}(q-1)}{q^j} \cdot \sum_{l=1}^{j-1} q^l - \frac{\theta \tilde{p}(q-1)}{q^j \cdot \lambda^{j+1}} \cdot \sum_{l=1}^{j-1} (\lambda q)^l \\
&= \frac{\theta \tilde{p}}{q^{j-1}} - \frac{\theta(q-1)(1-\tilde{p})}{(\lambda q)^j} + \frac{\theta \tilde{p} (q^{j-1}-1)}{q^{j-1}} - \frac{\theta \tilde{p}(q-1)((\lambda q)^{j-1} - 1)}{q^{j-1} \cdot \lambda^{j} (\lambda q - 1)} \\
&= \theta \tilde{p} \cdot \frac{(2q-1)\tilde{p} - (q-1)}{\tilde{p}(1+(q-1)\tilde{p})} = \theta(1 - \lambda^{-1}) = \text{LHS of \eqref{constraint:alg1D3}}.
\end{align*}
$\alpha_{j+1}^*$ is not involved in any other constraints and the objective, therefore we omit its derivation. Direct calculation using Constraint~\eqref{constraint:alg1D4} will yield $\alpha_{j+1}^* = 1$. This concludes the proof. 
\end{proof}

\begin{proof}[Proof of Claim \ref{claim:dmonotone}]
First, note that the last constraint forces $d_{j+1}^* = 0$. Moreover, the optimal solution $\{d_k^*\}_{k = 1}^{j+1}$ tightens all constraints in \ref{form:LPunif}: suppose there exists $k$ such that the constraint corresponding to $d_k$ is not tight; then we can increase $d_k$ to tighten this constraint and remain feasible, since $d_k$ only appears in the RHS of other constraints, which means no other constraints are violated. We prove the monotonicity of $\{d_k^*\}_{k = 1}^{j+1}$ by induction. By contradiction, suppose $d_0^* < d_1^*$, then 
\begin{align*}
d_0^* &= \frac{1-\beta^2}{2\beta} + C_0 + \sum_{l=1}^j \frac{q-1}{q^l}\tilde{d}_l^* > C_0 + \frac{q-1}{q} d_0^* + \sum_{l=1}^{j-1} \frac{q-1}{q^{l+1}}d_{1+l}^*, 
\end{align*}
which by using the assumption $q\cdot C_0 \geq C_{1}$ and rearranging, implies $d_0^* > C_1 + \sum_{l=1}^{j-1} \frac{q-1}{q^l}d_{1+l}^*$. Since $d_0^* < d_1^*$, this in turn gives $C_1 + \sum_{l=1}^{j-1} \frac{q-1}{q^l}d_{1+l}^* < d_1^*$. Now,
\begin{align*}
d_1^* &= (1 - p_1) \cdot d_{0}^* + p_1 \cdot \left(C_1 + \sum_{l=1}^{j-1} \frac{q-1}{q^l}d_{1+l}^* \right) < (1 - p_1) \cdot d_{1}^* + p_1 \cdot d_1^* = d_1^*,
\end{align*}
which is a contradiction. Thus $d_0^* \geq d_1^*$. 

Now suppose $d_{k-1}^* \geq d_k^*$ for all $k$ up to some $h \leq j$. By contradiction, suppose $d_h^* < d_{h+1}^*$. Then, 
\begin{align*}
d_h^* &= (1 - p_h) \cdot d_{h-1}^* + p_h \cdot \left( C_h + \sum_{l=1}^{j-h} \frac{q-1}{q^l}d_{h+l}^* \right) \\
&> (1 - p_h) \cdot d_{h}^* + p_h \cdot \left( C_h + \frac{q-1}{q} d_h^* + \sum_{l=1}^{j-(h+1)} \frac{q-1}{q^{l+1}}d_{h+1+l}^* \right),
\end{align*}
where in the second inequality we applied the induction hypothesis $d_{h-1}^* \geq d_h^*$ in the first term and the assumption $d_h^* < d_{h+1}^*$ in the second term. Similar to the case with $k=0$, it follows from $q\cdot C_h \geq C_{h+1}$ that $d_{h+1}^* > d_h^* \geq C_{h+1} + \sum_{l=1}^{j-(h+1)} \frac{q-1}{q^l}d_{h+1+l}^*$. Now,
\begin{align*}
d_{h+1}^* &= (1 - p_{h+1}) \cdot d_{h}^* + p_{h+1} \cdot \left( C_{h+1} + \sum_{l=1}^{j-(h+1)} \frac{q-1}{q^l}d_{h+1+l}^* \right) < (1 - p_{h+1}) \cdot d_{h+1}^* + p_{h+1} \cdot d_{h+1}^* = d_{h+1}^*, 
\end{align*}
which is a contradiction, so we must have $d_h^* \geq d_{h+1}^*$. By induction, $d_{k-1}^* \geq d_k^*$ holds for all $k = 0, \ldots, j$. Since $d_{j+1}^* = 0$, $d_j^* \geq d_{j+1}^*$ follows from the non-negativity. Therefore, $d_{k-1}^* \geq d_k^*$ for all $k = 0, \ldots, j+1$. 
\end{proof}

\begin{proof}[Proof of Claim \ref{claim:jOPT}]
From Remark \ref{lm:uniformopt}, $\OPT_n = \sum_{i=1}^n \frac{1}{i+1} = \mathcal{H}_{n+1} - 1$, where $\mathcal{H}_n$ denotes the sum of first $n$ terms in the Harmonic series. Combining $j \leq \log_q \frac{n}{a} + 1 = \frac{\ln n}{\ln q} - \frac{\ln a}{\ln q} + 1$ and $\ln n < \ln (n+1) = \mathcal{H}_{n+1} - \gamma - \text{error}_n = \OPT_n + 1 - \gamma - \text{error}_n$, where $\text{error}_n \leq \frac{1}{2n}$ (see \cite{Havil2003}).
\end{proof}

\section{Missing Proofs from Section~\ref{section:general}} \label{app:general}

\begin{proof}[Proof of Lemma~\ref{lm:weakdual}]
Multiplying Constraint~\eqref{constraint:alg2D5} by $r(v)$, then integrate and simplify,
\begin{align*}
\zeta \int_0^1 r(v) \cdot \sum_{i=1}^n (1-v)^i \, \mathrm{d}v &\geq \alpha_0 \left( \frac{1}{\beta} \int_0^1 r(v) \cdot \psi_{\text{fail}}(v) \, \mathrm{d}v + \tilde{C}(0)\right) + \tilde{p} \sum_{k=1}^{j} \alpha_k \cdot \tilde{C}(k).
\end{align*}
Applying Constraint~\eqref{constraint:alg2P5} to the LHS and Constraint~\eqref{constraint:alg2P1}-\eqref{constraint:alg2P3} to the RHS gives
\begin{align*}
\zeta &\geq \alpha_0 \left( d_0 - \sum_{l=1}^j \frac{q-1}{q^l}d_l - \frac{1}{q^j} d_{j+1} \right) + \sum_{k=1}^{j} \alpha_k \left( d_k - (1 - \tilde{p}) \cdot d_{k-1} - \tilde{p}\sum_{l=1}^{j-k} \frac{q-1}{q^l}d_{k+l} - \frac{\tilde{p}}{q^{j-k}} d_{j+1} \right) \\
&= d_0(\alpha_0 - (1 - \tilde{p}) \cdot \alpha_1 - 1) + \sum_{k=1}^{j-1} d_k \left( \alpha_k - (1 - \tilde{p}) \cdot \alpha_{k+1} - \frac{q-1}{q^k} \alpha_0 - \tilde{p} \cdot \sum_{l=1}^{k-1} \frac{q-1}{q^{k-l}} \alpha_l \right) \\
&\quad + d_j \left( \alpha_j - \frac{q-1}{q^j} \alpha_0 - \tilde{p} \cdot \sum_{l=1}^{j-1} \frac{q-1}{q^{j-l}} \alpha_l \right) + d_{j+1} \left( \alpha_{j+1} - \frac{1}{q^j} \alpha_0 - \tilde{p} \cdot \sum_{l=1}^{j} \frac{1}{q^{j-l}} \alpha_l \right) + d_0 - d_{j+1} \cdot \alpha_{j+1} \\
&\geq 0 + d_0 - 0 = d_0,
\end{align*}
where the equality is simply a regrouping of the terms and plus minus $d_0$ and $d_{j+1} \cdot \alpha_{j+1}$. Observe that rearranging Constraint~\eqref{constraint:alg2D1}-\eqref{constraint:alg2D4} into the form LHS $\geq 0$, the LHS are exactly the terms inside the brackets in the equality above, thus these terms simplify to the first zero in the second inequality. The second zero is due to $d_{j+1} \leq 0$ in Constraint~\eqref{constraint:alg2P4}. This shows the weak duality. 
\end{proof}

\begin{proof}[Proof of Lemma~\ref{lm:Nclosedform}]
Recall $\alpha_0 = \theta \cdot \frac{q\tilde{p}}{q-1}$, $\alpha_k = \theta$ for $k = 1, \ldots, j$, where $\theta = \frac{q-1}{E}$. Evaluating $N(v)$,  
\begin{align*}
N(v) &= \frac{\theta \cdot \frac{q\tilde{p}}{q-1}}{\beta} \left( \psi_{\text{fail}}(v) + \sum_{l=0}^j (aq^l + \rho_l)\psi_{l}(v) \cdot \mathbbm{1}_{[0, \beta / q^l]}(v)\right) + \frac{\theta\tilde{p}}{\beta} \sum_{k=1}^{j} q^k \left( \sum_{l=k}^j (aq^l + \rho_l)\psi_{l}(v) \cdot \mathbbm{1}_{[0, \beta / q^l]}(v)\right) \\
&= \frac{\theta}{\beta} \cdot \frac{q\tilde{p}}{q-1} \cdot \psi_{\text{fail}}(v) + \frac{\tilde{p}}{\beta E} \sum_{l=0}^j (aq^l + \rho_l)\psi_{l}(v) \cdot \mathbbm{1}_{[0, \beta / q^l]}(v) + \frac{\theta \tilde{p}}{\beta} \sum_{k=0}^{j} q^k \left( \sum_{l=k}^j (aq^l + \rho_l)\psi_{l}(v) \cdot \mathbbm{1}_{[0, \beta / q^l]}(v)\right),
\end{align*}
where we applied $\alpha_0 = \theta + \frac{\tilde{p}}{E}$, and regrouped terms so that the first summation runs from $k=0$. Swapping the order of the summation in the last term and simplifying,  
\begin{align*}
\frac{\theta \tilde{p}}{\beta} \sum_{k=0}^{j} q^k \left( \sum_{l=k}^j (aq^l + \rho_l)\psi_{l}(v) \cdot \mathbbm{1}_{[0, \beta / q^l]}(v)\right) &= \frac{\theta \tilde{p}}{\beta} \sum_{l=0}^j \left((aq^l + \rho_l)\psi_{l}(v) \cdot \mathbbm{1}_{[0, \beta / q^l]}(v) \right) \cdot \sum_{k=0}^l q^k \\*
&= \frac{\tilde{p}}{\beta E}\sum_{l=0}^j (q^{l+1}-1)(aq^l + \rho_l)\psi_{l}(v) \cdot \mathbbm{1}_{[0, \beta / q^l]}(v).
\end{align*}
Substituting back and simplifying, 
\begin{align}
N(v) &= \frac{a\tilde{p}}{\beta E} \left( \frac{q}{a} \cdot \psi_{\text{fail}}(v) + \sum_{l=0}^j q^{l+1}\left(q^l + \frac{1}{a}\cdot\rho_l \right) \psi_{l}(v) \cdot \mathbbm{1}_{[0, \beta / q^l]}(v)\right). \label{eq:Nrewrite}
\end{align}
We break the summation in \eqref{eq:Nrewrite} into terms involving $q^l$ and $\rho_l$, which results in two sums $\sum_{l=0}^j q^{2l+1} \cdot \psi_{l}(v) \cdot \mathbbm{1}_{[0, \beta / q^l]}(v)$ and $\sum_{l=0}^j q^{l+1}\cdot\rho_l \cdot \psi_{l}(v) \cdot \mathbbm{1}_{[0, \beta / q^l]}(v)$. For the first sum, observe that when $l$ is temporarily fixed and $v \in \left( \frac{\beta}{q^{l+1}}, \frac{\beta}{q^l}\right]$, $\mathbbm{1}_{[0, \beta / q^{l'}]}(v) = 0$ for $l' > l$. In the non-zero terms, $\frac{\beta}{q^{l'}} - v < \frac{\beta(q-1)}{q^{l'+1}}$ only when $l' = l$, so $\psi_{l}(v) = \frac{\beta(q-1)}{q^{l'+1}}$ for $l' < l$, and $\psi_{l}(v) = \frac{\beta}{q^l} - v$ for $l' = l$. Similarly, $\psi_{\text{fail}}(v) = 1 - \beta$. Adding the first term $\frac{q}{a} \cdot \psi_{\text{fail}}(v)$ to this sum, we obtain, upon simplification, that for $v \in \left( \frac{\beta}{q^{l+1}}, \frac{\beta}{q^l}\right]$, 
\begin{align*}
\frac{q}{a} \cdot \psi_{\text{fail}}(v) + \sum_{l'=0}^j q^{2l'+1} \cdot \psi_{l'}(v) \cdot \mathbbm{1}_{[0, \beta / q^{l'}]}(v) = \frac{q}{a} \cdot (1-\beta) + \beta(q^l - 1) + \beta q^{l+1} - q^{2l+1} \cdot v = A_l - B_l \cdot v = M_l(v). 
\end{align*}
For $v \in (\beta, 1]$, $\mathbbm{1}_{[0, \beta / q^{l}]}(v) = 0$ for all $l$, therefore the only non-zero term is $\frac{q}{a} \cdot \psi_{\text{fail}}(v)$, and in this case $\psi_{\text{fail}}(v) = 1 - v$, which yields $M_{\text{fail}}(v) = \frac{q}{a}(1-v)$. 

The sum $\sum_{l=0}^j q^{l+1} \cdot\rho_l \cdot \psi_{l}(v) \cdot \mathbbm{1}_{[0, \beta / q^l]}(v)$ can be evaluated and simplified by repeating the same steps to obtain $R_l(v)$, so we omit the calculations for brevity. Since $M_l$ and $R_l$ are defined on intervals $\left( \frac{\beta}{q^{l+1}}, \frac{\beta}{q^l}\right]$ that partition $[0, \beta]$ and $M_{\text{fail}}$ is defined over $(\beta, 1]$, summing across $l$ and substituting back into \eqref{eq:Nrewrite} recovers the expression for $N(v)$ in the lemma. 
\end{proof}

\begin{proof}[Proof of Lemma~\ref{lm:betastar}]
We first check $\beta^*$ is in $(0, 1)$. For $\beta^* > 0$, we require $2\sqrt{q} - 1 - \frac{q}{a} > 0$, so $a \geq \frac{q}{2\sqrt{q}-1}$, which is directly satisfied with $\frac{a}{q} \geq 2$ and $q > 1$. Now rewrite $\beta^*$ as $\beta^* = \frac{a}{q(a-1)}\left(2\sqrt{q} - 1 - q\right) + 1$, so for $\beta^* < 1$ we require $\frac{a(2\sqrt{q}-1-q)}{q(a-1)} < 0$, thus $\frac{-a(\sqrt{q}-1)^2}{q(a-1)} < 0$, which is directly satisfied with $a , q > 1$. 

Recall that $\frac{a}{q} \geq 2$, so $R_0(v) < 0$, thus, $\Gamma_0(v) \leq \frac{M_0(v) \cdot v}{(1-v)(1-(1-v)^n)}$. Let $f(v) =  \frac{M_0(v) \cdot v}{1-v}$, then $\frac{M_0(v) \cdot v}{(1-v)(1-(1-v)^n)} = f(v) + f(v) \cdot \frac{(1-v)^n}{1-(1-v)^n}$. For $v \in \left(\frac{\beta}{q}, \beta\right]$, $\frac{(1-v)^n}{1-(1-v)^n} \rightarrow 0$ as $n \rightarrow \infty$, therefore to upper bound $\Gamma_0(v)$, it suffices to consider the maximum of $f(v)$, which is given by taking $l=0$ in Claim \ref{claim:maxsmalll}. 
\begin{claim} \label{claim:maxsmalll}
For all nonnegative integer $l$ and $0 < \beta < 1$, $\frac{M_l(v) \cdot v}{1-v}$ has a unique interior maximum $\phi_l$ on $v \in [0, 1]$, where $\phi_l = B_l\left(1 - \sqrt{1 - \frac{A_l}{B_l}} \right)^2$.
\end{claim}
Now, since $\sup_{v \in \left( \beta / q, \beta\right]} \Gamma_0(v) \leq \max_{v \in \left( \beta / q, \beta\right]} f(v) \leq \max_{v \in [0, 1]} f(v)$ we have, 
\begin{align*}
\lim_{n\to \infty}\sup_{v \in \left( \beta / q, \beta\right]} \Gamma_0(v) \leq \lim_{n\to \infty}\max_{v \in [0, 1]} f(v) &= B_0\left(1 - \sqrt{1 - \frac{A_0}{B_0}} \right)^2 = q \left(1 - \sqrt{(1-\beta)\cdot \frac{a-1}{a}} \right)^2,
\end{align*}
which is at most $1$ when $\beta \leq \beta^*$ by routine calculation. 
\end{proof}

\begin{proof}[Proof of Lemma~\ref{lm:rhovsm}]
Taking the first derivative of $\varepsilon_l(v)$, $\varepsilon_l'(v) = \frac{\tilde{A}_l \cdot B_l - A_l \cdot \tilde{B}_l}{(M_l(v))^2}$. The numerator is independent of $v$, and it determines the sign of $\varepsilon_l'(v)$. This implies that $\varepsilon_l(v)$ is monotone on the interval $\left( \frac{\beta}{q^{l+1}}, \frac{\beta}{q^l}\right]$, so its maximum happens at the endpoints. Therefore it suffices to upper bound $\varepsilon_l(v)$ at the two endpoints. 

Case 1: $\varepsilon_l'(v) > 0$, $\varepsilon_l(v)$ increasing. The maximum happens at $v = \frac{\beta}{q^{l}}$, so for all $v \in \left( \frac{\beta}{q^{l+1}}, \frac{\beta}{q^l}\right]$, $\varepsilon_l(v) \leq \frac{1/a \cdot R_l\left(\beta / q^{l}\right)}{M_l\left(\beta / q^{l}\right)}$. Expanding $R_l$ and $M_l$ gives 
\begin{align*}
\varepsilon_l(v) \leq \frac{\beta(q-1)\left( \left( 2-\frac{a}{q} \right)l + \frac{(l-1)l}{2}\right)}{a\beta \left( q^l - \frac{q}{a} - 1\right) + q} &\leq \frac{\beta(q-1)\left( \left( 2-\frac{a}{q} \right)l + \frac{(l-1)l}{2}\right)}{a\beta \left( q^l - \frac{q}{a} - 1\right) + \beta q} \leq \frac{(q-1)(l-1)l}{2a(q^l - 1)} \leq \frac{(l-1)l}{2aq^{l-1}},
\end{align*}
where the second inequality used $\beta \leq 1$, the third inequality used $\frac{a}{q} \geq 2$ and the last inequality used $\frac{q-1}{q^x-1} \leq \frac{1}{q^{x-1}}$ for $x \geq 1$, which follows from $\frac{q-1}{q^x-1} \leq \frac{q - 1/q^{x-1}}{q^x-1} = \frac{(q^x-1) / q^{x-1}}{q^x-1} = \frac{1}{q^{x-1}}$.

Case 2: $\varepsilon_l'(v) < 0$, $\varepsilon_l(v)$ decreasing. The supremum happens at $v \rightarrow \frac{\beta}{q^{l+1}}$, so for all $v \in \left( \frac{\beta}{q^{l+1}}, \frac{\beta}{q^l}\right]$, $\varepsilon_l(v) \leq \lim_{v \rightarrow \beta / q^{l+1}} \frac{1/a \cdot R_l\left(v\right)}{M_l\left(v\right)}$. Expanding $R_l$ and $M_l$ gives 
\begin{align*}
\varepsilon_l(v) &\leq \frac{\beta(q-1)\left( \left( 3-{a}/{q} \right)l + {(l-1)l}/{2} + 2-{a}/{q}\right)}{a\beta \left( q^{l+1} - {q}/{a} - 1\right) + q} \\
&\leq \frac{\beta (q-1)l\left( 3 - {a}/{q} + {l-1}/{2}\right)}{a\beta \left( q^{l+1} - {q}/{a} - 1\right) + \beta q} \leq \frac{(q-1)l(l+1)}{2a(q^{l+1} - 1)} \leq \frac{l(l+1)}{2aq^l},
\end{align*}
where the second inequality used $\beta \leq 1$ and $\frac{a}{q} \geq 2$, the third inequality used $3 - \frac{1}{2} - \frac{a}{q} \leq \frac{5}{2} - 2 = \frac{1}{2}$, and the last inequality again applied $\frac{q-1}{q^x-1} \leq \frac{1}{q^{x-1}}$.

As a result, for all $l \geq 1$, $\varepsilon_l(v) \leq \max \left\{{(l-1)l} / {2aq^{l-1}}, {l(l+1)}/{2aq^{l}}\right\} = \tilde{\varepsilon}_l$.
\end{proof}

\begin{proof}[Proof of Lemma~\ref{lm:maxbound}]
Recall $\Gamma_l = \frac{M_l(v)\cdot v (1 + \varepsilon_l(v))}{(1-v)(1-(1-v)^n)}$. As discussed in the proof sketch, our task is to upper bound $\sup_{l \in \{0, \ldots, j/2 - 1\}} \sup_{v \in (\beta / {q^{l+1}}, \beta / q^l]} \Gamma_l(v)$ and $\sup_{l \in \{j/2, \ldots,j\}} \sup_{v \in (\beta/q^{l+1}, \beta/q^l]} \Gamma_l(v)$ in the limit $n \rightarrow \infty$. In what follows, recall $j = \left\lceil \log_q \frac{n}{a} \right\rceil$, so we upper bound $j$ by $\log_q \frac{n}{a} + 1$. 

Case 1: $l \leq \frac{j}{2} - 1$. To upper bound $\sup_{l \in \{0, \ldots, j/2 - 1\}} \sup_{v \in (\beta / {q^{l+1}}, \beta / q^l]} \Gamma_l$, we first focus on $\frac{M_l(v) \cdot v}{1-v}$. From Claim \ref{claim:maxsmalll}, $\sup_{v \in (\beta / {q^{l+1}}, \beta / q^l]} \frac{M_l(v) \cdot v}{1-v}$ is at most $\phi_l = B_l\left(1 - \sqrt{1 - \frac{A_l}{B_l}} \right)^2$. Let $z_l = \frac{\beta (1+q)}{q^{l+1}}$, then $\frac{A_l}{B_l} = \frac{\beta (1+q)}{q^{l+1}} - \frac{(q/a+1)\beta - q/a}{q^{2l+1}}$ approaches $z_l$ as $l$ increases, so for algebraic simplicity, we consider $B_l\left(1 - \sqrt{1 - z_l} \right)^2$ instead. Using $z_{l+1} = \frac{z_l}{q}$ and $B_{l+1} = \frac{B_l}{q^2}$, we have 
\begin{align*}
\frac{\phi_l}{\phi_{l+1}} = \frac{B_l\left(1 - \sqrt{1 - z_l} \right)^2}{B_{l+1}\left(1 - \sqrt{1 - z_{l+1}} \right)^2} = \frac{1}{q^2} \left(\frac{1 - \sqrt{1-z_l}}{1 - \sqrt{1-z_l/q}}\right)^2 \geq \frac{1}{q^2} \cdot q^2 = 1,
\end{align*}
where the inequality follows from Claim \ref{claim:abratiobound}.

\begin{claim} \label{claim:abratiobound}
$\frac{1 - \sqrt{1-x}}{1 - \sqrt{1-x/q}} \geq q$ for $x > 0$, $q > 1$.
\end{claim}

Since $\frac{\phi_l}{\phi_{l+1}} \geq 1$, $\phi_l$ is nonincreasing in $l$. Now, with $\beta \leq \beta^*$, the proof of Lemma \ref{lm:betastar} guarantees that $\phi_0 \leq 1$, so we have $\sup_{v \in (\beta / {q^{l+1}}, \beta / q^l]} \frac{M_l(v) \cdot v}{1-v} \leq \phi_l \leq \phi_0 \leq 1$. It follows that 
\begin{align*}
\sup_{l \in \{0, \ldots, j/2 - 1\}} \sup_{v \in (\beta / {q^{l+1}}, \beta / q^l]} \Gamma_l(v) &= \sup_{l \in \{0, \ldots, j/2 - 1\}} \sup_{v \in (\beta / {q^{l+1}}, \beta / q^l]} \frac{M_l(v)\cdot v (1 + \varepsilon_l(v))}{(1-v)(1-(1-v)^n)} \\*
&\leq \max_{l \in \{0, \ldots, j/2-1\}} \sup_{v \in (\beta / {q^{l+1}}, \beta / q^l]} \frac{1}{1-(1-v)^n} \cdot (1 + \varepsilon_l). 
\end{align*}
Note that in this case, $\frac{\beta}{q^{l+1}} \geq \frac{\beta}{q^{j/2}} \geq \frac{\beta \sqrt{a}}{\sqrt{qn}}$, so for any $v \in (\beta / {q^{l+1}}, \beta / q^l]$, when $n \rightarrow \infty$,
\begin{align*}
1-(1-v)^n \geq 1-\left(1-\frac{\beta}{q^{l+1}} \right)^n \geq 1 - e^{-\beta n / q^{l+1}} \geq 1 - e^{-\beta \sqrt{an} / \sqrt{q}} \rightarrow 1.
\end{align*}
Therefore, 
\begin{align*}
\lim_{n \rightarrow \infty} \sup_{l \in \{0, \ldots, j/2 - 1\}} \sup_{v \in (\beta / {q^{l+1}}, \beta / q^l]} \Gamma_l(v) \leq \lim_{n \rightarrow \infty} \max_{l \in \{0, \ldots, j/2-1\}} \sup_{v \in (\beta / {q^{l+1}}, \beta / q^l]} \frac{1}{1-(1-v)^n} \cdot (1 + \varepsilon_l) \leq 1 + \varepsilon^*. 
\end{align*}

Case 2: $l \geq \frac{j}{2}$. In this regime, when $n$ is large, $l$ will be large enough such that we can approximate $1 - v$ by $1$ and $1-(1-v)^n$ by $1-e^{-vn}$ with negligible error for $v \in \left( \frac{\beta}{q^{l+1}}, \frac{\beta}{q^l}\right]$. Also, by Lemma \ref{lm:rhovsm}, $\varepsilon_l(v)$ is close to $0$. Therefore, we can consider $\sup_{v \in (\beta/q^{l+1}, \beta/q^l]} \frac{M_l(v) \cdot v}{1 - e^{-vn}}$ in place of $\sup_{v \in (\beta/q^{l+1}, \beta/q^l]} \Gamma_l(v)$. Re-index $l$ as $j - k$ for $k = 0, \ldots, \frac{j}{2}$, with the change of variable $x = vn$ and $j \leq \log_q \frac{n}{a} + 1$, we have
\begin{align*}
M_l(v) \cdot v &=\left(\frac{\beta nq(1+q)}{aq^k} - \left( \frac{q}{a}+1\right)\beta + \frac{q}{a} - \frac{q^3 n^2}{a^2 q^{2k}} \cdot \frac{x}{n} \right)\cdot \frac{x}{n} \rightarrow \left(\frac{\beta q(1+q)}{aq^k} - \frac{q^3}{a^2 q^{2k}} \cdot x \right)x
\end{align*}
as $n \rightarrow \infty$. Therefore,  $\sup_{v \in (\beta/q^{l+1}, \beta/q^l]} \Gamma_l(v)$ can be reformulated as 
\begin{align} 
\sup_{x \in (\beta aq^{k-2}, \beta aq^{k-1}]} \frac{1}{1 - e^{-x}}\left(\frac{\beta q(1+q)}{aq^k} - \frac{q^3}{a^2 q^{2k}} \cdot x\right)x. \label{eq:largelsup}
\end{align}
Let $h_k(x) = \frac{\beta q(1+q)}{aq^k} - \frac{q^3}{a^2 q^{2k}} \cdot x$ and $I_k = (\beta aq^{k-2}, \beta aq^{k-1}]$. To understand the behavior of \eqref{eq:largelsup}, first note that for any $x \in I_k$, we can map it to $I_{k+1}$ by $qx \in (\beta aq^{k-1}, \beta aq^{k}] = I_{k+1}$, so $h_{k+1}(\tilde{x})$ on $\tilde{x} \in I_{k+1}$ can be reformulated on $x \in I_k$ as 
\begin{align*}
h_{k+1}(qx) &= \frac{\beta q(1+q)}{aq^{k+1}} - \frac{q^3}{a^2 q^{2k+2}} \cdot (qx) = \frac{1}{q}\left(\frac{\beta q(1+q)}{aq^k} - \frac{q^3}{a^2 q^{2k}} \cdot x \right) = \frac{1}{q} \cdot h_k(x).
\end{align*}
It follows for \eqref{eq:largelsup} on $I_{k+1}$, 
\begin{align*}
\sup_{\tilde{x} \in I_{k+1}} \frac{1}{1-e^{-\tilde{x}}} \cdot h_k(\tilde{x}) \cdot \tilde{x} &= \sup_{x \in I_k} \frac{1}{1-e^{-qx}} \cdot \frac{1}{q} \cdot h_k(x) \cdot (qx) \\
&= \sup_{x \in I_k} \frac{1}{1-e^{-qx}}\cdot h_k(x) \cdot x < \sup_{x \in I_k} \frac{1}{1-e^{-x}}\cdot h_k(x) \cdot x,
\end{align*}
where the last inequality is due to $e^{-qx} < e^{-x}$ since $q > 1$ and $x > 0$. It follows that \eqref{eq:largelsup} is decreasing in $k$, so the maximum is achieved at $k = 0$. We now have 
\begin{align*}
\lim_{n \rightarrow \infty} \sup_{l \in \{j/2, \ldots,j\}} \sup_{v \in (\beta/q^{l+1}, \beta/q^l]} \Gamma_l(v) &\leq \max_{x \in [0, \beta a / q]} \frac{1}{1-e^{-x}}\left(\frac{\beta q(1+q)}{a} - \frac{q^3}{a^2} \cdot x\right) x = \varphi_0^*. 
\end{align*} 
Combining the two cases give the desired result.
\end{proof} 

\begin{proof}[Proof of Claim \ref{claim:largejub}]
For the first part, it suffices to show $x+1 \geq \frac{x}{1-e^{-x}}$ on $x \geq 0$, which is equivalent to $\frac{1+x}{x} \geq \frac{1}{1-e^{-x}}$. From Taylor expansion, $e^x \geq 1+x$ for $x \geq 0$, so $e^{-x} \leq \frac{1}{1+x}$. It follows that $1 - e^{-x} \geq \frac{x}{1+x}$. Taking the reciprocal on both sides concludes the proof. 

For the second part, observe that $g_k''(x) = - \frac{2q^3}{a^2 q^{2k}} < 0$, so $g_k$ is strictly concave, thus its global maximizer is unique. Solve for the maximizer $x^*$ using the first order condition $\frac{\beta q(1+q)}{aq^k} - \frac{q^3}{a^2 q^{2k}} - \frac{2q^3}{a^2 q^{2k}} \cdot x = 0$, which yields $x^* = \frac{\beta a(1+q)q^k}{2q^2} - \frac{1}{2}$. By evaluating $g_k(x^*)$ we get $\max_x g_k(x) = \varphi_k$.
\end{proof}

\begin{proof}[Proof of Claim \ref{claim:maxsmalll}]
The second derivative of $\frac{M_l(v) \cdot v}{1-v}$ is $\frac{2(A_l - B_l)}{(1-v)^3}$. Note that $B_l - A_l = q^l(q^{l+1} - \beta(1+q)) + \left( \frac{q}{a} + 1\right)\beta - \frac{q}{a}$ is increasing in $l$, and $B_0 - A_0 = \left( q - \frac{q}{a}\right)(1-\beta) > 0$ for $q > 1$, $\frac{q}{a} \leq \frac{1}{2}$, and $\beta < 1$, so we have $A_l - B_l < 0$ for all $l$. It follows that $\frac{2(A_l - B_l)}{(1-v)^3} < 0$, thus $\frac{M_l(v) \cdot v}{1-v}$ is strictly concave on $v \in [0, 1]$. The first derivative of $\frac{M_l(v) \cdot v}{1-v}$ is $\frac{B_l \cdot v^2 - 2B_l \cdot v + A_l}{(1-v)^2}$. Equating this to $0$, the only solution inside $v \in [0, 1]$ is $v^* = 1 - \sqrt{1 - \frac{A_l}{B_l}}$. As a result, $\frac{M_l(v) \cdot v}{1-v}$ has a unique maximum at $v^*$, which evaluates to $B_l\left(1 - \sqrt{1 - \frac{A_l}{B_l}} \right)^2$. 
\end{proof}

\begin{proof}[Proof of Claim \ref{claim:abratiobound}]
Consider $\phi(x) = q\sqrt{1 - \frac{x}{q}} - \sqrt{1-x}$. $\phi$ is non-decreasing: $\phi'(x) = -\frac{1}{2\sqrt{1 - \frac{x}{q}}} + \frac{1}{2\sqrt{1-x}}$ is nonnegative because $\sqrt{1 - \frac{x}{q}} \geq \sqrt{1-x}$ for $x > 0, q > 1$. Therefore, $\phi(x) \geq \phi(0) = q - 1$ for all $x > 0$. Rearranging, this recovers $\frac{1 - \sqrt{1-x}}{1 - \sqrt{1-x/q}} \geq q$. 
\end{proof}

\end{document}